\documentclass{article}
\usepackage{ijcai26}

\usepackage{times}
\usepackage{soul}
\usepackage{url}
\usepackage[hidelinks]{hyperref}
\usepackage[utf8]{inputenc}
\usepackage[small]{caption}
\usepackage{graphicx}
\usepackage{amsmath}
\usepackage{amsthm}
\usepackage{booktabs}
\usepackage{algorithm}
\usepackage{algorithmic}
\usepackage[switch]{lineno}
\usepackage{multirow}
\usepackage{subcaption}

\usepackage{natbib}

\usepackage{diagbox}

\usepackage{booktabs} % For formal tables
\usepackage{pifont}
\usepackage{nicefrac}
\usepackage{color}
\usepackage{xcolor,colortbl}
\usepackage{wrapfig}
\usepackage{graphicx}
\usepackage{amsmath,amssymb,amsthm}
\usepackage{mathtools}
 
\usepackage{bm}
\usepackage{rotating}
\usepackage{multirow} 
\usepackage{enumitem}
\usepackage[hidelinks]{hyperref}
\newtheorem{thm}{Theorem}
\newtheorem{dfn}{Definition}

\newtheorem{lem}{Lemma}
\newtheorem{ex}{Example}

\newtheorem{remark}{Remark}

\newtheorem{coro}{Corollary}
\newenvironment{sketch}{\noindent{\em Proof sketch.}\rm }{\hfill $\Box$ }

\newtheorem{claim}{Claim}
\newcounter{newct}

\newcommand{\bv}{\begin{array}}

\newcommand{\calP}{\mathcal P}

\newcommand{\expect}{{\mathbb E}}

\newcommand{\myparagraph}[1]{\vspace{1mm}\noindent {\bf\boldmath #1}}

\newcommand{\Omit}[1]{}

\newcommand{\calC}{\mathcal C}

\newcommand{\calD}{\mathcal D}

\newcommand{\atn}{\mathcal{A}}
\newcommand{\qual}{\ell}
\newcommand{\Quallist}{{\mathbf L}}
\newcommand{\Qual}{L}
\newcommand{\qualset}{\mathcal{L}}
\newcommand{\topqual}{\bar{\ell}}
\newcommand{\qualp}{\mathbf{P}}
\newcommand{\wn}{W}
\newcommand{\wnset}{\mathcal{W}}
\newcommand{\wnmax}{\wnset_{\textrm{max}}}
\newcommand{\vt}{v}
\newcommand{\ut}{u}

\newcommand{\Sig}{S}
\newcommand{\Sigv}{\mathbf{S}}
\newcommand{\sigp}{\mathbf{Q}}
\newcommand{\sigset}{\mathcal{S}}

\newcommand{\Rp}{\Tilde{\Sigv}}

\newcommand{\inst}{\mathcal{I}}

\newcommand{\stg}{\sigma}
\newcommand{\stgp}{\Sigma}

\newcommand{\av}{\textrm{AV}}
\newcommand{\avcost}{\textrm{AV/cost}}
\newcommand{\pav}{\textrm{PAV}}
\newcommand{\gc}{\textrm{Greedy Cover}}
\newcommand{\phr}{\textrm{Phragm\'{e}n}}
\newcommand{\mes}{\textrm{MES}}

\newcommand{\informratio}{\mathcal{P}}

\newcommand{\ranking}{R}
\newcommand{\rankset}{\mathcal{R}}

\newcommand{\atone}{a}
\newcommand{\attwo}{b}

\newcommand{\xrv}{x}

\newcommand{\calA}{\mathcal{A}}
\newcommand{\calB}{\mathcal{B}}
\newcommand{\calE}{\mathcal{E}}
\newcommand{\calAp}{\mathcal{A}'}
\newcommand{\calCp}{\mathcal{C}'}
\newcommand{\KL}[2]{D(#1 \| #2)}

\newcommand{\piv}{\textit{Piv}}
\newcommand{\wntop}{\wn_{\textrm{top}}}

\title{Truth-Revealing Participatory Budgeting}

\author{
Qishen Han$^1$
\and
Artem Ivaniuk$^1$\and
Edith Elkind$^2$\and
Lirong Xia$^1$\\
\affiliations
$^1$Rutgers University\\
$^2$Northwestern University\\
}

\begin{document}

\maketitle

\begin{abstract}
Participatory Budgeting (PB) is commonly studied from an axiomatic perspective, where the aim is to design procedurally fair and economically efficient rules for voters with full information regarding their preferences. In contrast, we 
take an epistemic perspective and consider a framework where PB projects have different levels of underlying quality, indicating how well the project will take effect, which cannot be directly observed before implementation. Agents with noisy information cast votes to aggregate their information, and aim to elect a high-quality set of projects. We evaluate the performance of common PB rules by measuring the expected utility of their outcomes, 
compared to the optimal set of projects. We find that the quality of approximation improves as the range of project costs shrinks. When projects have unit cost, these common rules can identify the ``best'' set with probability converging to 1. We also study whether strategic agents have incentives to honestly convey their information in the vote. We find that it happens only under very restrictive conditions. We also run numerical experiments to examine the performance of different rules empirically and support our theoretical findings. 
%EE I feel this does not belong at the end of abstract
%Consequently, we see our truth-revealing PB framework as an important and interesting perspective that views PB as an information aggregation problem.  
\end{abstract}

\section{Introduction}
Participatory budgeting (PB) is a democratic process that allows citizens to directly decide how to allocate a portion of public funds. It solves community needs and improves civic engagement by giving citizens decision power over budget allocation. This form of democratic innovation enjoys growing popularity around the world  
\citep{cabannes}, and has also attracted significant attention from the (computational) social choice community, with researchers formulating new fairness and proportionality desiderata as well as proposing novel aggregation rules~\citep{aziz2021proportionally,peters2021proportional,lackner2021fairness,aziz2020participatory,kraiczy2023adaptive,brill2023proportionality}. 

Despite its popularity, practical uses of PB frequently face implementation issues: it may happen that, after voting is concluded,  the selected projects are discovered to be infeasible due to cost overruns or safety concerns, or simply fail to work as intended~\citep{roth2022pb,cambridge2024pbflaw}. The resulting implementation failures may undermine people's trust in the system and discourage citizens from subsequent participation; eventually, this becomes a threat to democracy. 
Ideally, we would like the voting rules used for PB, in addition to being procedurally fair, to also credibly identify  the objectively ``best'' projects 
based on the votes.
\begin{ex} \label{ex:motive} 
    The participatory budgeting program in Cambridge, MA~\citep{cambridgePB2025} allocates one million dollars annually based on citizen input. All residents aged 12 and older can vote for their preferred projects, choosing from among approximately 20 options. Each project comes with a brief description, which provides a high-level outline of its goals, proposed actions, and anticipated benefits, but no detailed implementation plans. As a consequence, voters cannot know with certainty whether a project will deliver its intended outcomes, with no undesirable side effects. Yet, voters'
    impressions of whether each project is worthwhile are shaped by information they
    receive from various sources: PB campaigns, expert reviews, social media discussions, and local networks. In this setting, can voters successfully identify and elect the ``best projects''?
\end{ex}

Example~\ref{ex:motive} illustrates a different perspective from previous theoretical studies of PB. Indeed, prior works focus on making fair and efficient decisions, given that agents have fixed subjective preferences. We, on the other hand, model the objective perspective: agents do not have full information on which projects are worthwhile. The {\em quality} is the abstraction of this hindsight ground truth on whether a project can be implemented as planned. However, agents must commit to a decision before they know the qualities of the projects. 

This uncertainty suggests a different way of thinking about participatory budgeting. Rather than viewing PB purely as a mechanism for aggregating exogenous (subjective) preferences, we view it through an {\em epistemic} lens: votes reflect not just what citizens want, but what they believe about project quality. Voters' preferences are endogenous, as a result of noisy signals about an underlying ground truth. From this perspective, we also view PB as an information aggregation procedure: voting is a way to aggregate private information from the agents and choose objectively high-quality projects.

In fact, this perspective dates back to the celebrated Condorcet Jury Theorem~\citep{Condorcet1785:Essai} for single-winner voting. It has been subsequently developed into a subfield of social choice termed {\em epistemic social choice}~\citep{Condorcet1785:Essai,Austen96:Information,Pivato13:Voting,ding2021deliberation}, which analyzes whether and when collective decisions can be ``truth-revealing'', that is, whether voting can aggregate information so as to reveal the ``correct'' decision. While there is a large literature on single-winner voting~\citep{Nitzan17:Collective,Caragiannis13:When}, little is known about PB, where multiple projects are chosen subject to the budget constraint. Recent work has begun to explore epistemic aspects of PB~\citep{rey2025epistemic,goel2019knapsack}, yet the following question remains largely open.

\begin{center}
    \textbf{Can voting reveal the high-quality projects in PB?}
\end{center}

\subsection{Our Contribution}
In this paper, we make the first step towards addressing this question to the best of our knowledge. Our contributions are threefold. 

\myparagraph{Modeling (Section~\ref{sec:model}).} 
We propose a truth-revealing PB framework, where agents receive noisy signals on alternatives' qualities; the qualities themselves are not directly observable. An approval PB vote is applied to elect the winners. The goal is to choose a set of alternatives with maximum utilitarian total utility under the budget constraint. 

\myparagraph{Truth-revealing (Section~\ref{sec:informative}).} 
Like the Condorcet Jury Theorem, we first investigate whether informative voting, where agents honestly vote according to their signals, can choose the ``best'' alternative. To provide quantitative analysis, we define {\em performance} of a rule as the ratio between (1) the utility of the outcome under that rule when agents have noisy signals over the ground truth and (2) maximum utility, and analyze the asymptotic performance of several voting rule as the number of agents increases. A performance close to 1 implies a strong capability to reveal the optimal alternatives. Our results suggest that the performance of a voting rule is negatively correlated with the range of project costs. On one extreme, if alternatives have equal costs, many commonly studied PB rules achieve performance converging to $1$, indicating full information aggregation (Theorem~\ref{thm:unit_positive}). On the other extreme, when there is no constraint on the costs, no voting rule achieves an asymptotic performance that is better than $\frac 12$ (Theorem~\ref{thm:general_impossibility}). In the middle, we characterize bounds on the performance parameterized by the ratio $\alpha$ between the highest and the lowest cost, with an upper bound of $\frac{\alpha}{2\alpha - 1}$ (Corollary~\ref{coro:general_impossibility_a}) and a lower bound of $\frac{1}{\lceil \alpha \rceil}$ (Theorem~\ref{thm:general_positive}). 

\myparagraph{Incentive analysis (Section~\ref{sec:strategic}).} We also make a first attempt to analyze strategic behavior in PB through the lens of strategic information aggregation initiated by the seminar work of~\citet{Austen96:Information} for single-winner voting. %Beyond the performance analysis, a natural question is whether agents have incentives to vote informatively. 
We examine our framework from a game-theoretic perspective, and investigate whether informative voting constitutes a Bayes--Nash equilibrium. We find that even a necessary condition for informative voting to be an equilibrium is highly restrictive, as it defines a Lebesgue zero-measure set even in very simple cases (Theorem~\ref{thm:strategic_binary}). This game-theoretic glimpse suggests that informative voting is generically not an equilibrium, and therefore
suggests a rich and highly non-trivial direction for future work, as~\citet{Austen96:Information} did for single-winner voting. 

To complement our theoretical analysis, we conduct experiments on synthetic data. Most rules converge to performance 1 with unit costs and substantially exceed the worst-case upper bound with non-unit costs. Surprisingly, a rule called MES+AV, which was designed to satisfy normative properties, exhibits the strongest truth-revealing capability. Finally, we observe a clear negative correlation between performance and cost dispersion. These empirical findings are consistent with—and further support—our theoretical results.

% It differs from the Gibbard-Satterthwaite theorem, which concerns manipulation of preferences under complete information. In our setting, voters share common preferences over outcomes but hold private signals about project quality; the strategic question is not whether voters misreport preferences, but whether they vote according to their signals.

% Besides the performance, we are also interested in the strategic behavior of agents in the truth-revealing PB framework. Specifically, we investigate the condition that informative voting forms a Nash equilibrium, inspired by~\citep{Austen96:Information}. Unlike the mixed result in the performance, it turns out that even the necessary condition is highly restrictive, leading to a Lebesgue zero-measure set characterized by a series of strict equalities, even in the very simple case (Theorem~\ref{thm:strategic_rare}). This game-theoretical glimpse reveals the rich and highly non-trivial landscape of this truth-revealing PB framework. \Qishen{Play the similar row with ASB. And compare it with Gibbard and Satherwhite.}

% As a consequence, we see truth-revealing participatory budgeting as a potentially interesting and important perspective to understanding PB, combining mechanism design and information aggregation, axiomatic and epistemic social choice, subjective preferences and objective information, and game theory and strategic behaviors. We believe this study can initiate this direction and attract interesting future findings. 

\subsection{Related Work and Discussion}
%\lirong{rewrote this section a little bit. Please double check.}
%EE proofread, edited the wording a bit
\myparagraph{Participatory budgeting} Existing work on PB primarily focuses on the axiomatic approach. Key directions include proposing new efficiency and fairness desiderata~\citep{aziz2021proportionally,aziz2024fair,fain2016core}, designing new rules and mechanisms~\citep{peters2021proportional,brill2023proportionality}, and analyzing strategic aspects~\citep{freeman2021truthful,goel2019knapsack}. Our paper does not investigate fairness; rather,  it focuses on the utilitarian welfare, i.e., the selection of the ``best'' alternatives. It is an interesting direction to combine truth-seeking and fairness concerns in a scenario where agents have different preferences contingent on the quality.

\myparagraph{The epistemic perspective on voting} can be traced back to the celebrated Condorcet Jury Theorem~\citep{Condorcet1785:Essai}, which proves that, for two alternatives, majority voting is truth-revealing in large elections. 
Subsequent work has explored the power and limitations of truth-revealing voting for single-winner~\citep{Young88:Condorcet,Owen89:Proving, Conitzer05:Common,Pivato13:Voting,Caragiannis2016:When,Azari12:Random,Azari14:Statistical,Xia2016:Bayesian} and multi-winner elections~\citep{Procaccia12:Maximum,allouche2022truth}. For PB, \citet{rey2025epistemic} initiated the epistemic analysis by asking whether common PB rules can be viewed as maximum likelihood estimators of some probabilistic models. \citet{goel2019knapsack} show that a knapsack mechanism for PB is a maximum likelihood estimator when the projects are divisible, though their results do not extend to indivisible projects. These two contributions can be viewed as PB counterparts to the seminal work of \citet{Conitzer05:Common}. In contrast, the truth-revealing part of our paper follows the approach of \citet{Caragiannis2016:When} and addresses a different question: is a PB rule truth-revealing for a given model? 
Our positive results on the performance extend the Condorcet Jury Theorem intuition to participatory budgeting, showing that collective decisions can reveal project quality under appropriate conditions. The inapproximability result, however, indicates that the multi-project structure and budget constraints in PB introduce fundamental barriers to information aggregation that do not arise in binary or single-winner elections.

%a rule being the MLE of {\em some} model does not mean that the rule is truth-revealing for another model. Additionally, a rule not being the MLE of {\em any} model does not mean that the rule is not truth-revealing for some model.

\myparagraph{Strategic behavior in epistemic social choice}  The classic paper of~\citet{Austen96:Information} modeled the Condorcet Jury theorem scenario as a Bayesian game and showed that truthful voting is often not a Bayes--Nash Equilibrium (BNE). A series of works~\citep{Feddersen97:Voting,Xia16:Quantitative,Schoenebeck2021:Wisdom,han2023wisdom,deng2024aggregation} further investigated the existence and truth-revealing capability of BNE in a range of single-winner settings. Our work on incentive analysis (Section~\ref{sec:strategic}) is a first step towards understanding agents' strategic behavior in epistemic PB; our model differs from the setting of previous work on strategic aspects of PB~\citep{freeman2021truthful,goel2019knapsack}, where agents' (subjective) preferences are assumed to be given exogenously. At a high level, our results resemble those of~\citet{Austen96:Information}, while our proofs are more challenging.

The strategic behavior in truth-revealing PB is also related to the famous Gibbard--Satterthwaite impossibility theorem~\citep{Gibbard73:Manipulation,Satterthwaite75:Strategy}; however, in our model strategic agents do not misreport their fully informed preferences; rather, they manipulate how they infer their preference from noisy information, and act accordingly.

%when there is a ground truth yet voters may receive different signals.

%A closely related question to epistemic social choice is informed voting, which aims to unify the axiomatic and epistemic perspectives of social choice in one framework. In their framework, agents have different preferences contingent on the underlying ground truth, and the aim is to make an informed decision, the favorable decision, as if everyone knows the ground truth. 

\section{Preliminaries}
For any $n\in\mathbb N$, let $[n] = \{1, 2, \dots, n\}$ denote the set of agents. Let $\atn$ denote the set of alternatives with $|\atn| = m$. We use $i$ to denote a generic agent and $j$ to denote a generic alternative. 
%Let $\wn$ be the set of winning alternatives. 
Each alternative $j$ has a cost $c_j > 0$. The costs are additive: for every subset $\wn \subseteq \atn$, we write $c(\wn) = \sum_{j \in \wn} c_j$. We denote the ratio between the highest and lowest cost among the alternatives by $\alpha \ge 1$.  There is a budget constraint $B$: a winning set $\wn\subseteq\atn$ is {\em feasible} if it satisfies $c(\wn)\le B$. Without loss of generality, we assume that $B \ge \max_{j \in \atn} c_j$. Let $\wnset \subseteq2^\atn$ be the set of all feasible winning sets. In the {\em unit cost setting} we have $c_j=1$ for all $j\in \atn$. In this case, PB simplifies to multi-winner voting; the budget constraint $B$ becomes the cardinality constraint, and is assumed to be an integer. 

In this paper, we focus on PB with approval ballots. That is, each agent $i$ submits an approval ballot $\Rp_i \in \{0, 1\}^\atn$.
A voting rule $r$ takes the ballot profile $\Rp = (\Rp_i)_{i \in [n]}$, the costs of alternatives, and the budget constraint $B$ as input, and outputs a winning set $\wn \in \wnset$.  %Importantly, $r$ has no information on the qualities $\Qual_j$, priors $\qualp_j$, or the information structures $\sigp_j$.\lirong{unclear what this means. A voting rule is not a person. Wouldn't it suffice to just define the rule as a mapping from preferences to the collective decision?} 
%In the first part of the paper, we focus on {\em informative voting}: the behavior where all agents vote according to their signals, i.e., $\Sigv_i = \Rp_i$ for all $i\in[n]$. The game-theoretic setting, where agents may act strategically, is considered in Section~\ref{sec:strategic}. 
We consider several voting rules commonly used and studied in PB. Their formal definitions can be found in Appendix~\ref{apx:rules}.

\myparagraph{(Greedy) Approval Voting} ($\av$) orders the alternatives by their number of approvals, from largest to smallest, and adds them in this order, skipping those that do not fit in the remaining budget. Ties are resolved using a tie-breaking rule.

\myparagraph{Approval voting per cost} ($\avcost$) runs like $\av$, except the alternatives are ordered by the number of approvals divided by the cost. Under unit cost, $\av$ and $\avcost$ coincide. 

\myparagraph{Proportional Approval Voting} ($\pav$) assigns a score to each feasible set $\wn$, so that each agent who approves $k$ alternatives in $\wn$ contributes $(1+ \frac12+\cdots +\frac1k)$ to the score. It outputs a feasible winning set with the maximum score. 

\myparagraph{Greedy Cover} adds alternatives one by one; at each step, it selects an alternative that increases the ``coverage'' (i.e., the number of agents who approve at least one alternative) as much as possible.

\myparagraph{Method of Equal Shares~\citep{peters2021proportional}} ($\mes$): The budget is split equally among the agents, who can use their share to pay for alternatives they approve. In each round, the rule selects an alternative with the lowest per-approver cost, and shares its cost among its approvers as equally as possible. It terminates when no remaining alternative is affordable. $\mes$ may fail to exhaust the budget. In this case, a second voting rule $r$ is applied to spend the leftover budget on the remaining alternatives. We denote this by $\mes$+$r$.  

\myparagraph{$\phr$}: Agents accumulate "voting power" over time uniformly. A project is selected at the moment when the total accumulated power of its approvers first reaches its cost, and that power is then consumed. The process continues until the budget is exhausted. 

% Given a voting profile $\Rp$, let $\hist(\Rp)$ be the histogram of $\Rp$. That is, for a vote $S \in \{0, 1\}^{\atn}$, $\hist(\Rp)_S$ is the number of agents voting $S$.

\section{An Epistemic Model for PB}
\label{sec:model}
We assume that each alternative $j\in\atn$ has an underlying ground truth {\em quality} $\Qual_j$ drawn from the {\em quality range} $\qualset = \{0, 1, \dots, \topqual\}$ according to a prior 
$\qualp_j = (\qualp_j^1,\ldots, \qualp_j^{\topqual})$, so that $\Pr[\Qual_j = \qual] = \qualp_j^\qual$. The prior $\qualp_j$ is common knowledge. %, and satisfy $\qualp_j^\qual \in (0, 1)$ for all $j\in\atn$, $\qual\in\qualset$.
 Let $\Quallist=(\Qual_j)_{j\in\atn}$ denote the quality vector of all alternatives.

The qualities of the alternatives are not directly observable, but agents receive signals about them. In this paper, we study binary signals $\sigset = \{0, 1\}$, where signals $0$ and $1$ represent an overall negative/positive impression towards an alternative, respectively. 
Each alternative $j\in\atn$ is associated with an {\em information structure} $\sigp_j\in(0, 1)^\qualset$, where $\sigp_j^\qual = \Pr[\Sig = 1 \mid \Qual_j = \qual]$ denotes the probability for an agent to receive a positive signal about $j$ conditioned on $j$'s quality being $\qual$; this information structure is shared by all agents. Let $\sigp = (\sigp_j)_{j\in\atn}$ be the list of information structures of all alternatives. Given a quality vector $\Quallist=(\Qual_j)_{j\in\atn}$, the vector $\sigp^{\Quallist}=(\sigp^{L_j}_j)_{j\in\atn}$
captures the distribution of the agents' signals: 
every agent $i\in [n]$ receives a vector of signals $\Sigv_i \in \sigset^\atn$, where $\Sig_{ij} = (\Sigv_i)_j$ ($i$' signal regarding alternative $j$) is independently drawn according to $\sigp_j^{\Qual_j}$.
%We assume that signals regarding different alternatives are independent. For each alternative $j$, the signals $\Sig_{ij}$ that agents receive regarding $j$  are i.i.d conditioned on $j$'s quality $\Qual_j$. 

%In this paper, the goal is to identify an optimal feasible set of alternatives $W$, as determined by their qualities. 
We assume there is a common ground-truth utility function $\vt$. We consider two utility functions:
\begin{itemize}
    \item {\em Normal utility}. Each agent receives utility $\Qual_j$ from each winner $j \in \wn$, and $\vt^N(\wn) = \sum_{j \in \wn} \Qual_j$. This function assumes that the utility of an alternative does not depend on its ``size'' (as captured by the cost). 
    \item {\em Cost-proportional utility}. Each agent receives utility $c_j\cdot \Qual_j$ from each winner $j \in \wn$, and $\vt^C(\wn) = \sum_{j \in \wn} c_j\cdot \Qual_j$. This function assumes that the ``impact'' of an alternative scales with its cost. 
\end{itemize}
%We assume all agents have the same utility function.
Note that in the unit cost setting, normal utility coincides with cost-proportional utility. 

Given a utility function $\vt$, let $\wnmax^\vt = \{\wn \in \wnset \mid \vt(\wn) \ge \vt(\wn')\text{ for all $\wn' \in \wnset$}\}$ denote the set of all feasible winning sets that maximize the utility function $\vt$. When $\vt$ is clear from the context, we sometimes omit it from the notation and write $\wnmax$. Our goal is to select a winning set in $\wnmax$ based on agents' votes. 

We consider both {\em informative voting}, where each agent directly reports their signal for each alternative (Section~\ref{sec:informative}), and strategic voting (Section~\ref{sec:strategic}).
Both utility functions depend on the ground truth, i.e., the qualities of the alternatives, which are not directly observable. Therefore, agents with different signals may have different preferences. 

\begin{ex}
    \label{ex:running}
    Consider a PB instance with $n = 100$ agents and $m = 3$ alternatives $\atn = \{1, 2, 3\}$ 
    with costs $c_1 = 4$, $c_2 = 3$, and $c_3 = 2$, and budget $B = 7$. The quality range is binary, i.e., $\qualset = \{0, 1\}$. The common priors are $\qualp_1^1 = 0.8$, $\qualp_2^1 = 0.6$, and $\qualp_3^1 = 0.4$, so that alternative $1$ has probability $0.8$ of having quality $1$, alternative $2$ has probability $0.6$ of having quality $1$, and alternative $3$ has probability $0.4$ of having quality $1$. The signal distributions $\sigp_j^\qual = \Pr[\Sig = 1 \mid \Qual_j = \qual]$ are shown in Table~\ref{tbl:running_signal}.
    
    \begin{table}[htbp]
    \centering
    \begin{tabular}{@{}crrrr@{}}
    \toprule
    \multicolumn{1}{r}{} & $\Qual_j$ & Alternative $1$ & Alternative $2$ & Alternative $3$ \\ \midrule
    \multirow{2}{*}{$\sigp_j^\qual$} & $1$ & $0.7$ & $0.65$ & $0.6$ \\
     & $0$ & $0.3$ & $0.35$ & $0.4$ \\ \bottomrule 
    \end{tabular}
    \caption{Signal distributions for Example~\ref{ex:running}\label{tbl:running_signal}}
    \end{table}
    
    Under normal utility $\vt^N$, if the true quality vector is $\Quallist = (1, 1, 0)$, then the optimal winning set is $\wn = \{1, 2\}$ with utility $\vt^N(\wn) = 2$, while $\wn' = \{1, 3\}$ has utility $1$ and $\wn'' = \{2, 3\}$ has utility $1$. \hfill \qed
\end{ex}

A {\em PB environment} $\inst$ specifies the number of agents $n$, the set of alternatives $\atn$, the costs $(c_j)_{j \in \atn}$, the budget $B$, the quality range $\qualset$ and priors $(\qualp_j)_{j\in\atn}$, the information structures $\sigp$, and the utility function $\vt$. 
% $$\inst = (n, \atn, (c_j)_{j \in \atn}, B, \qualset, (\qualp_j)_{j\in \atn}, (\sigp_j^{\qual})_{j\in \atn, \qual \in \qualset}, \vt ).$$  

% Given an environment $\inst$ and a strategy profile $\stgp$, we can measure the ex-ante likelihood of a ``best'' winning set being reached. Let {\em fidelity} denote this likelihood. $\acc^\inst (\stgp) = \Pr[r(\Rp) \in \wnmax \mid \inst, \stgp]$. When the context is clear, we omit the environment $\inst$ and write it as $\acc(\stgp)$. 

Given an environment $\inst$, we want to measure how well a voting rule $r$ can elect the ``best'' set of winners. We define the {\em performance} of a rule $r$ as the worst-case ratio, over all possible quality vectors $\Quallist$, between (a) the expected utility of the set selected by $r$ and (b) the utility of an optimal winning set, where the expectation is taken over $n$ agents' signals conditioned on $\Quallist$:  
\begin{equation}
    \informratio_\inst(r, n) = \min_{\Quallist \in \qualset^\atn} \frac{\expect_{\Rp \sim (\sigp^{\Qual})^n} [\vt(r(\Rp))]}{ \max_{\wn \in \wnset} \vt(\wn)}. 
\end{equation}

When the context is clear, we omit $\inst$ and write $\informratio(r, n)$.

% In order to present our asymptotical results, we define a sequence of environments $\{\inst_n\}_{n=1}^{\infty}$, with every $\inst_n$ sharing all attributes except for the number of agents $n$.\lirong{this is unnecessarily complicated. I see why you need this for the generality of proving positive result for strategic behavior, but since we don't do it in this paper, I strongly recommend simplifying it.}

\section{Performance of PB Rules}\label{sec:informative}

In this section, we focus on
{\em informative voting}, i.e., we assume that all agents vote according to their signals, so that $\Sigv_i = \Rp_i$ for all $i\in[n]$. First we show that, in the unit-cost setting, under mild assumptions the performance of many voting rules converges to $1$ as $n\to+\infty$. This implies that, by voting truthfully, agents are able to identify optimal winning sets even though their signals are noisy.
%outcomes in these scenarios are as good as what can be obtained when the agents have full information on the underlying qualities. 

Before we present the proof, we formulate a condition in information structures that is used in the proof.
\begin{dfn}
    An information structure $(\sigp_j^{\qual})_{j\in \atn, \qual \in \qualset}$ is {\em quality dominant} if for every pair of (not necessarily distinct) alternatives $j, j' \in \atn$ and for all $\qual, \qual' \in \qualset$ such that $\qual > \qual'$ it holds that $\sigp_j^{\qual} > \sigp_{j'}^{\qual'}$.  An environment $\inst$ is {\em quality dominant} if its information structure is quality dominant; a sequence of environments is {\em quality dominant} if every element of this sequence is quality dominant. 
\end{dfn}
Quality dominance requires that higher-quality alternatives are more likely to convey positive signals to agents, both across alternatives (i.e., for $j\neq j'$) and for a fixed alternative (i.e., for $j=j'$). We believe that this is a mild assumption as quality is often legible. %\lirong{unclear why introducing this definition. Must motivate.}

\begin{thm}
    \label{thm:unit_positive}
    %EE added that this is for unit costs
    In the unit-cost setting, 
    for any quality dominant environment, every voting rule $r \in \{\av, \gc,  \phr, \mes\textrm{+}\av, \mes \textrm{+} \phr,$ $ \pav\}$ satisfies $\lim_{n\to\infty} \informratio(r, n) = 1$. 
\end{thm}

Theorem~\ref{thm:unit_positive} shows an interesting insight that under the unit cost, even rules designed to satisfy normative properties can reveal the ground truth with high probability. 

\begin{sketch} The key tool in the proof is 
the Law of Large Numbers (formalized via Hoeffding’s inequality): we use it, together with the union bound, to show that, as $n \to \infty$, for each alternative $j\in\atn$ its approval count $A_j$ concentrates tightly around its expected value with probability converging to 1. In a quality-dominant environment, high-quality alternatives have strictly higher signal probabilities than low-quality ones, and therefore receive strictly more approval votes. This directly implies that $\av$ picks a highest-quality remaining alternative (within the budget constraint) at every step. For other rules, we leverage this observation more indirectly. 

For $\gc$, $\phr$, and $\mes$ we apply the approval concentration idea to subsets of agents. For $\gc$, we consider the uncovered agents at each step. For $\phr$ and $\mes$, we consider all sets of agents with the same remaining voting power/budget at each step, respectively. Crucially, these subsets are defined in terms of already-selected alternatives. Consequently, in each subset, higher-quality alternatives have more approval votes. For $\gc$, this directly means that a highest-quality alternative will be selected. For $\phr$ and $\mes$, a highest-quality alternative will have the earliest affordable time/lowest affordable per-agent price, respectively, and consequently will be selected. While $\mes$ may not exhaust the budget, we subsequently apply a second-round vote on the remaining budget and alternatives, and exactly the same profile, so we can apply the results for $\av$/$\phr$. 

For $\pav$, quality dominance ensures that optimal sets have strictly higher expected $\pav$ scores than non-optimal sets, since higher-quality alternatives contribute more to the harmonic function. Hoeffding's inequality and the union bound guarantee that the $\pav$ score of each committee concentrates around its expectation with a probability converging to 1. As $n\to+\infty$, the gap between optimal and suboptimal expected scores dominates the concentration error, so $\pav$ selects a quality-maximizing set with probability converging to 1. The full proof is in Appendix~\ref{apx:unit_positive}.
\end{sketch}

On the other hand, for general costs, the performance of any voting rule $r$ does not exceed $\frac12$.  

% \begin{dfn}
%      A voting rule $r$ is {\em anonymous} if swapping for any two voting profiles $\Rp$ and $\Rp'$ such that $\hist(\Rp) = \hist(\Rp')$, $r(\Rp) = r(\Rp')$. 
    
%     % A voting rule $r$ is {\em neutral} if swapping any two alternatives $j$ and $j$' in every vote leading to swapping $j$ and $j'$ in the outcome (on whether $j$ and $j'$ is in $\wn$). \Qishen{Need improvement on this definition.}
% \end{dfn}

\begin{thm}
\label{thm:general_impossibility}
    For any constant $\varepsilon > 0$, any $n > 0$, any voting rule $r$, and a utility function $\vt \in \{\vt^N, \vt^C\}$, there exists a quality dominant environment $\inst$ such that $ \informratio_{\inst}(r, n) < \frac{1}{2} + \varepsilon$. 
\end{thm}

\begin{proof}
We construct two scenarios such that: (1) the optimal winning sets are different, and the utilities of different winning sets are drastically different, and (2) $r$ cannot distinguish these two scenarios and hence has to suffer utility loss in at least one of them. We first present the construction and the proof for $\vt^N$. Then we present the construction for $\vt^C$, which admits a similar proof. 

\myparagraph{Construction.} Each scenario consists of an environment and a quality vector. We denote the two scenarios by $(\inst_1, \Quallist_1)$ and $(\inst_2, \Quallist_2)$, respectively. The two environments share all attributes except 
for 
the signal distribution. In both environments, we have alternatives $\atn = \{1, 2, \dots, m\}$, budget $B = m$, and costs $c_1 = m$, $c_2 = c_3 = \cdots = c_{m} = 1$. The quality range is binary $\qualset = \{0, 1\}$.
The information structures for both environments are quality dominant and are given by Table~\ref{tbl:impossibility}.
The quality vectors are
 $\Quallist_1=(1, 0, \dots, 0)$ and 
 $\Quallist_2=(1, \dots, 1)$, respectively.
 
\begin{table}[htbp]

\centering
\begin{tabular}{@{}crrr@{}}
\toprule
\multicolumn{1}{r}{} & Quality & Alternative $1$ & other alternatives \\ \midrule
\multirow{2}{*}{$\inst_1$} & 1 & 0.6 & 0.6 \\
 & 0 & 0.4 & 0.4  \\ \cmidrule(l){1-4} 
\multirow{2}{*}{$\inst_2$} & 1 & 0.6 & 0.4  \\
 & 0 & 0.3 & 0.2 \\ \bottomrule 
\end{tabular}
\caption{Signal distribution of two instances\label{tbl:impossibility}}
\end{table}
 
% \lirong{confusing. Didn't you already define the two environment above?} Firstly, in $\inst_1$, $\Qual_1 = 1$ and $\Qual_j = 0$ for all $j = 2, \cdots, m$. We denote this quality vector as $\Quallist_1 = (1, \vec{0})$. Secondly, in $\inst_2$, $\Qual_j = 1$ for all $j =1, 2, \cdots, m$. We denote this quality vector as $\Quallist_2 = (1, \vec{1})$. 
For scenario 1, the optimal winning set is $\{1\}$ with $\vt^N(\{1\}) = 1$, whereas $\vt^N(\atn \setminus \{1\}) = 0$.  For scenario 2, on the contrary, the optimal winning set is $\atn \setminus \{1\}$ with $\vt^N(\atn \setminus \{1\}) = m-1$, whereas $\vt^N(\{1\}) = 1$. However, the two environments have exactly the same signal distribution, so no voting rule can distinguish them. 
% Finally, we construct the common prior so that these two scenarios occurs in the corresponding environment with high probability, respectively. Let $\delta \in (0, 1)$ be a small constant whose concrete value will be determined later in the proof. For $\inst_1$, let $\qualp_1^1 = 1 - \delta$, and $\qualp_j^1 = \delta$ for $j = 2, \cdots, m$. For $\inst_2$, let $\qualp_j^1 = 1 - \delta$ for $j = 1, \cdots, m$. Consequently, $\Pr[\Qual = \Quallist_1 \mid \inst_1] \ge 1 - m\cdot \delta$ and $\Pr[\Qual = \Quallist_2 \mid \inst_2] \ge 1 - m\cdot \delta$.

% Now let $\pi$ be the distribution of voting profile derived by the signal distribution $\sigp^{\Qual}$ in the two scenarios and informative voting. Given an instance $\inst$ and a quality vector $\Qual$, we can write the expected utility as follows. 
% \begin{align*}
%    \expect_{\Rp \sim \sigp^{\Qual}} [\vt(r(\Rp))] = &\ \sum_{\Rp} \pi(\Rp) \cdot \vt(r(\Rp)). 
% \end{align*}

% By assigning each scenario into this formula and clustering voting profiles with the outcome, we have
Now, given a voting rule $r$, we can write the expected utility under $r$ in each scenario as follows. For $(\inst_1, \Quallist_1)$, we have 
\begin{align*}
    \expect_{\Rp \sim (\sigp^{\Quallist_1})^n} [\vt^N(r(\Rp))] = &\ \Pr[r(\Rp) = \{1\}].
\end{align*}
For $(\inst_2, \Quallist_2)$, we have 
\begin{align*}
    \expect_{\Rp \sim (\sigp^{\Quallist_2})^n} [\vt^N(r(\Rp))] = &\ \Pr[r(\Rp) = \{1\}]\\
    +&\  (m-1)\cdot \Pr[r(\Rp) = (\atn \setminus \{1\})].
\end{align*}

Thus, no matter how the voting rule $r$ maps profiles to outcomes, either $\expect_{\Rp \sim (\sigp^{\Quallist_1})^n} [\vt^N(r(\Rp))] \le \frac12 = \frac12\max_\wn\vt^N(\wn)$ in the first scenario, which implies $\informratio_{\inst_1} (r, n) \le \frac12$, or $\expect_{\Rp \sim (\sigp^{\Quallist_2})^n} [\vt^N(r(\Rp))] \le \frac{m}{2} = \frac{m}{2(m-1)} \max_\wn \vt^N(\wn)$ in the second scenario, which implies $\informratio_{\inst_2}(r, n) \le \frac{m}{2(m-1)}$.
For $m > \frac{1}{2\varepsilon} + 1$ we have $\informratio_{\inst_2}(r, n) < \frac{1}{2} + \varepsilon$. This completes the proof for $\vt^N$.
% \begin{align*}
%     \informratio_{\inst_1}(r, \stgp^*) = &\ \frac{\sum_{\Qual \in \qualset^\atn} \Pr[\Qual \mid \inst]\cdot\expect[\vt(r(\Rp)) \mid \inst_1, \stgp^*, \Qual]}{\sum_{\Qual \in \qualset^\atn} \Pr[\Qual \mid \inst]\cdot \max_{\wn \in \wnset} \vt(\wn)}\\
%     \le &\ \frac{(1 - m\delta)\cdot \expect[\vt(r(\Rp)) \mid \inst_1, \stgp^*, \Qual] + (m-1)m\delta }{(1 - m\delta)\cdot\max_{\wn \in \wnset} \vt(\wn) + (m-1)m\delta}\\
%     \le&\ \frac{\frac12(1 - m\delta) + (m-1)m\delta}{(1 - m\delta) + (m-1)m\delta}
% \end{align*}
% The first inequality holds by assume the max utility is reached in all other quality vectors expect for $\Quallist_1$, of which the prior is at most $m\delta$ and the utility is at most $m - 1$. By taking $\delta < \frac{2\varepsilon}{m(m-1)}$, we have $\informratio_{\inst_1}(\stgp^*) < \frac{1}{2} + \varepsilon$.
% In the second case, similarly, 
% \begin{align*}
%     \informratio_{\inst_2}(r, \stgp^*) \le&\ \frac{\frac{m}{2}(1 - m\delta) + (m-1)m\delta}{(m-1)(1 - m\delta) + (m-1)m\delta}
% \end{align*}
The construction for $\vt^C$ is in Appendix~\ref{apx:general_impossible}. 
\end{proof}

We can strengthen Theorem~\ref{thm:general_impossibility} to instances where the ratio between project costs is bounded, 
using similar techniques.

\begin{dfn}
    An environment $\inst$ is {\em $\alpha$-cost bounded} if $\frac{\max_{j \in \atn} c_j}{\min_{j \in \atn} c_j} \le \alpha$.  
    A sequence of environments is $\alpha$-cost bounded if every one of its environments is $\alpha$-cost bounded. 
\end{dfn}

\begin{coro}
\label{coro:general_impossibility_a}
    For any constant $\varepsilon > 0$, any $n > 0$, and any voting rule $r$, there exists a quality-dominant and $\alpha$-cost bounded environment $\inst$ with utility $\vt^N$ such that $\informratio_{\inst}(r,n) < \frac{\lceil\alpha\rceil - 1}{2\lceil\alpha\rceil - 3} + \varepsilon$; and there exists a quality-dominant and $\alpha$-cost bounded environment $\inst$ with utility $\vt^C$ such that $ \informratio_{\inst}(r,n) < \frac{\alpha}{2\alpha - 1} + \varepsilon$. 
\end{coro}
The proof of Corollary~\ref{coro:general_impossibility_a} can be found in Appendix~\ref{apx:coro}. 

However, for $\av$ we can still obtain a lower bound on performance when the cost ratio is bounded. 

\begin{thm}
\label{thm:general_positive}
    For any $\alpha > 1$ and any quality-dominant environment that is $\alpha$-cost bounded, 
    $\lim_{n\to\infty} \informratio (\av, n) \ge \frac{1}{\lceil\alpha\rceil}$. 
\end{thm}

\begin{sketch}
We fix an arbitrary $\Quallist$ and show that the utility ratio between the outcome of AV and the optimum is at least $1/\lceil \alpha \rceil$ with high probability for all sufficiently large $n$. 
First, reasoning as in Theorem~\ref{thm:unit_positive}, we show that, with probability converging to 1, higher-quality alternatives receive more approvals and hence appear earlier in $\av$'s selection order. 
Thus, at each step, $\av$ selects an alternative with the 
%EE noun phrase, so no hyphen
highest quality among all remaining alternatives within the budget constraint. We then use this property to bound the ratio between the utility of $\av$'s outcome and that of the optimum. We deal with two utility models separately. 

\myparagraph{Normal Utility.} For $\vt^N$, we leverage the following important observation. We call a winning set $\wn$ {\em maximal} if for any $j\in\atn\setminus \wn$ the cost of $\wn\cup\{j\}$ exceeds the budget. For any two maximal sets $\wn$ and $\wn'$, the ratio of their cardinalities is bounded by $1 / \lceil \alpha \rceil$. Based on this observation, we construct a mapping between the outcome of $\av$ (denoted by $\wn_{\av}$ and the optimum (denoted by $\wn^*$). We order the alternatives in $\wn_{\av}$ and in $\wn^*$ in the descending order of their expected approval votes (equivalently, their quality), and map each alternative in $\wn_{\av}$ to at most $\lceil \alpha \rceil$ items in $\wn^*$. 
The first observation ($\av$ greedily picks highest-quality alternatives) guarantees that the quality of each alternative in $\wn_{\av}$ is not lower than that of any alternatives it is mapped to in $\wn^*$. Therefore, the quality of each alternative in $\wn_{\av}$ is at least $\frac{1}{\lceil \alpha \rceil }$ times the quality of its image. The second observation ($|\wn^*| \le \lceil \alpha \rceil \cdot |\wn_{\av}|$) guarantees that the image of $\wn_{\av}$ is the entire set $\wn^*$. Consequently, %merging all the alternatives together, 
we obtain $\vt^N(\wn_{\av}) \ge \frac{1}{\lceil \alpha \rceil} \vt^N(\wn^*)$. 

\myparagraph{Cost-Proportional Utility.} We consider a subset $\wntop \subseteq \wn_{\av}$ such that the expected approval count (thus the quality) of any alternative in $\wntop$ is not lower than that of any alternative in $(\wn^*\setminus \wntop)$. For example, $\wntop$ includes all alternatives $\av$ selects in the beginning before skipping any alternative because of the budget constraint. Therefore, the total cost of $\wntop$ is at least $B - \alpha$. 
When $B \ge 2\alpha$ or $c(\wntop) \ge B/2$, we immediately obtain $\vt^C(\wntop) \ge \frac12 \vt^C(\wn^*)$. Otherwise, we consider $j \in (\wn^* \setminus \wntop)$ with the highest expected number of approvals. Then $c_j > B/2$, as otherwise $j$ would be in $\wntop$. Then, like the mapping in the $\vt^N$ case, we can partition $\wn^*$ and map each subset to a non-overlapping subset of $\wn_{\av}$, so that, within each pair, the $\wn_{\av}$ subset has higher quality, and the cost ratio of the two subsets is at most $\alpha$. Merging all the subsets, we get $\vt^C(\wn_{\av}) \ge \frac{1}{\alpha} \vt^C(\wn^*)$. 

The full proof is in Appendix~\ref{apx:general_positive}. 
\end{sketch}

\section{Strategic Behavior in Epistemic PB}
\label{sec:strategic}
In this section, we ask whether informative voting is likely to be a stable outcome under strategic behavior. Unfortunately, we find that even the necessary condition for this is very restrictive and satisfied by very few environments, even for the very simple case of two alternatives and unit costs.  

%\subsection{Strategy and Solution Concept}
The {\em strategy} of an agent $i$, denoted by $\stg_i$, is a mapping from his/her signals $\Sigv_i$ to his/her report $\Rp_i$. The {\em informative strategy} $\stg^*$ identically maps the signals to the report. A {\em strategy profile} $\stgp = (\stg_i)_{i\in [n]}$ is the vector of all agents' strategies. In the {\em informative profile} $\stgp^*$ all agents play $\stg^*$. 
Given a strategy profile $\stgp$, a utility function $\vt$, a private signal $\Sigv_i$, and a rule $r$, agent $i$'s expected utility is given by $\ut_i(\stgp, \Sigv_i) = \expect_\wn[\vt(\wn) \mid \stgp, \Sigv_i]$, where $W$ is the random variable representing the winning set in agent $i$'s eyes in the following process. Given the signal $\Sigv_i$, agent $i$ forms a posterior belief about the qualities $\Quallist$, which leads to a distribution over  other agents' signals; given their strategies $\stgp_{-i}$, agent $i$ forms a belief about their votes, and hence the winning set $W$ after applying the rule $r$.

A strategy profile $\stgp=(\stg_i)_{i \in [n]}$ is a {\em Bayes--Nash equilibrium} (BNE) under a rule $r$ and utility function $\vt$ if no agent can unilaterally change his/her strategy to achieve a higher expected utility: for every agent $i$, every signal $\Sigv_i$, and every strategy $\sigma_i'$ for $i$, it holds that $\ut_i(\stgp, \Sigv_i) \ge \ut_i((\stg_i, \stgp_{-i}), \Sigv_i)$, where $\stgp_{-i}$ denotes the other agents' strategies in $\stgp$. 

% Let $\{\stgp_n^*\}_{n = 1}^{\infty}$ be the sequence of strategy profiles defined on $\{\inst_n\}_{n=1}^{\infty}$, such that $\stgp_n^*$ is the informative strategy profile in the environment $\inst_n$.
% We say that $\{\stgp_n^*\}_{n = 1}^{\infty}$ is asymptotical BNE if there exists an $N_0 > 0$ such that for all $n > N_0$, $\stgp_n^*$ is a BNE in instance $\inst_n$. \lirong{Let's not introduce the full definition of BNE in sequence. Just define the notaion for a fixed strategy of agent and assuming that all agents for all $n$ use that strategy. This suffices because we only study informative voting in this paper.}

% Let $\inst_{-n\sigp} = (\atn, (c_j)_{j \in \atn}, B, \qualset, (\qualp_j)_{j\in \atn}, \vt)$ denote the attributes in an environment except for the number of voters and the signal distribution. 
% \begin{thm}
% \label{thm:strategic_rare}
%     Let $r = \av$. For any fixed attributes $\inst_{-n\sigp}$, let $\sigpset^* \subseteq (0, 1)^{\qualset m}$ be the set of signal distributions that for any $\sigp \in \sigpset^*$, the environment sequence $\{\inst_n\}$, where $\inst_n = (n, \sigp, \inst_{-n\sigp})$, satisfies that the informative voting sequence $\{\inst_*_n\}$ is asymptotical BNE. Then $Q^*$ is a Lebesgue 0-measure set under $(0, 1)^{\qualset m}$. 
% \end{thm}

% \Qishen{The proof will proceed in three steps. We first consider the toy example with the clearest characterization. Then we go to unit cost case. Finally we go to general case. If the general case does not work out, we do a different special case for general case, such as m = 4. }

\subsection{Two Alternatives: Rarely a BNE}

We first consider the simplest possible case: the set of alternatives is $\atn = \{\atone, \attwo\}$, alternatives have unit costs, the budget is $B = 1$, the quality range is binary $\qualset = \{0, 1\}$, and the voting rule is $\av$.

% \lirong{People may wonder why this is not just Austen-Smith\& Banks. Highlight the small difference and highlight why such small difference makes our result in this section hard to proof (and how you proved it).}

\begin{thm}
\label{thm:strategic_binary}
    Given an environment $\inst$ with two alternatives, unit costs, and binary qualities, when $n \to \infty$, a necessary condition for a sequence of informative voting profiles to be a BNE is $\sqrt{\sigp_{\atone}^1 \cdot \sigp_{\attwo}^0} + \sqrt{(1-\sigp_{\atone}^1) \cdot (1-\sigp_{\attwo}^0)} =    \sqrt{\sigp_{\atone}^0 \cdot \sigp_{\attwo}^1} + \sqrt{(1-\sigp_{\atone}^0) \cdot (1-\sigp_{\attwo}^1)},$
    % \begin{align*}
    % &\ \sqrt{\sigp_{\atone}^1 \cdot \sigp_{\attwo}^0} + \sqrt{(1-\sigp_{\atone}^1) \cdot (1-\sigp_{\attwo}^0)}\\
    % &\ =    \sqrt{\sigp_{\atone}^0 \cdot \sigp_{\attwo}^1} + \sqrt{(1-\sigp_{\atone}^0) \cdot (1-\sigp_{\attwo}^1)},  
    % \end{align*}
    which induces a Lebesgue 0-measure set under $(0, 1)^{\bar{\qual} \times m}$. 
\end{thm}

Theorem~\ref{thm:strategic_binary} is similar in spirit to the results of ~\citet{Austen96:Information} for majority voting: informative voting forms a BNE for very few information structures. However, as we will show in the proof sketch, the nature of approval voting leads to considerably more complicated \textit{pivotal} cases (cases in which changing a vote will change the outcome), so that the analysis for majority voting cannot be easily extended. 

\begin{sketch}
The beginning of the analysis resembles that of~\citet{Austen96:Information}.  We start from an informative voting profile $\stgp^*$ and fix an agent $i$ with signal $\Sigv_i$. Let $\stgp$ be the strategy profile after $i$ changes his/her strategy. Agent $i$ has an incentive to deviate to $\stgp$ only if $\ut_i(\stgp^*, \Sigv_i) <\ut_i(\stgp, \Sigv_i)$. An important observation is that a change of strategy affects $i$'s utility only in the event that (1) $\atone$ and $\attwo$ obtain the same number of approvals from the rest of the agents (so that $i$'s strategy change flips the outcome), (2) $\atone$ and $\attwo$ have different quality (so that the outcome affects utility). For example, suppose $i$ gets a signal $\Sigv_i = [0,0]$ (not approving alternatives) and is deviating to $[0, 1]$ (approving $\attwo$ only). Let $\xrv_{\atone}$ and $\xrv_\attwo$ be the number of approvals that $\atone$ and $\attwo$ obtain from other agents, respectively. We can write the difference in expected utility as $\ut_i(\stgp^*, \Sigv_i) - \ut_i(\stgp, \Sigv_i) = $
\begin{align*}
        &\ \Pr[\Quallist = [1, 0] \mid \Sigv_i] \cdot \Pr[\xrv_{\atone} = \xrv_{\attwo} \mid \Quallist = [1, 0]] \\
        &\ - \Pr[\Quallist = [0, 1] \mid \Sigv_i] \cdot \Pr[\xrv_{\atone} = \xrv_{\attwo} \mid \Quallist = [0, 1]]. 
    \end{align*}

Here is where approval voting differs from majority voting over two alternatives (as considered by \citet{Austen96:Information}). Under the majority rule, each agents votes for exactly one alternative in $\atn$, so $\xrv_{\atone} = \xrv_{\attwo}$ only when both approval counts are exactly $(n-1)/2$. Under approval voting, however, an agent can approve both or neither of the alternatives, so the approval count can vary from $0$ to $n-1$, and the probability is highly nontrivial to calculate for finite $n$. Therefore, we seek an approximation when $n$ is large. 
Notice that, given $\Qual_{\atone}$ and $\Qual_{\attwo}$, $\xrv_{\atone}$ and $\xrv_{\atone}$ are random variables following binomial distributions $B(n-1, \sigp_{\atone}^{\Qual_{\atone}})$ and $B(n-1, \sigp_{\attwo}^{\Qual_{\attwo}})$, respectively. Therefore, Lemma~\ref{lem:saddlepoint} (see Appendix~\ref{apx:saddlepoint} for the proof) provides a fairly accurate approximation for large $n$. 

\begin{lem}
    \label{lem:saddlepoint}
    For any integer $n > 0$ and $p_1, p_2 \in (0, 1)$, let $X_1 \sim B(n, p_1)$ and $X_2 \sim B(n, p_2)$, set $Q(p_1, p_2) = \sqrt{p_1p_2} + \sqrt{(1-p_1)(1-p_2)}$, and let $M$ be a constant determined by $p_1$ and $p_2$.
    Then $\Pr[X_1 = X_2] = $
    \begin{equation*}
    M(p_1, p_2)\cdot (Q(p_1, p_2))^{2n+1}\cdot n^{-1/2}
    \cdot (1 + O(\frac1n)).
    \end{equation*}
\end{lem}
 
By Lemma~\ref{lem:saddlepoint}, the difference between the utilities is dominated by $Q$. If $Q(\sigp_{\atone}^0, \sigp_{\attwo}^1) > Q(\sigp_{\atone}^1, \sigp_{\attwo}^0)$, then $\ut_i(\stgp^*, \Sigv_i) < \ut_i(\stgp, \Sigv_i)$ for all sufficiently large $n$. Thus, we obtain a necessary condition: informative voting is a BNE only if $Q(\sigp_{\atone}^0, \sigp_{\attwo}^1) \le Q(\sigp_{\atone}^1, \sigp_{\attwo}^0)$. 

Then, we apply the same reasoning in the ``opposite'' scenario: $i$ gets signal $\Sigv_i = [0,1]$ and deviates to $[0, 1]$. By symmetry, we get a similar yet opposite necessary condition: informative voting is a BNE only if $Q(\sigp_{\atone}^0, \sigp_{\attwo}^1) \ge Q(\sigp_{\atone}^1, \sigp_{\attwo}^0)$. Combining the two conditions completes the proof. 

The full proof of Theorem~\ref{thm:strategic_binary} appears in Appendix~\ref{apx:saddlepoint}.\end{sketch}

\subsection{Unit Cost: A Restrictive Necessary Condition }

A similar line of analysis extends to the general setting where alternatives have unit costs. We construct two ``opposite'' scenarios and identify the necessary conditions under which deviation is not profitable. The scenarios are designed so that any pivotal instance in which a deviation leads to a utility gain in the first scenario corresponds to one in which the same deviation leads to a utility loss in the second, and vice versa. Consequently, the most likely utility-decreasing pivotal events in the two scenarios must occur with comparable likelihood. This observation leads to the following result. 

\begin{thm}[informal]
\label{thm:strategic_unit}
Given an environment $\inst$ with unit cost, when $n\to\infty$, informative voting is a BNE only if $G^-_{\min}(\sigp) = G^+_{\min}(\sigp)$, where $G^+_{\min}$ and $G^-_{\min}$ are two functions that are inversely related to the likelihood of the most likely utility-decreasing pivotal events in the two scenarios. 
\end{thm}

 Theorem~\ref{thm:strategic_unit} does not give a direct zero-measure result like in the binary case, but, intuitively, the strict equality is still restrictive. We illustrate this intuition by a numerical simulation. For each setting, we emulate all the pivotal events and randomly sample 1000 information structures. For each information structure $\sigp$, we test if $G^-_{\min}(\sigp)$  and $G^+_{\min}(\sigp)$ are equal (under a small error tolerance). In the setting with $m=5$ alternatives and $B=2$ winners, the condition holds in 7 out of 1000 samples; in the setting with $m = 6$ and $B = 3$, the condition holds in 2 out of 1000 samples. This simulation result supports our intuition that the necessary condition is indeed restrictive. The formal statement and the full proof of Theorem~\ref{thm:strategic_unit} can be found in Appendix~\ref{apx:strategic_unit}, and more details of the simulation are in Appendix~\ref{apx:rarity}.

\section{Experiments}
We conduct experiments to evaluate the average-case performance of different voting rules across various environments. In each experiment, we fix the number of agents $n$, the number of alternatives $m$, the cost ratio $\alpha$ (and fix the min-cost at 1), the budget $B$, the utility function $\vt$, and the quality range $\qualset$. We first randomly generate the costs and qualities of alternatives and the information structure, and draw agents' signals. The detailed generation process is in Appendix~\ref{apx:experiment}.

We consider eight rules: $\av$, $\avcost$, $\pav$, $\gc$, $\phr$, $\mes$, $\mes$+$\av$, and $\mes$+$\phr$. For each environment, we generate 10000 instances and take the average ratio between the utility of the voting outcome and the optimum utility as the empirical performance. 

%\subsection{Result}
We evaluate the performance of the proposed rules across multiple environments. In the unit-cost setting ($\alpha=1$; Figure~\ref{fig:performance_1}), most rules converge rapidly to performance 1 as $n\to\infty$, consistent with Theorem~\ref{thm:unit_positive}. However, $\mes$ may stop before fully spending the budget, leading to utility loss, and $\gc$ exhibits notably slow convergence to 1 (see Figure~\ref{fig:gc} in the Appendix). In the general-cost setting ($\alpha=5$; Figure~\ref{fig:performance_5}), performance remains below 1, but substantially exceeds the worst-case upper bound of approximately 0.57 by Corollary~\ref{coro:general_impossibility_a}. More surprisingly, we find that $\mes$+$\av$ has the highest performance among all rules, guaranteed by a paired $t$-test with 95\% confidence. 
\begin{figure}[H]
    \centering
\includegraphics[width=0.99\linewidth]{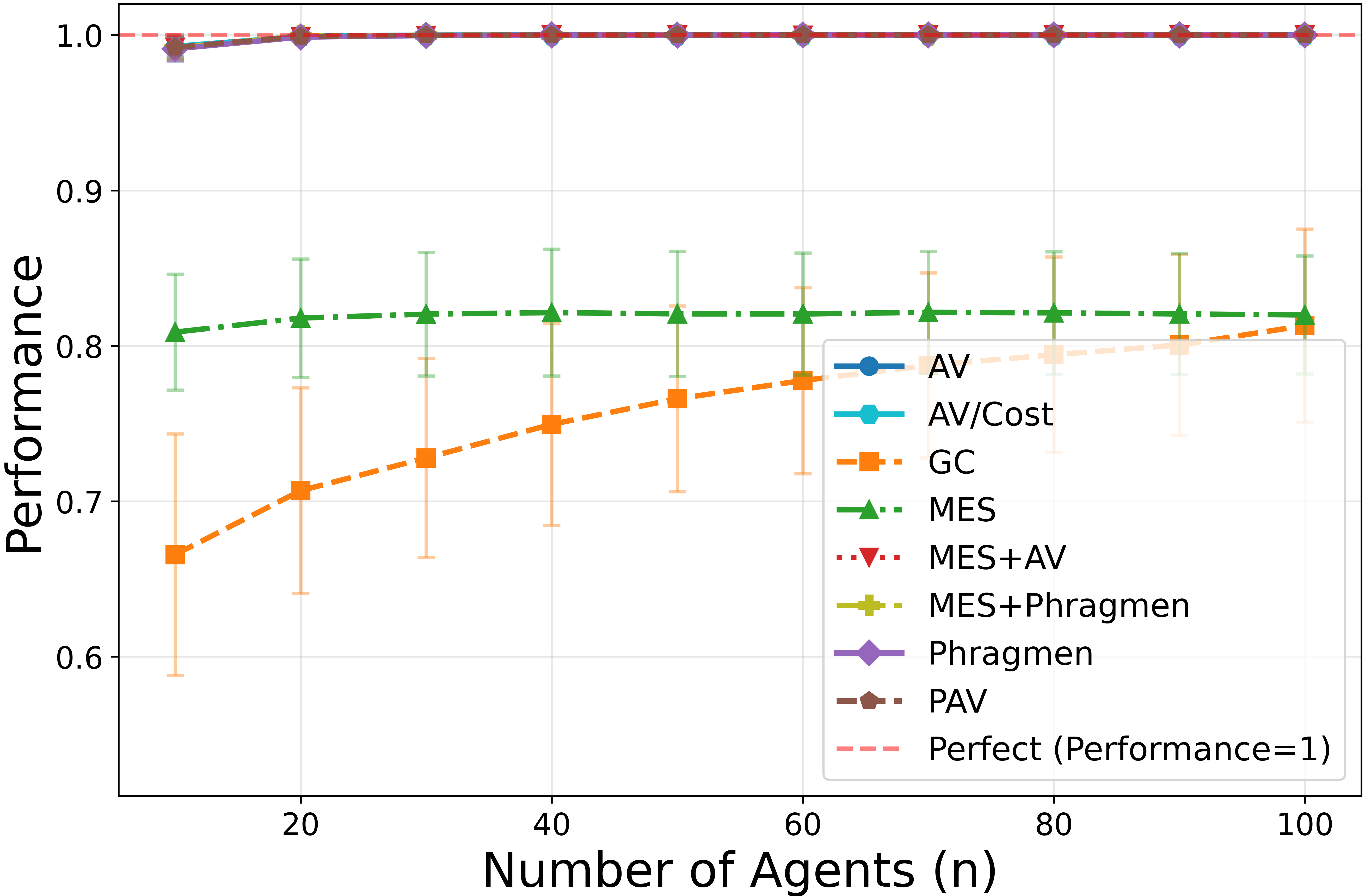}
    \caption{Performance with unit cost ($\alpha = 1$)}
    \label{fig:performance_1}
\end{figure}
 \begin{figure}[H]
    \centering
\includegraphics[width=0.99\linewidth]{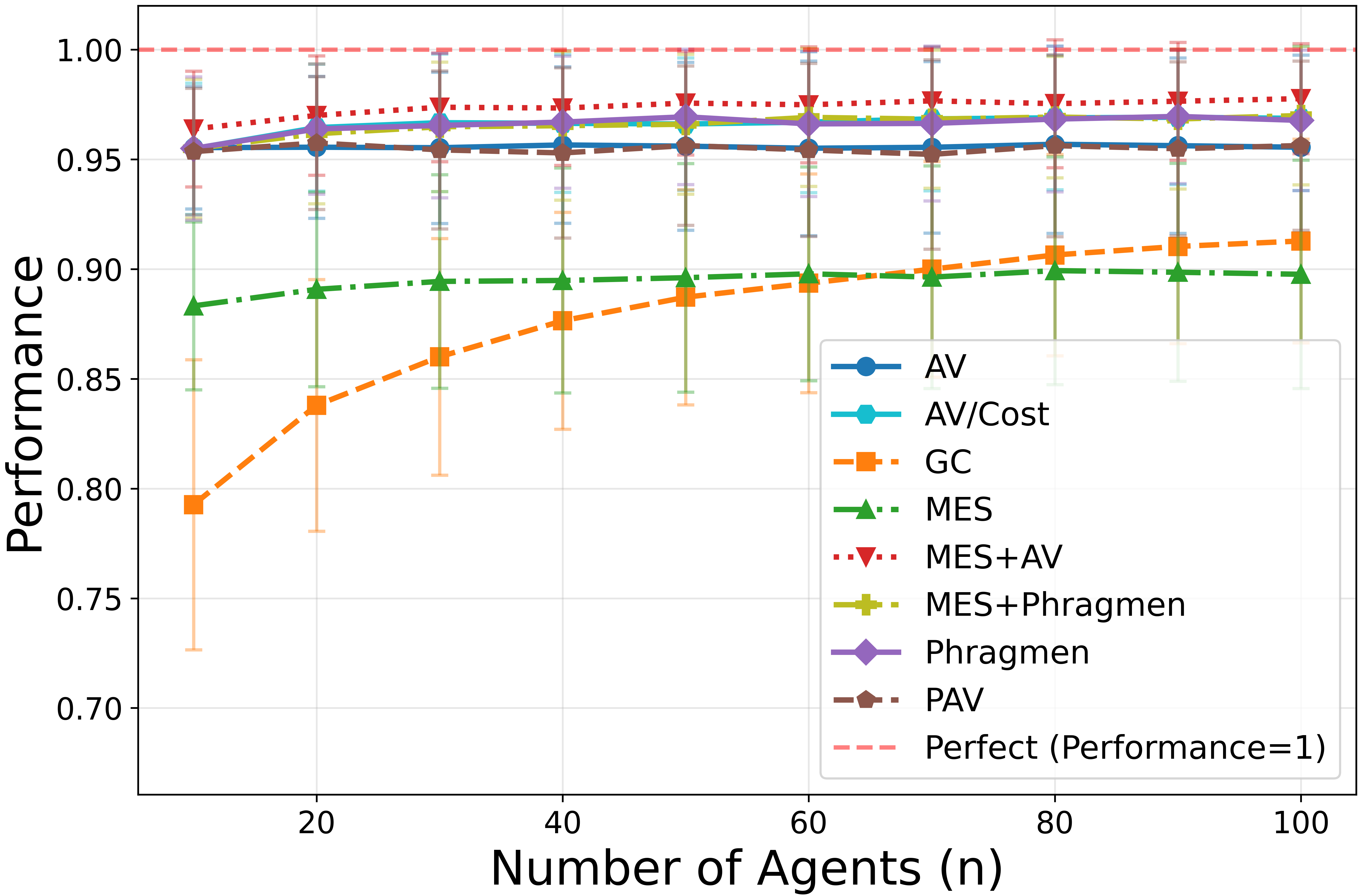}
    \caption{Performance with general cost ($\alpha = 5$)}
    \label{fig:performance_5}
\end{figure}

\begin{figure}[H]
    \centering
\includegraphics[width=0.99\linewidth]{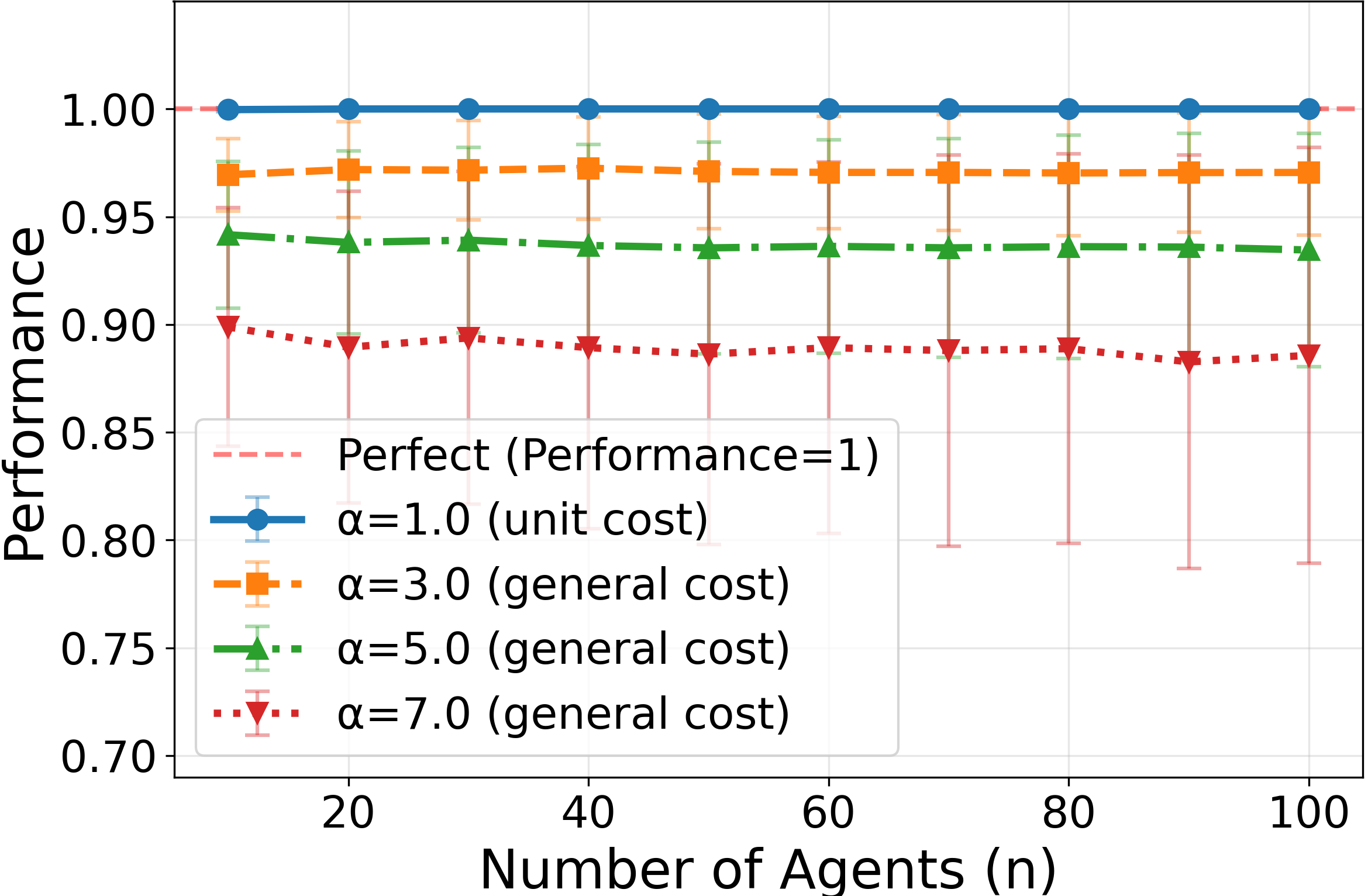}
    \caption{Performance on different $\alpha$ ($\av$ rule, $B = 7$)}
    \label{fig:alpha}
\end{figure}

We further analyze the impact of the cost ratio (Figure~\ref{fig:alpha}). In line with our theoretical bounds (Corollary~\ref{coro:general_impossibility_a} and Theorem~\ref{thm:general_positive}, performance generally decreases as $\alpha$ increases. A key open question concerns the underlying cause of this degradation: whether it is driven primarily by tighter budget constraints (compared to the total cost) or by increased price dispersion. Additional experimental results, including the effects of budget size, utility specifications, and other factors, are provided in Appendix~\ref{apx:experiment}.

% Secondly, we test the effect of different utility functions. It turns out that there is no significant difference. 

% Thirdly, we fix other parameters and test the effect of the budget. It turns out that the performance is positively related to the budget in general. One reason for this is that when the budget is small, the ``friction loss" (budget that does not contain any alternative) is large compared to the full cost of the alternatives. This also aligns the worst-case construction in Theorem~\ref{thm:general_impossibility} and Theorem~\ref{thm:general_positive}. 

% Finally, we test the effect of cost ratio $\alpha$, we fix other parameters and keep the budget $B$ to be proportional to the expected total cost of the full alternative set, which is $(m+1)(\alpha+1)/2$. This also aligns with our theoretical observation that a larger $\alpha$ leads to a lower performance in general. 

% \myparagraph{Take Away Message.} Our experiments suggest that the choice of voting rule has limited impact on average-case performance—practitioners may select rules based on other criteria such as fairness or simplicity without sacrificing information aggregation. The key factors affecting performance are environmental: budget slack and cost heterogeneity. These findings reinforce our theoretical results and suggest that ensuring sufficient budget relative to project costs is more critical than fine-tuning the aggregation rule.

\section{Future Work}
We view this epistemic perspective as a promising direction with substantial room for further exploration. Our current framework assumes that all agents share the same ground-truth utility function. A natural next step is to unify the epistemic and axiomatic perspectives and investigate ``informed participatory budgeting,'' where agents hold heterogeneous utility functions conditioned on the ground-truth quality, and the objective is to achieve fair and efficient outcomes as if the ground truth were known. Additional research directions include designing rules with stronger performance guarantees, conducting beyond-worst-case analysis, and studying strategic behavior in the general-cost setting.

\bibliographystyle{named}
\bibliography{references,new_ref}

\clearpage
\onecolumn
\appendix
\section{Definition of Voting Rules.}
\label{apx:rules}

\begin{dfn}[Approval Voting ($\av$)]
Given a voting profile $\Rp$, $\av$ orders alternatives in non-increasing order of their approval counts, i.e., by $|\{i \in [n] : \Rp_{ij} = 1\}|$. It then selects alternatives greedily according to this order, adding an alternative whenever doing so does not violate the budget constraint, until no further alternative can be added. Ties are resolved using a fixed tie-breaking rule.
\end{dfn}

\begin{dfn}[Approval Voting per cost ($\avcost$)]
    Given a voting profile $\Rp$, $\avcost$ orders alternatives in non-increasing order of their approval-per-cost ratios, that is, 
    $$
    \frac{|\{i \in [n] : \Rp_{ij} = 1\}|}{c_j}.
    $$
    It then selects alternatives greedily according to this order, adding an alternative whenever doing so does not violate the budget constraint, until no further alternative can be added. Ties are resolved using a fixed tie-breaking rule.
\end{dfn}
    
\begin{dfn}[Proportional Approval Voting ($\pav$)]
Given a voting profile $\Rp$, for any feasible set of alternatives $\wn$ with $c(\wn) \le B$, define its PAV score as
\[
\pav(\wn) \;=\; \sum_{i \in [n]} \sum_{t=1}^{|\wn \cap \{j \in \atn \mid \Rp_{ij} = 1\}|} \frac{1}{t}.
\]
The rule $\pav$ selects a feasible set $\wn$ that maximizes the PAV score. If several sets achieve the maximum score, a tie-breaking rule is applied.
\end{dfn}
    
\begin{dfn}[Greedy Cover ($\gc$)]
Given a voting profile $\Rp$, $\gc$ proceeds iteratively as follows. Initially, all agents are unsatisfied and the winning set is empty. At each step, among the alternatives that can be added without exceeding the remaining budget, the rule selects an alternative $j$ that is approved by the largest number of currently unsatisfied agents, i.e., that maximizes
\[
|\{ i \in [n] : \Rp_{ij} = 1 \text{ and } i \text{ is unsatisfied} \}|.
\]
After selecting $j$, all agents approving $j$ become satisfied and are removed from further consideration. The process continues until adding any alternative will exceed the budget constraint. Ties are resolved using a fixed tie-breaking rule.
\end{dfn}
    
\begin{dfn}[Method of Equal Shares ($\mes$)]
Given a voting profile $\Rp$, $\mes$ initially assigns each agent $i \in [n]$ an equal budget share $b_i = B/n$. The rule proceeds iteratively. At each step, among all alternatives that can be funded by their approving agents, $\mes$ selects an alternative $j$ that minimizes the maximum individual payment required from any approving agent, subject to the constraint that the alternative is affordable. Formally, for an alternative $j$, let $p_j$ be the minimum value such that
\[
\sum_{i \in \{k \in [n] : \Rp_{kj} = 1\}} \min(b_i, p_j) \;\ge\; c_j.
\]
The rule selects an alternative with minimal $p_j$, subtracts the corresponding payments from the budget shares of the approving agents, and adds $j$ to the winning set. This process continues until no alternative can be funded. The rule may leave some budget unspent; in this case, an additional voting rule may be applied to exhaust the remaining budget among the unselected alternatives. Ties are resolved using a fixed tie-breaking rule.
\end{dfn}
    
\begin{dfn}[Phragmén ($\phr$)] The $\phr$ rule has two equivalent representations, one in a continuous formation, and the other in a discrete formation.

\textbf{Continuous Formation.} The rule is modeled as a dynamic process where each voter accumulates voting power at a uniform rate over time. Starting from zero at $t=0$, each agent's voting power increases linearly until an alternative’s approvers collectively possess the total cost of the alternative. At the instant this threshold is met, the alternative is added to the committee, and the voting power of their approvers is reset to zero to reflect the expenditure of their voting power. Conversely, agents who do not approve the selected alternative retain their accumulated voting power. This sequential selection continues until the committee reaches its target budget, with ties resolved by a predefined ordering.

\textbf{Discrete Formation.} Given a voting profile $\Rp$, (sequential) $\phr$ iteratively constructs a winning set $W$ starting from $W = \emptyset$ with initial agent loads $O_i = 0$. In each iteration, we select the alternative $j^* \in \mathcal{A} \setminus W$ that maintains the budget constraint $\sum_{k \in W \cup \{j^*\}} c_k \le B$ and minimizes the value $x$ satisfying $\sum_{i \in [n]: \Rp_{i,j}=1} \max(0, x - O_i) = c_j$; subsequently, we update $W \leftarrow W \cup \{j^*\}$ and set $O_i \leftarrow \max(O_i, x)$ for all $i$ such that $\Rp_{i,j^*} = 1$, repeating the process until no further alternatives can be affordably added.
\end{dfn} 

\section{Proof of Theorem~\ref{thm:unit_positive}}
\label{apx:unit_positive}
\begin{proof}    
      \myparagraph{Case 1: Approval Voting ($\av$).}
      A mistake occurs only if there exists a welfare-maximizing winning set 
      $\wn^\star \in \wnmax$, an alternative $j \in \wn^\star$, and an alternative 
      $\ell \notin \wn^\star$ such that the number of approvals for $\ell$ is at least the number of approvals for $j$. Since $\av$ selects the $k$ alternatives with highest approval counts, $\av(\Rp)\notin\wnmax$ implies that there exists $j\in\wnmax$ and $\ell\notin\wnmax$ such that $A_\ell \ge A_j$, where $A_j = |\{i \in [n] : \Rp_{ij} = 1\}|$ and $A_\ell = |\{i \in [n] : \Rp_{i\ell} = 1\}|$ denote the approval counts for alternatives $j$ and $\ell$ respectively.
      
      Define
      \[
      D := \min_{\wn^\star \in \wnmax}\;
     \min_{j\in\wn^\star,\;\ell\notin\wn^\star}
     \bigl(\sigp_j^{\Qual_j}-\sigp_\ell^{\Qual_\ell}\bigr).
      \]
      By quality dominance, $D>0$.
      Fix $\wn^\star\in\wnmax$, $j\in\wn^\star$, and $\ell\notin\wn^\star$, and define $D' := A_j - A_\ell$.

      For each agent $i$, define
      \[
        X_i := \mathbf{1}\{\Rp_{ij} = 1\} - \mathbf{1}\{\Rp_{i\ell} = 1\}.
      \]
      Then $X_i \in [-1,1]$ and $D' = \sum_{i=1}^n X_i$. Applying Hoeffding's inequality:
      \[
        \Pr[D' \le 0]
        \le \exp\!\Big(-\tfrac{1}{2} n D^2\Big).
      \]
      By the union bound over all $\wn^\star\in\wnmax$ and all pairs $j\in\wn^\star$, $\ell\notin\wn^\star$,
      \[
        \Pr[\av(\Rp) \notin \wnmax] \le m^2 \exp\!\Big(-\tfrac{1}{2} n D^2\Big).
      \]
      Hence
      \[
        \Pr[\av(\Rp) \in \wnmax] \ge 1 - m^2 \exp\!\Big(-\tfrac{1}{2} n D^2\Big),
      \]
      which tends to $1$ exponentially fast as $n\to\infty$. This implies that $\lim_{n\to\infty} \informratio(\av, n) = 1$.
      
      \myparagraph{Case 2: Proportional Approval Voting ($\pav$).}
      For a subset $\wn\subseteq \atn$, define per-agent contribution
      \[
        \phi_\wn(\Rp_i) = \sum_{t=1}^{|\wn\cap \{j \in \atn \mid \Rp_{ij} = 1\}|} \frac{1}{t},
      \]
      and total score $\pav(\wn) = \sum_{i=1}^n \phi_\wn(\Rp_i)$. Let $s(\wn) = \expect[\phi_\wn(\Rp_i)]$ be the expected per-agent score (as agent signals are i.i.d).
      Given the quality vector $\Qual$, the explicit form of $s(W)$ is
      \begin{equation*}
          s(W) = \sum_{A \in W} \prod_{j \in A} \sigp_j^{\Qual_j} \prod_{j \in (\wn\setminus A)} (1 - \sigp_j^{\Qual_j}) \cdot H_{|A|}. 
      \end{equation*}
     
      Let $\wnset_k = \{\wn \subseteq \atn : |\wn|=k\}$ and define
      \[
        D_s := \min_{\wn\neq\wn^\star} \bigl(s(\wn^\star)-s(\wn)\bigr),
      \]
      Let $\wnmax := \arg\max_{\wn\in\wnset_k} s(\wn)$ and define
      \[
      D_s := \min_{\wn^\star\in\wnmax}\;
           \min_{\wn\in\wnset_k\setminus\wnmax}
           \bigl(s(\wn^\star)-s(\wn)\bigr).
      \]

      Note that since $|\wn^*| = |\wn| = B$ and $\wn^* \in \wnmax$, we can create a mapping $f$ from $\wn^*$ to $\wn$, such that for all $j \in \wn^*$, $\Qual_j \ge \Qual_{f(j)}$ and thus $\sigp_j^{\Qual_j} \ge \sigp_{f(j)}^{\Qual_{f(j)}}$ (with at least one strictly larger), by quality dominance. Note that $s(w)$ solely depends on the signal distribution and increases when some $\sigp_j^{\Qual_j}$ increases. Therefore, $ D_s>0$.
      Under the quality dominance assumption, the welfare ordering over committees coincides with the ordering induced by $s(\wn)$. Hence $\wnmax$ is exactly the set of welfare-maximizing winning sets.

      Each $\phi_\wn(\Rp_i)\in [0, H_k]$ where $H_k = 1 + \tfrac12 + \cdots + \tfrac{1}{k}$ is the $k$-th harmonic number. Note that $\phi_\wn(\Rp_i)\le H_k$ since an agent can approve at most $k$ alternatives in $\wn$. Writing
      \[
      \pav(\wn)-n\cdot s(\wn) =\sum_{i=1}^n\bigl(\phi_\wn(\Rp_i)-\expect[\phi_\wn(\Rp_i)]\bigr),
      \]
      we obtain a sum of independent, bounded, mean-zero random variables.
      Hence, by Hoeffding's inequality, for each fixed $\wn$ and $\varepsilon>0$:
      \[
        \Pr\big[|\pav(\wn) - n \cdot s(\wn)| \ge n\varepsilon\big]
        \le 2\exp\!\Big(-\frac{2n\varepsilon^2}{(H_k)^2}\Big).
      \]
      Applying the union bound over $\binom{m}{k}$ possible $\wn$:
      \begin{align*}
    &\ \Pr\!\big[\forall \wn\in\wnset_k:\;|\pav(\wn) - n\cdot  s(\wn)| < n\varepsilon\big]\\
        &\ \ge 1 - 2\binom{m}{k}\exp\!\Big(-\frac{2n\varepsilon^2}{(H_k)^2}\Big).
      \end{align*}
      Choosing $\varepsilon = D_s/2$, we have
      \[
        \pav(\wnmax) - \pav(\wn) \ge n\cdot (s(\wnmax)-s(\wn)-2\varepsilon) \ge 0.
      \]
      Hence
      \[
        \Pr[\pav(\Rp) = \wnmax]
        \ge 1 - 2\binom{m}{k}\exp\!\Big(-\frac{2n D_s^2}{(H_k)^2}\Big),
      \]
      which tends to $1$ exponentially fast as $n\to\infty$. This implies that $\lim_{n\to\infty} \informratio(\pav, n) = 1$.
      
      \myparagraph{Case 3: Greedy Cover.}
      We prove this by induction on the number of selections $k$. Let $U_k \subseteq [n]$ denote the set of unsatisfied agents before the $k$-th
      selection, with $U_1 = [n]$, and $U_{k+1} = U_k \setminus \{i \in U_k : j_k \in \Rp_i\}$ where $j_k$ is the alternative selected at step $k$. Let $S_k$ denote the set of alternatives remaining unselected at step $k$, with $S_1=\atn$ and $S_{k+1}=S_k\setminus\{j_k\}$.

      \textbf{Base Case ($k=1$).} By Lemma~\ref{lem:gc_av}, for $k=1$, $\gc(\Rp) = \av(\Rp)$, so the result holds by Case 1.
      
      \textbf{Inductive Step.} Assume inductively that, with probability tending to $1$ as $n\to\infty$, there exists $\wn^\star\in\wnmax$ such that $\wn_k\subseteq\wn^\star$.
      
      At step $k+1$, the $\gc$ rule selects
      \[
      j^*_{k+1} = \arg\max_{j \in \atn \setminus \wn_k} |\{ i \in U_k : \Rp_{ij} = 1 \}|,
      \]
      and updates
      \[
      \wn_{k+1} = \wn_k \cup \{j^*_{k+1}\}, \quad
      U_{k+1} = U_k \setminus \{ i \in U_k : \Rp_{ij^*_{k+1}} = 1\}.
      \]
      
      \textit{Case 3a: $U_{k+1} \neq \emptyset$.} 
      the $(k+1)$-th selection chooses the highest approval alternative in the unsatisfied agents and remaining alternatives, e.g., $(U_k,S_k)$.
      Conditional on $U_k$, the approval signals of agents in $U_k$ remain independent and identically distributed, since satisfaction depends only on previously selected alternatives. Fix $\wn^\star\in\wnmax$ such that $\wn_k\subseteq\wn^\star$. Among the remaining alternatives $\wn^\star\setminus\wn_k$, quality dominance implies strictly higher expected approval from agents in $U_k$ than for any alternative in $S_k\setminus\wn^\star$. Applying Hoeffding’s inequality as in Case~1, the probability that a suboptimal alternative is selected at step $k+1$ decays exponentially in $|U_k|=\Theta(n)$.

      \textit{Case 3b: $U_{k+1} = \emptyset$.} Then all agents are already satisfied, and any further selection is immaterial to $\vt(\wn)$. In this case, $\gc$ selects arbitrarily among the remaining alternatives.
      
      By Lemma~\ref{lem:gc_unsatisfied}, $\Pr[\text{$\exists k$ such that $|U_k| = o(n)$}] \to 0$ as $n \to \infty$. Hence, with probability tending to $1$, $|U_k|=\Theta(n)$ at each step, so Case 3a applies.
      
      \textbf{Conclusion.} By Lemmas~\ref{lem:gc_unsatisfied} and~\ref{lem:gc_av} and induction, each step of the $\gc$ rule asymptotically selects the best remaining alternative under $\vt(\wn)$, and thus $\lim_{n \to \infty} \informratio(\gc, n) = 1$.

\begin{lem}
  \label{lem:gc_unsatisfied}
  For every fixed budget $k$,
  \[
  \Pr[|U_k| = o(n)] \to 0 \quad \text{as } n\to\infty,
  \]
  where $U_k$ is the set of unsatisfied agents after $k$ selections under $\gc$.
\end{lem}

\begin{proof}[Proof of Lemma~\ref{lem:gc_unsatisfied}]
\textbf{Base Case.} For $k=1$, $|U_2|$ is binomial with parameters $(n,1-\sigp_{j^\star}^{\Qual_{j^\star}})$, where $j^\star$ is the selected alternative. By Hoeffding's inequality, for any $\epsilon>0$,
\[
\Pr[|U_2| \le \epsilon n] \le \exp(-\Theta(n)).
\]

\textbf{Inductive Step.} Assume the lemma holds for some $k$. Then $U_k \neq \emptyset$. Conditional on $U_k$ and the selected alternative $j_k$, $|U_{k+1}|$ is a sum of $|U_k|$ independent Bernoulli variables with mean $1-\sigp_{j_k}^{\Qual_{j_k}}$. By Hoeffding's inequality,
\[
\Pr[|U_{k+1}| \le \epsilon |U_k|] \le \exp(-\Theta(|U_k|)).
\]

\textbf{Conclusion.} Since the number of steps $k$ is fixed, applying the union bound over $k$ preserves exponential decay. Hence,
\[
\Pr[\exists k : |U_k| = o(n)] \to 0.
\]
\end{proof}

\begin{lem}
\label{lem:gc_av}
For $k = 1$,
\[
\gc(\Rp) = \av(\Rp).
\]
\end{lem}

\begin{proof}[Proof of Lemma~\ref{lem:gc_av}]
At $k = 1$, no agents have yet been satisfied, so $\gc$ selects
\[
j^* = \arg\max_{j \in \atn} |\{i \in [n] : \Rp_{ij} = 1\}|,
\]
which coincides with $\av$'s $k=1$ choice.
\end{proof}

\medskip 
\myparagraph{Set-up of Proof for $\phr$ and $\mes+\av$}
Fix an environment and a quality vector $\Qual$.
If $B\ge m$, then the unique maximal feasible set is $\atn$ and both $\phr$ and $\mes+\av$ output $\atn$, so the performance is $1$.
Hence assume $B\le m$. If $\max_{\wn\in\wnset}\vt(\wn)=0$, then $\vt(r(\Rp))=0$ for all rules and the ratio is $1$; hence assume the optimum is positive.

Write $p_j:=\sigp_j^{\Qual_j}$ and $A_j:=|\{i\in[n]:\Rp_{ij}=1\}|$.
Quality dominance implies a strict global gap
\[
D := \min_{\substack{j,\ell\in\atn\\ \qual>\qual'}}\bigl(\sigp_j^{\qual}-\sigp_\ell^{\qual'}\bigr) > 0,
\]
so in particular, whenever $\Qual_j>\Qual_\ell$ we have $p_j-p_\ell\ge D$.

Let $q^\star$ be the $B$-th highest value among $\{\Qual_j:j\in\atn\}$ (ties allowed) and define
\[
H:=\{j:\Qual_j>q^\star\}
\]
\[
T:=\{j:\Qual_j=q^\star\}
\]
\[
L:=\{j:\Qual_j<q^\star\}
\]
Then $\wnmax$ consists exactly of sets $H\cup S$ where $S\subseteq T$ and $|S|=B-|H|$.

Define the event
\[
\mathcal{E}:=\{A_j>A_\ell\ \text{for all }(j,\ell)\text{ with }\Qual_j>\Qual_\ell\}.
\]
Arguing as in Case~1 (AV) with gap $D$ and a union bound over at most $m^2$ ordered pairs, we have $\Pr(\mathcal{E})\to 1$.

\begin{lem}[Uniform within-type concentration]
\label{lem:within_type_conc}
We show that the approval vote dominance leading by quality dominance can be extended to certain subset of agents. Fix an environment and a quality vector $\Qual$, and write $p_j:=\sigp_j^{\Qual_j}$ for all $j\in\atn$.
For any $S\subseteq\atn$ with $|S|\le B$ and any $u\in\{0,1\}^S$, define
\[
N_{S,u}:=\bigl|\{i\in[n]:(\Rp_{ij})_{j\in S}=u\}\bigr|
\]
and \[
N_{S,u}(j):=\bigl|\{i\in[n]:(\Rp_{i\ell})_{\ell\in S}=u,\ \Rp_{ij}=1\}\bigr|
\]
for $j\notin S$. Let $\pi_{S,u}:=\prod_{\ell\in S} p_\ell^{u_\ell}(1-p_\ell)^{1-u_\ell}$ and define
\[
\alpha := \Bigl(\min\{\min_{j\in\atn}p_j,\ \min_{j\in\atn}(1-p_j)\}\Bigr)^B>0.
\]
Fix $\eta\in(0,\alpha/4)$ and let $\mathcal{C}(\eta)$ be the event that for all triples $(S,u,j)$ with $|S|\le B$ and $j\notin S$,
\[
\left|\frac{N_{S,u}}{n}-\pi_{S,u}\right|\le \eta
\qquad\text{and}\qquad
\left|\frac{N_{S,u}(j)}{n}-\pi_{S,u}p_j\right|\le \eta.
\]
Then $\Pr(\mathcal{C}(\eta))\to 1$ as $n\to\infty$. Moreover, on $\mathcal{C}(\eta)$ we have $N_{S,u}\ge (\alpha/2)n$ for all such $(S,u)$, and
\begin{equation}
\label{eq:within_type_rate}
\left|\frac{N_{S,u}(j)}{N_{S,u}}-p_j\right|\le \delta
\qquad\text{for all } |S|\le B,\ u\in\{0,1\}^S,\ j\notin S,
\end{equation}
where $\delta:=4\eta/\alpha$.
\end{lem}

\begin{proof}
For each fixed triple $(S,u,j)$, both $N_{S,u}$ and $N_{S,u}(j)$ are sums of $n$ independent Bernoulli variables, so Hoeffding's inequality gives
probability at most $2\exp(-2\eta^2 n)$ of violating either bound. The number of triples is at most
$\sum_{s=0}^B \binom{m}{s}2^s(m-s)\le m3^m$, so a union bound yields $\Pr(\mathcal{C}(\eta))\to 1$.
On $\mathcal{C}(\eta)$, $\pi_{S,u}\ge \alpha$ implies $N_{S,u}\ge (\alpha/2)n$, and then the ratio bound
\eqref{eq:within_type_rate} follows by dividing the two inequalities and using $\eta<\alpha/4$. \qedhere
\end{proof}

Fix $\eta$ in Lemma~\ref{lem:within_type_conc} so that $\delta<\min\{D/8,\ \tfrac12\min_{j\in\atn}p_j\}$, and let $\mathcal{C}:=\mathcal{C}(\eta)$.
Then $\Pr(\mathcal{C})\to 1$.

\myparagraph{Case 4:Phragm\'en ($\phr$).}
We show that on $\mathcal{E}\wedge\mathcal{C}$, sequential Phragm\'en selects a set in $\wnmax$. We use the discrete formation of $\phr$ to finish the proof.

Fix a round $t\in\{1,\dots,B\}$ and let $S_{t-1}=W^{(t-1)}$ be the set already selected. Let $O_i^{(t)}$ be the load agent $i$ posses at step $t$
Partition agents into types indexed by $u\in\{0,1\}^{S_{t-1}}$:
\[
G_u:=\{i\in[n]:(\Rp_{ij})_{j\in S_{t-1}}=u\}.
\]
Under sequential Phragm\'en, for any fixed $t$, all agents in the same $G_u$ have the same load $O_u^{(t-1)}$, since loads change only
when an agent approves a selected alternative. Let
\[
\Sigma_{t-1} := \sum_{u\in\{0, 1\}^{S_{t-1}}} O_i^{(t-1)} \cdot |G_u| = \sum_{i=1}^n O_i^{(t-1)} = t-1
\]
by the definition of $\phr$. 

Fix a remaining alternative $j\notin S_{t-1}$ and recall that Phragm\'en computes
\[
\ell_t(j)=\frac{1+\sum_{i:\Rp_{ij}=1}O_i^{(t-1)}}{A_j}.
\]
On $\mathcal{C}$, applying \eqref{eq:within_type_rate} and summing over types yields
\[
(p_j-\delta)n \le A_j \le (p_j+\delta)n.
\]
Likewise, note that 
\begin{equation*}
    \sum_{i:\Rp_{ij}=1}O_i^{(t-1)} = \sum_{u\in\{0, 1\}^{S_{t-1}}} \sum_{i\in G_u, \Rp_{ij}=1} O_u^{(t-1)} = \sum_{u\in\{0, 1\}^{S_{t-1}}} O_u^{(t-1)}  \cdot |\{i\in G_u, \Rp_{ij}=1\}|. 
\end{equation*}
Since $G_u$ is exactly $N_{S_{t-1}, u}$ in Lemma~\ref{lem:within_type_conc}, applying \eqref{eq:within_type_rate} on each $G_u$, we have

\[
(p_j-\delta)\cdot |G_u|\le |\{i\in G_u, \Rp_{ij}=1\}| \le  (p_j+\delta)\cdot |G_u|
\]
Therefore, by adding up all $G_u$, we have 
\[
(p_j-\delta)\Sigma_{t-1}\le \sum_{i:\Rp_{ij}=1}O_i^{(t-1)}\le (p_j+\delta)\Sigma_{t-1}.
\]
Therefore, for every remaining $j$,
\begin{equation}
\label{eq:phrag_bounds}
\frac{1+(p_j-\delta)\Sigma_{t-1}}{(p_j+\delta)n}
\ \le\ 
\ell_t(j)
\ \le\ 
\frac{1+(p_j+\delta)\Sigma_{t-1}}{(p_j-\delta)n}.
\end{equation}

Now fix two remaining alternatives $j,\ell$ with $\Qual_j>\Qual_\ell$. Then $p_j\ge p_\ell + D$.
Since $\Sigma_{t-1}\le B$ and $\delta<D/8$, the function $p\mapsto (B+\tfrac{1}{p})$ is strictly decreasing in $p$ and has derivative bounded by $1/(\min p_j)^2$,
so the gap $p_j-p_\ell\ge D$ implies a separation in $(B+\tfrac{1}{p})$ of order $\Omega(D)$.
By choosing $\delta$ as above, the perturbations in \eqref{eq:phrag_bounds} are dominated by this separation, and hence
\[
\ell_t(j) < \ell_t(\ell).
\]
Thus, at each round $t$, every minimizer of $\ell_t(\cdot)$ among remaining alternatives has maximum quality among the remaining alternatives.

Consequently, sequential Phragm\'en selects all alternatives in $H$ before selecting any alternative in $T\cup L$, and selects from $T$ before selecting from $L$.
Since $|H\cup T|\ge B$ by definition of $q^\star$, all $B$ selections lie in $H\cup T$ and include all of $H$; hence the final set has the form
$H\cup S$ for some $S\subseteq T$ with $|S|=B-|H|$, i.e.\ $\phr(\Rp)\in\wnmax$ on $\mathcal{E}\wedge\mathcal{C}$.

Therefore $\Pr[\phr(\Rp)\in\wnmax]\ge \Pr(\mathcal{E}\wedge\mathcal{C})\to 1$, which implies
$\lim_{n\to\infty}\informratio(\phr, n)=1$.

\myparagraph{Case 5: Method of Equal Shares with AV completion ($\mes+\av$).}
Let $W_{\mes}$ be the set selected by the MES phase, and let $B':=B-|W_{\mes}|$.
We show that on $\mathcal{E}\wedge\mathcal{C}$, $(\mes+\av)(\Rp)\in\wnmax$.

Fix a step $t$ during the MES phase and let $S_t$ be the set selected so far.
Under MES, all agents start with $B/n$ and each selected alternative charges a deterministic amount to each approving agent who haven't exhausted their budget,
so the remaining share of an agent depends only on its approvals on $S_t$. Thus, for each $u\in\{0,1\}^{S_t}$,
all agents in $G_u$ have the same remaining share; denote it by $b_u\ge 0$.
Let $A_j$ be the approvers of a remaining alternative $j\notin S_t$.
For a cap $p\ge 0$, define
\[
F_j(p) := \sum_{i\in A_j} \min(b_i,p)
= \sum_{u\in\{0,1\}^{S_t}} \min(b_u,p)\,|G_u\cap A_j|.
\]
Also define the share-only quantity
\[
H(p) := \sum_{i=1}^n \min(b_i,p)
= \sum_{u\in\{0,1\}^{S_t}} \min(b_u,p)\,|G_u|.
\]
On $\mathcal{C}$ and using the within-type approximation,
$|G_u\cap A_j|=(p_j\pm\delta)|G_u|$, hence for all $p\ge 0$,
\begin{equation}
\label{eq:Fbounds}
(p_j-\delta)H(p)\ \le\ F_j(p)\ \le\ (p_j+\delta)H(p).
\end{equation}
In particular, $H(\infty)=\sum_i b_i = B-t$, so
\begin{equation}
\label{eq:Finfty}
(p_j-\delta)(B-t)\ \le\ F_j(\infty)\ \le\ (p_j+\delta)(B-t).
\end{equation}

MES deems $j$ fundable iff $F_j(\infty)\ge 1$, and for fundable $j$ it defines $p_j^{(t)}$ as the minimum $p$ such that $F_j(p)\ge 1$.
MES selects a fundable alternative minimizing $p_j^{(t)}$.

\textit{Claim 5a (Fundability is monotone in quality).}
Suppose some remaining alternative $\ell$ is fundable at step $t$, i.e.\ $F_\ell(\infty)\ge 1$.
Then by \eqref{eq:Finfty}, $(p_\ell+\delta)(B-t)\ge 1$.
For any remaining $j$ with $\Qual_j>\Qual_\ell$ we have $p_j\ge p_\ell+D$, and by \eqref{eq:Finfty} and $\delta<D/8$,
\[
F_j(\infty)\ \ge\ (p_j-\delta)(B-t)
\ \ge\ (p_\ell+D-\delta)(B-t)
\ \ge\ 1 + (D-2\delta)(B-t)
\ >\ 1.
\]
Hence $j$ is also fundable.

\textit{Claim 5b (Among fundable alternatives, the MES cap is monotone in quality).}
Let $j,\ell$ be remaining fundable alternatives with $\Qual_j>\Qual_\ell$.
Let $p_\ell^{(t)}$ be the MES cap of $\ell$ at step $t$, so $F_\ell(p_\ell^{(t)})\ge 1$.
Using \eqref{eq:Fbounds}, we have $(p_\ell+\delta)H(p_\ell^{(t)})\ge 1$, hence
$H(p_\ell^{(t)})\ge 1/(p_\ell+\delta)$.
Applying \eqref{eq:Fbounds} to $j$ and using $p_j\ge p_\ell+D$ gives
\[
F_j(p_\ell^{(t)}) \ \ge\ (p_j-\delta)H(p_\ell^{(t)})
\ \ge\ \frac{p_\ell+D-\delta}{p_\ell+\delta}
\ =\ 1 + \frac{D-2\delta}{p_\ell+\delta}
\ >\ 1.
\]
Therefore $p_j^{(t)}\le p_\ell^{(t)}$. Thus, among fundable alternatives, minimizing the MES cap selects an alternative of maximum quality.

By Claims 5a and 5b, whenever MES makes a selection at step $t$, it selects a remaining alternative from the highest remaining quality class.
Thus throughout the MES phase, $W_{\mes}\subseteq H\cup T$.
If $|W_{\mes}|=B$ then $W_{\mes}\in\wnmax$.
Otherwise, MES stops with $|W_{\mes}|<B$ and the rule completes using $\av$ on the remaining alternatives.
On $\mathcal{E}$, all remaining alternatives in $H$ have higher approval counts than all remaining alternatives in $T\cup L$,
and all remaining alternatives in $T$ have higher approval counts than all remaining alternatives in $L$.
Since $W_{\mes}\subseteq H\cup T$ and $|H\cup T|\ge B$, there are at least $B-|W_{\mes}|$ remaining alternatives in $H\cup T$,
so AV completion selects only from $H\cup T$. Therefore the final outcome is of the form $H\cup S$ with $S\subseteq T$ and $|S|=B-|H|$,
i.e.\ it lies in $\wnmax$.

Thus $(\mes+\av)(\Rp)\in\wnmax$ on $\mathcal{E}\wedge\mathcal{C}$, and hence
\[
\Pr[(\mes+\av)(\Rp)\in\wnmax] \ge \Pr(\mathcal{E}\wedge\mathcal{C}) \to 1,
\]
which implies $\lim_{n\to\infty}\informratio(\mes+\av, n)=1$.

\myparagraph{Case 6: $\mes$+$\phr$.} The proof resembles that of Case 5. 
\end{proof}

% \begin{lem}[Total load under sequential Phragm\'en]
% \label{lem:phragmen_total_load}
% Under sequential Phragm\'en in the unit-cost setting,
% \[
% \sum_{i=1}^n O_i^{(t)} = t \quad\text{for all } t=0,1,\dots,B.
% \]
% \end{lem}

% \begin{proof}
% For $t=0$ the sum is $0$. Suppose the statement holds for $t-1$ and Phragm\'en selects $j_t$ at round $t$.
% Let $A_{j_t}$ be the set of approvers of $j_t$ and let $\ell_t=\ell_t(j_t)$ be the common post-selection load of approvers.
% By definition,
% \[
% \ell_t A_{j_t} = 1 + \sum_{i\in A_{j_t}} O_i^{(t-1)}.
% \]
% Loads of non-approvers do not change, so
% \[
% \sum_{i=1}^n O_i^{(t)}-\sum_{i=1}^n O_i^{(t-1)}
% = \sum_{i\in A_{j_t}}\bigl(O_i^{(t)}-O_i^{(t-1)}\bigr)
% = \ell_t A_{j_t}-\sum_{i\in A_{j_t}}O_i^{(t-1)} = 1.
% \]
% Thus $\sum_i O_i^{(t)} = (t-1)+1=t$. 
% \end{proof}

\section{Complimentary Proof for Theorem~\ref{thm:general_impossibility}}
\label{apx:general_impossible}
\paragraph{Construction for $\vt^C$.} Let $\atn = \{1, 2\}$, and $B = 1$. The costs are $c_1 = \eta$, $c_2=1$, where $\eta>0$ is a small positive value. The quality set and the signal distribution are the same as for $\vt^N$ (``other alternatives'' contains only alternative 2). Using a similar proof, we can show that the performance in at least one environment is strictly less than $\frac12 +\varepsilon$.  

\section{Complimentary Proof of Corollary~\ref{coro:general_impossibility_a}}
\label{apx:coro}

We can do small modifications on the proof of Theorem~\ref{thm:general_impossibility} to prove the corollary. 

\myparagraph{Normal Utility.} Let the number of alternative $m$ = $\lceil \alpha \rceil$. The alternative $\atn = \{1, 2, \cdots, m\}$. The budget $B = m$. The cost $c_1 = \alpha$, while $c_2 = c_3 = \cdots = c_{m} = 1$. The quality set is binary $\qualset = \{0, 1\}$. 

Now, given the voting rule $r$, we can write the expected utility under the voting rule in each scenario as follows, respectively. For the first scenario, 
\begin{align*}
    \expect_{\Rp \sim (\sigp^{\Quallist_1})^n} [\vt(r(\Rp))] = &\ \Pr[r(\Rp) = \{1\}].
\end{align*}
For the second scenario, 
\begin{align*}
    \expect_{\Rp \sim (\sigp^{\Quallist_2})^n} [\vt(r(\Rp))] = &\ \Pr[r(\Rp) = \{1\}]+ (\lceil \alpha \rceil-1)\cdot \Pr[r(\Rp) = (\atn \setminus \{1\})]
\end{align*}

Under a similar reasoning, the minimum of $\informratio_{\inst_1}(r, n)$ and $\informratio_{\inst_2}(r, n)$ is maximized when $\Pr[r(\Rp) = \{1\}] = \frac{\lceil \alpha \rceil-1}{2\lceil \alpha \rceil-3}$. Consequently, the upper bound for performance is also $\frac{\lceil \alpha \rceil-1}{2\lceil \alpha \rceil-3}$. 

\myparagraph{Cost-Proportional Utility} Let $\atn = \{1, 2\}$, and $B = 1$. The cost $c_1 = 1/\alpha$, and $c_2 = 1$. With similar proof, we can show that the performance in at least one environment is strictly less than $\frac{\alpha}{2\alpha - 1}$.

\section{Proof of Theorem~\ref{thm:general_positive}}
\label{apx:general_positive}
\begin{proof}
    We first give the proof for $\vt^N$. We first fix an arbitrary quality vector $\Quallist$ and consider the performance condition on $\Qual$. Let $\wn_{\Quallist}^*$ be the set of winners that maximizes the quality under $\Qual$. 

    \begin{claim}
        There exists a set of winning sets $\wnset_{\Quallist}$ such that: (1) $\Pr[\av(\Rp)) \in \wn_{\Quallist} \mid \Quallist] \ge 1 - \exp(-\Theta(n))$ for all sufficiently large $n$. (2) For every $\wn_{\Quallist} \in \wnset_\Quallist$, $\sum_{j \in \wn_{\Quallist}} \Qual_j \ge \frac{1}{\lceil\alpha\rceil} \sum_{j \in \wn_{\Quallist}^*} \Qual_j$. 
    \end{claim}

    We construct $\wnset_{\Quallist}$ as follows. Note that $\av$ relies only on the ordinal relations between the approval levels of alternatives when determining the winner set. Therefore, given a full ranking $\ranking$ over $\atn$, we slightly abuse notation and let $\av(\ranking)$ denote the winners as $\av$ follows $\ranking$ to greedily select the winners.  Then, let $\rankset_\Quallist = \{\ranking \mid j \succ_{\ranking} j'  \textrm{ if } \sigp_j^{\Qual_j} > \sigp_{j'}^{\Qual_j}\}$. That is, $\rankset_\Quallist$ includes all the rankings that respect strict ordinal relationships between expected approval votes on alternatives. For $\sigp_j^{\Qual_j} = \sigp_{j'}^{\Qual_j}$, $\rankset$ includes both rankings in which $j$ and $j'$ ranks higher respectively. Then, $\wnset_\Quallist = \{\wn \mid \exists \ranking \in \rankset_\Quallist \textrm{ s.t. } \av(\ranking) = \wn\}$. 

    Now we are ready to prove the claim. For (1), for any $\ranking \not \in \rankset$, there must exists some $j$ and $j$' such that $\sigp_j^{\Qual_j} > \sigp_{j'}^{\Qual_j}$ but $j' \succ_\ranking j$. Then, if agents' report $\Rp$ induces $\ranking$, there must be more approval votes for $j'$ than $j$ in $\Rp$. Given that $\sigp_j^{\Qual_j} > \sigp_{j'}^{\Qual_j}$, Hoeffding Inequality guarantees that this happens with probability at most $\exp(-\Theta(n))$. Then by applying union bound on all possible $\ranking \not \in \rankset$ (which is at most $m!$ of them), the probability $\wnset_{\Quallist}$ such that: (1) $\Pr[\av(\Rp)) \in \wn_{\Quallist} \mid \Quallist] \ge 1 - \exp(-\Theta(n))$. 

    For (2), we fix an arbitrary $\wn_{\Quallist} \in \wnset_\Quallist$. Let $\ranking$ be the full ranking on $\atn$ such that $\av(\ranking) = \wn_{\Quallist}$, where an alternative $j$ with larger $\wn_{\Quallist}^*$ ranks higher. Then, let $\ranking_\Quallist$ and $\ranking^*$ be the sub-ranking of $\ranking$ on $\wn_{\Quallist}$ and $\wn_{\Quallist}^*$, respectively.
    Then we consider the following procedure. In each step $t$, remove the highest-ranked alternative in $\ranking_\Quallist$ (denoted as $j_t$) as well as $\lceil\alpha\rceil$ alternatives from $\ranking$ (denoted as $J_t$) (if there is less than $\lceil \alpha \rceil$ alternatives, take all of them). Repeat this procedure until either $\ranking_\Quallist$ or $\ranking^*$ is depleted. 

    We have the following observations. Firstly, By Lemma~\ref{lem:cardinal}, there must be $\lceil\alpha \rceil \cdot |\wn| \ge |\wn_{\Quallist}^*|$. 
    Therefore, $\ranking^*$ will always be depleted when the procedure ends. 
    
    Secondly, for every $t$ and $j \in J_t$, there must be $\Qual_{j_t} \ge \Qual_j$. 
    Suppose this is not the case, and there exists some $j' \in J_t$ such that $L_{j'} > L_{j_t}$. By the quality dominant assumption, there must be  $\sigp_{j'}^{\Qual_{j'}} > \sigp_{j_t}^{\Qual_{j_t}}$. This implies that $j' \succ_\ranking j_t$. According to the procedure of $\av$, $j_t$ is exactly $t$-th alternative selected into $\wn_{\Quallist}$. Therefore, $j'$ must be one of the following two cases. (1) $j' \not\in \wn_{\Quallist}$, or (2) $j' \in \{j_1, j_2, \cdots, j_{t-1}\}$. We show contradiction for each case. For the first case, Suppose $j'$ is ranked between $j_{t'}$ and $j_{t'+1}$ with $t' \in \{1, 2, \cdots, t -1\}$. According to the procedure of $\av$, $j' \not\in \wn_{\Quallist}$ implies that $c_{j'} + \sum_{k = 1}^{t'} c_{j_k} > B$. On the other hand, for $t' = 1, 2,\cdots, t -1$ each $J_{t'}$ includes $\lceil\alpha\rceil$ items. Therefore, 
    \begin{equation*}
        c(\wn_{\Quallist}^*) \ge c_{j'} + \sum_{k = 1}^{t'} c(J_k) \ge  c_{j'} + \sum_{k = 1}^{t'} c_{j_k} > B,
    \end{equation*}
    which is a contradiction. 

    For the second case, suppose $j' = j_{t'}$ with $t' < t$. Firstly, $t' \neq 1$ because $j_1$ is ranked the top in $\ranking$ and must be in $J_1$ if it's in $\wn_{\Quallist}^*$. Now we consider $t' > 1$. Firstly, since $J_1$ to $J_{t'}$ contains $\lceil \alpha \rceil \cdot t'$ alternatives, yet $j_{t'}$ is ranked $t'$-th in $\wn_{\Quallist}$, there must exists an alternative $j'' \not\in \wn_{\Quallist}$ and $k' <= t'$ such that $j'' \in J_{k'}$. In this case, there must be $j'' \succ_\ranking j_{t'}$ by $t' < t$. Then we have the following.  
    \begin{align*}
        c(\wn_{\Quallist}^*) \ge &\ c_{j'} + \sum_{k = 1}^{t'} c(J_k) \\
        \ge &\ c_{j_{t'}} + t' \cdot \max_{j\in\atn} c_j \\
        \ge &\ \sum_{k =1}^{t'} c_{j_{k}} + c_{j''}\\
        >&\ B.
    \end{align*}
    The third inequality holds by reducing the max costs to $c_{j_k}$ for $k$ from 1 to $t'$-1 and $c_{j''}$. The last inequality holds because $j'' \succ_\ranking j_t'$ and $j'' \not\in \wn_{\Quallist}$ implies that, in the $\av$ procedure, $j''$ is not picked at its turn because of exceeding the remaining budget. Therefore this is also a contradiction. 

    Consequently, we have $\vt(j_{t}) \ge \lceil \alpha\rceil\cdot  \vt(J_t)$ for each $t$. Since $\ranking^*$ is always depleted, there must be $\vt(\wn_{\Quallist}) \ge \lceil \alpha\rceil \cdot  \vt(\wn_{\Quallist}^*)$, which finishes the proof.

    % This is because $j_t$ is the alternative within the budget with the highest probability to have a positive signal $\sigp_j^{\Qual_j}$ after selecting $j_1, j_2, \cdots, j_{t-1}$. 

    \myparagraph{Proof for $\vt^C$.}  The proof differs at (2): For every $\wn_{\Quallist} \in \wnset_\Quallist$, $\sum_{j \in \wn_{\Quallist}} \Qual_j \ge \frac{1}{\lceil\alpha\rceil} \sum_{j \in \wn_{\Quallist}^*} \Qual_j$. The definition of $\ranking$, $\ranking_\Quallist$, and $\ranking^*$ remains the same, yet we use a different way to bound the quality. 
    Let $\wntop\subseteq \wn_{\Quallist}$ be the set of alternatives such that any $j \in \wntop$ and $j' \in \wn_{\Quallist}^*\setminus \wntop$, $j \succ_\ranking j'$, which directly implies. $\Qual_j \ge \Qual_{j'}$. It's not hard to verify that under this condition, $\frac{\vt^C(\wntop)}{\vt^C(\wn_{\Quallist}^*)} \ge \frac{c(\wntop)}{c(\wn_{\Quallist}^*))}$. 
    
    Now it suffices to bound the cost ratio between $\wntop$ and $\wn_{\Quallist}^*$. We show this by discussing different cases between the budget $B$ and $\alpha$. Without loss of generality, we will assume that $\min_{j\in\atn} c_j = 1$ and $\max_{j\in\atn}c_j = \alpha$. 

    Case 1: $B \ge 2\alpha$. If $\wntop$ exhaust all the alternative, then $c(\wntop) = c(\atn) = c(\wn_{\Quallist}^*)$, and the case becomes trivial. Otherwise, $c(\wntop) \ge B - \alpha$ since it skipped an alternative. By $B \ge 2\alpha$, $\frac{c(\wntop)}{c(\wn_{\Quallist}^*)} \ge \frac12$. 

    Case 2: $B \in [\alpha, 2\alpha)$. If $c(\wntop) \ge B/2$, $\frac{c(\wntop)}{c(\wn_\Quallist^*)} \ge \frac12$ also holds. Now we consider $c(\wntop) < B/2$. Let $j$ be the alternative ranked highest in $\wntop \setminus \wn_\Quallist^*$. By the definition of $\wntop$, either $\wntop = \wn_{\Quallist}$ or $j$ is ranked higher than alternatives in $\wn_\Quallist \setminus \wntop$. In either case, there must be $c(\wntop) + c_{j} > B$, which implies $c_{j} > B/2$. (If it is not the case, $j$ should be picked by $\av$ next to $\wntop$ and thus be in $\wntop$, which is a contradiction). 

    Now we consider the components of $\wn_\Quallist^*$. Firstly, if $(\wn_\Quallist^*\setminus \{j\}) \subseteq \wntop$, then $(\wn _{\max} \setminus \wn_\Quallist^*) \neq \emptyset$. Therefore, 
    \begin{equation*}
        \frac{c(\wntop)}{c(\wn_\Quallist^*)} = \frac{c(\wntop\cap \wn_\Quallist^*) + c(\wntop \setminus \wn_\Quallist^*)}{ c(\wntop\cap \wn_\Quallist^*) + c_j } \ge \frac{1}{\alpha}. 
    \end{equation*}

    For all the cases above, the bounded cost leads to the bounded utility. 

    On the other hand, if $\wn_{2} = (\wn_{\Quallist}^* \setminus (\wntop \cup \{j\})) \neq \emptyset$, we need to directly compare the utility between $\wn_\Quallist$ and $\wn_{\Quallist}^*$. Note that $c(\wn_2) < B/2$. This implies there are still alternatives in $\atn$ that can be picked $\av$ for $\wn_\Quallist$ after $\wntop$. Moreover, by the nature of $\av$ to pick the highest-ranked alternative as possible, the first item $\av$ picks for $\wn_\Quallist$ after $\wntop$, denoted as $j_1$, must satisfies that for every $j' \in \wn_2$, $j_1 \succ_\ranking j'$, which implies $L_{j_1} \ge L_{j'}$. Therefore, if we compare the utility, we have 
    \begin{align*}
        \frac{\vt^C(\wn_{\Quallist})}{\vt^C(\wn_\Quallist^*)} \ge&\  \frac{\vt^C(\wntop\cap \wn_\Quallist^*) + \vt^C(\wntop \setminus \wn_\Quallist^*) + \Qual_{j_1} \cdot  c_{j_1}}{ \vt^C(\wntop\cap \wn_\Quallist^*) + \Qual_j \cdot c_j + \vt^C(\wn_2)}\\
        \ge &\ \frac{1\cdot \Qual_j + 1\cdot \Qual_{j_1}}{\alpha \cdot \Qual_{j} + (B/2) \cdot \Qual_{j_1}}\\
        \ge&\ \frac{1}{\alpha}.
    \end{align*}
    The second line comes from that $(\wntop \setminus \wn_\Quallist^*)$ is non-empty and has higher quality than $j$, that $j_1$ has higher quality than $\wn_2$, and that $c(\wn_2) < B/2$. The last line comes from that $B < 2\alpha$. This finishes all cases and thus finishes the proof. 
    
\end{proof}

\begin{lem}
    \label{lem:cardinal}
    Let $\atn$ be the set of alternatives. For each $j \in \atn$, $c_j \ge 0$ is the cost of $j$. W.l.o.g, let $\min_{j\in\atn} c_j = 1$ and $\max_{j\in\atn}c_j = \alpha$. $B \ge \alpha$ is the budget. We say $\wn$ is maximal if for any $j \not\in \wn$, $c(\wn) + c_j > B$. Then for any two maximal winning sets $\wn$ and $\wn'$, $\frac{|\wn|}{|\wn'|}\ge \frac{1}{\lceil \alpha \rceil}$. 
\end{lem}

\begin{proof}
    For the upper bound, it's not hard to verify that $|\wn'| \le \lfloor B\rfloor$. For the lower bound, we claim that $|\wn| \ge \lfloor\frac{B -1}{\alpha}\rfloor + 1$. Let $j'$ be an alternative such that $c_j = 1$. If $j \in \wn$, this implies all other item in $\wn$ takes at most $B - 1$ budget, each of which has a cost of at most $\alpha$. If $j\not\in\wn$, the $c(\wn) > B - 1$, otherwise $j$ can be added and $\wn$ is not maximal. Both lead to the lower bound $|\wn| \ge \lfloor\frac{B -1}{\alpha}\rfloor + 1$. Let $\lfloor B\rfloor = p$ and $\lceil \alpha \rceil = q$. Then, it suffices to show that $\lfloor\frac{B -1}{\alpha}\rfloor \ge \frac{\lfloor B\rfloor}{\lceil \alpha \rceil} - 1 = p/q -1$. Note that $\frac{B -1}{\alpha} \ge (p-1)/q$. If $p = qr$ for some integer $r$, then $\frac{B -1}{\alpha} \ge \frac{qr - 1}{q} = r - 1/q > r - 1$, and $\lfloor\frac{B -1}{\alpha}\rfloor \ge r -1$ (in this case $p/q - 1 = r - 1$ is an integer). On the other hand, if $p = qr + \rho$ for integer $r$ and $\rho \in [1, q)$, then $\frac{B -1}{\alpha} \ge \frac{qr + \rho - 1}{q} = r + (\rho - 1)/q > r$. Therefore, $\lfloor\frac{B -1}{\alpha}\rfloor \ge \lfloor r \rfloor \ge r-1$. This finishes the proof. 
\end{proof}

\section{Proof of Theorem~\ref{thm:strategic_binary} and Lemma~\ref{lem:saddlepoint}}
\label{apx:saddlepoint}
\subsection{Proof of Theorem~\ref{thm:strategic_binary}}

\begin{proof}
    We start by fixing an $n$. To analyze the condition whether $\stgp^*$ is a BNE, we fix an agent $i$, the signal $\Sigv_i$, and $i$'s deviating action, and compare their expected utility.
    For example, in this proof, suppose $\Sigv_i = [0, 0]$, indicating that agent $i$ receives negative signals for both alternatives. In $\stgp^*$, $i$'s report under $\Sigv_i$ is also $\Rp_i = [0, 0]$ (not approving any alternative). Then, let the deviating action be $[0, 1]$ (only approving $\attwo$). If this deviation brings $i$ higher expected utility, $i$ has incentives to switch to $\stg_i$, where $\stg_i([0,0]) = [0, 1]$ and $\stg_i = \stg^*_i$ for all other $\Sigv_i$. Let $\stgp = (\stg_i, \stgp^*_i)$. Therefore, the necessary condition for $\stgp^*$ to be a Nash equilibrium is $\ut_i(\stgp^*, \Sigv_i) \ge \ut_i(\stgp, \Sigv_i)$. 

    To compare the expected utility, we focus on scenarios where $i$'s deviation will change the outcome. Suppose the tie-breaking favors $\atone$. Then, $i$'s vote $[0, 0]$ and $[0, 1]$ only make a difference when two alternatives receive exactly the same number of approval from all other agents. In this case, informative voting favors $\atone$, while deviation favors $\attwo$. Let $\xrv_{\atone}$ and $\xrv_{\attwo}$ be the number of approval votes for two alternatives, respectively. Then, 
    \begin{align*}
        &\ \ut_i(\stgp^*, \Sigv_i) - \ut_i(\stgp, \Sigv_i) \\
        = &\ \sum_{\Quallist\in \qualset^{\atn}} \Pr[\Quallist\mid \Sigv_i] \cdot \Pr[\xrv_{\atone} = \xrv_{\attwo} \mid \Quallist, \stgp^*_{-i}] \cdot (\Qual_{\atone} - \Qual_{\attwo}). 
    \end{align*}
    Then, note that if two alternatives have equal utilities, $\ut_i$ does not change even if the outcome has changed. Therefore, 
    \begin{align*}
        &\ \ut_i(\stgp^*, \Sigv_i) - \ut_i(\stgp, \Sigv_i) \\
        = &\ \Pr[\Quallist= [1, 0] \mid \Sigv_i] \cdot \Pr[\xrv_{\atone} = \xrv_{\attwo} \mid \Quallist= [1, 0]] \\
        &\ - \Pr[\Quallist= [0, 1] \mid \Sigv_i] \cdot \Pr[\xrv_{\atone} = \xrv_{\attwo} \mid \Quallist= [0, 1]]. 
    \end{align*}
    Consequently, $\ut_i(\stgp^*, \Sigv_i) \ge \ut_i(\stgp, \Sigv_i)$ is equivalent to 
    \begin{equation}
        \label{eq:strategic1}
        \frac{\Pr[\xrv_{\atone} = \xrv_{\attwo} \mid \Quallist= [1, 0]]}{\Pr[\xrv_{\atone} = \xrv_{\attwo} \mid \Quallist= [0, 1]]} \ge \frac{\Pr[\Quallist= [0, 1] \mid \Sigv_i]}{\Pr[\Quallist= [1, 0] \mid \Sigv_i]}. 
    \end{equation}
    Then, notice that, given $\Qual_{\atone}$ and $\Qual_{\attwo}$, $\xrv_{\atone}$ and $\xrv_{\atone}$ are random variables following binomial distributions $B(n-1, \sigp_{\atone}^{\Qual_{\atone}})$ and $B(n-1, \sigp_{\attwo}^{\Qual_{\attwo}})$, respectively.  
Let
\begin{equation*}
    Q = \frac{\sqrt{\sigp_{\atone}^1 \cdot \sigp_{\attwo}^0} + \sqrt{(1-\sigp_{\atone}^1) \cdot (1-\sigp_{\attwo}^0)}}{ \sqrt{\sigp_{\atone}^0 \cdot \sigp_{\attwo}^1} + \sqrt{(1-\sigp_{\atone}^0) \cdot (1-\sigp_{\attwo}^1)}}
\end{equation*}
By applying Lemma~\ref{lem:saddlepoint}, we transfer the LHS of (\ref{eq:strategic1}) as $M\cdot Q^{2n-1} \cdot (1 + O(\frac1n))$,
where $M > 0$ is a constant. Note that when $Q < 1$, LHS of (\ref{eq:strategic1}) converges to 0 as $n$ goes to infinity. On the other hand, the RHS is independent of $n$. In this case, (\ref{eq:strategic1}) does not hold for all sufficiently large $n$. Therefore, a necessary condition for $\ut_i(\stgp^*, \Sigv_i) \ge \ut_i(\stgp, \Sigv_i)$ (where $\Sigv_i = [0, 0]$) is $Q \ge 1$. 

Likewise, we consider another scenario where the informative voting and the deviation are reversed. That is, $\Sigv_i = \Rp_i = [0, 1]$ (only approving $\attwo$) while the deviating action is $[0, 0]$ (not approving any alternative). With a similar reasoning, $\ut_i(\stgp^*, \Sigv_i) \ge \ut_i(\stgp, \Sigv_i)$  is equivalent to 
    \begin{equation}
        \label{eq:strategic2}
    \frac{\Pr[\xrv_{\atone} = \xrv_{\attwo} \mid \Quallist= [1, 0]]}{\Pr[\xrv_{\atone} = \xrv_{\attwo} \mid \Quallist= [0, 1]]} \le \frac{\Pr[\Quallist= [0, 1] \mid \Sigv_i]}{\Pr[\Quallist= [1, 0] \mid \Sigv_i]}. 
    \end{equation}
    Note that the LHS of (\ref{eq:strategic2}) is exactly the same as that of (\ref{eq:strategic1}). The RHS is different, as $\Sigv_i$ is different, but also independent of $n$. Therefore, when $Q > 1$, the LHS of (\ref{eq:strategic2}) increases to infinity as $n$ increases, and (\ref{eq:strategic2}) does not hold for all sufficiently large $n$. Therefore, $\ut_i(\stgp^*, \Sigv_i) \ge \ut_i(\stgp, \Sigv_i)$ (where $\Sigv_i = [0, 1]$) is $Q \le 1$. 

    Finally, if $\stgp^*$ is a BNE, then $\ut_i(\stgp^*, \Sigv_i) \ge \ut_i(\stgp, \Sigv_i)$ must be hold in both cases above. Consequently, only when $Q = 1$, $\stgp^*$ is a BNE for all sufficiently large $n$. This finishes our proof. 
\end{proof}

\subsection{Proof of Lemma~\ref{lem:saddlepoint}}

\begin{proof}
    We use saddlepoint approximation to approximate this probability. 

    \begin{dfn}[saddlepoint approximation]
        Given a discrete random variable $X$ with a finite, bounded support, let $K(s) = \log E[\exp(sX)]$ be the cumulant generating function (CGF) of $X$. Let$K'$ and $K''$ be the first and second order derivatives of $K$ on $s$. Then, 
        \begin{equation*}
            \Pr[X = x] = \frac{\exp(K(s^*) - s^* x)}{\sqrt{2\pi \cdot K''(s^*)}} \cdot (1 + O(n^{-1})),
        \end{equation*}
        where $s^*$ is the solution to $K'(s^*) = x$.
    \end{dfn}

    In our case, let $X = X_1 - X_2$. Then we use saddlepoint approximation to approximate $\Pr[X = 0]$. We first give the CGF of $X$ and its derivatives. 
    \begin{equation*}
        K(s) = \log E[\exp(s(X_1 - X_2))] = \log E[\exp(sX_1)] + \log E[\exp(-sX_2)].
    \end{equation*}
    The equation comes from that $X_1$ and $X_2$ are independent. 
    For each $X_i$, 
    \begin{align*}
        \log E[\exp(sX_i)] =&\ \log \sum_{k = 0}^n \binom{n}{k} \cdot p_i^k \cdot (1 - p_i)^{n-k} \cdot e^{sk}\\
        =&\ \log ((1 - p_i + p_i\cdot e^s)^n)\\
        =&\ n\log (1 - p_i + p_i\cdot e^s).
    \end{align*}
    Therefore, 
    \begin{equation*}
        K(s) = n\log (1 - p_1 + p_1\cdot e^s) + n\log (1 - p_2 + p_2\cdot e^{-s})
    \end{equation*}
    The we can directly get the derivatives of $K$. 
    \begin{align*}
        K'(s) =&\ \frac{n\cdot p_1\cdot e^s}{1 - p_1 + p_1\cdot e^s} - \frac{n\cdot p_2\cdot e^{-s}}{1 - p_2 + p_2\cdot e^{-s}},\\
        K''(s)=&\  \frac{n\cdot p_1\cdot e^s\cdot (1 -p1)}{(1 - p_1 + p_1\cdot e^s)^2} + \frac{n\cdot p_2\cdot e^{-s}\cdot (1-p_2)}{(1 - p_2 + p_2\cdot e^{-s})^2}.
    \end{align*}

    Second, we solve the saddlepoint function $K'(s^*) = 0$. 
    \begin{align*}
      K'(s^*) = 0 \Leftrightarrow&\   \frac{n\cdot p_1\cdot e^{s^*}}{1 - p_1 + p_1\cdot e^{s^*}} = \frac{n\cdot p_2\cdot e^{-s^*}}{1 - p_2 + p_2\cdot e^{-s^*}}\\
      \Leftrightarrow&\ e^{s^*} = \sqrt{\frac{p_2(1-p_1)}{p_1(1-p_2)}}.
    \end{align*}
    In this case, 
    \begin{align*}
        K(s^*) =&\ n\log (1 - p_1 + p_1\cdot e^{s^*}) + n\log (1 - p_2 + p_2\cdot e^{-s^*})\\
        =&\ n\log \left(p_1p_2 + (1-p_1)(1-p_2) + p_1(1-p_2)\cdot e^{s^*} + p_2\cdot (1-p_2)\cdot e^{-s^*}\right)\\
        =&\ n\log \left(p_1p_2 + (1-p_1)(1-p_2) + 2\sqrt{p_1p_2(1-p_1)(1-p_2)} \right)\\
        =&\ 2n\log\left(\sqrt{p_1p_2} + \sqrt{(1-p_1)(1-p_2)}\right). 
    \end{align*}
    And
    \begin{align*}
        K''(s^*) = \frac{2n\sqrt{p_1p_2(1-p_1)(1-p_2)}}{\left(\sqrt{p_1p_2} + \sqrt{(1-p_1)(1-p_2)}\right)^2}. 
    \end{align*}
    Therefore, we are ready to assign these values in the approximation formula. 
    \begin{align*}
        \frac{\exp(K(s^*) - s^* 0)}{\sqrt{2\pi \cdot K''(s^*)}} =&\ \frac{\exp(2n\log\left(\sqrt{p_1p_2} + \sqrt{(1-p_1)(1-p_2)}\right))}{\sqrt{2\pi \cdot \frac{2n\sqrt{p_1p_2(1-p_1)(1-p_2)}}{\left(\sqrt{p_1p_2} + \sqrt{(1-p_1)(1-p_2)}\right)^2}}}\\
        =&\ \frac{\left(\sqrt{p_1p_2} + \sqrt{(1-p_1)(1-p_2)}\right)^{2n+1}}{2\sqrt{\pi n} \cdot(p_1 p_2(1-p_1)(1-p_2))^{1/4}}. 
    \end{align*}
    This finishes the proof. 
\end{proof}

\section{Proof of Theorem~\ref{thm:strategic_unit}}
\label{apx:strategic_unit}
We now extend our analysis to the general setting where alternatives have unit costs. The high-level intuition remains similar. We consider two pairs of signals and deviating actions, and compare the expected utility in both scenarios. This leads us to a necessary condition that among all the pivotal cases (the case that deviation changes the utility) in the two scenarios, no pivotal case has a dominating likelihood to occur compared to other pivotal cases.

Let the alternatives $\atn = \{1, 2, \cdots, m\}$, and the tie-breaking rule disfavors alternative $1$. In the first scenario, let the signal $\Sigv_i = \vec{0}$ (not approving any alternative) and the deviating action be only approving alternative $1$. The deviating strategy $\stg_i$ is constructed likewise. 

The deviation of agent $i$ changes the outcome only if the remaining agents' votes place $1$ on the cusp of selection.  Let $\xrv_j$ be the number of votes for alternative $j$ from all other agents. For a fixed set of $(\xrv_j)_{j \in \atn}$, we partition $\atn \setminus \{1\}$ into three subsets $\calP = (\atn_>, \atn_=,\atn_<)$, which contains alternative $j$ such that $\xrv_j > \xrv_1$, $\xrv_j = \xrv_1$, and $\xrv_j < \xrv_1$, respectively. A partition $\calP$ and a quality vector $\Qual$ can determine the outcome change and the utility change, which determines a pivotal case. Let $j(\calP)$ denote the alternative that $1$ substitutes in the winner set. 

% If $i$ votes informatively, $m$ will be ranked after all alternative in $\atn_>$ and $\atn_=$ in $\av$. If $i$ turns to approve $m$, $m$ will then ranked before all $\atn_=$ but still after all $\atn_>$. Therefore, the case where $i$'s vote matters is when $|\atn_>| < B$ and $|\atn_>| + |\atn_=| \ge B$, and changing the vote makes $m$ a winner substituting the $(B - |\atn_>|)$-th alternative in $\atn_=$. We denote this substituted agent as $j(\calP)$. In this case, the utility change on this substitution is $\Qual_m - \Qual_{j(\calP)}$. Therefore, a partition $\calP$ and a quality vector $\Qual$ can determine the outcome change and the utility change, which determines a pivotal case. 

We denote $\piv^+$ and $\piv^-$ to be the set of pairs $(\calP, L)$ that lead to pivotal cases with utility increase ($\Qual_1 >\Qual_{j(\calP)}$) and those with utility decrease, respectively. Consequently, we can write the expected utility of agent $i$ as follows. Let $\Pr[\calP \mid \Quallist]$ be the likelihood that the votes of all other agents follow $\calP$ given that the quality vector is $\Qual$ and that all other agents vote informatively. 
\begin{align*}
        &\ \ut_i(\stgp^*, \Sigv_i) - \ut_i(\stgp, \Sigv_i) \\
        = &\ \sum_{(\calP, \Quallist) \in \piv^+} \Pr[\calP \mid \Quallist] \cdot  \Pr[\Quallist\mid \Sigv_i]\cdot (\Qual_{j(\calP)} - \Qual_{1})\\
        -&\ \sum_{(\calP, \Quallist) \in \piv^-} \Pr[\calP \mid \Quallist] \cdot  \Pr[\Quallist\mid \Sigv_i]\cdot (\Qual_{1} - \Qual_{j(\calP)}). 
    \end{align*}

There are at most $(3\bar{\qual})^m$ terms and at least one term in each summation. 
Like the binary case, each $\Pr[\calP \mid \Quallist]$ can be approximated in some closed form. 

\begin{lem}
\label{lem:strategic_unit}

    Given a partition $\calP$, a quality vector $\Quallist$, and an information structure $\sigp$, the probability that
 $(\xrv_j)_{j\in \atn}$ follows the partition $\calP$ is 
    \begin{equation*}
        \Pr[\calP \mid \Quallist] = M(\calP, \Quallist, \sigp) \cdot e^{-nG(\calP, \Quallist, \sigp) + O(\ln n)}, 
    \end{equation*}
    where $G$ and $M$ are positive constants determined by $\calP, \Quallist,$ and $ \sigp$. 
\end{lem}
Intuitively, the $G$ represents the "statistical cost" or "rarity" of the pivotal case. It is derived from the KL-divergence between the expected vote shares under $L$ and the specific tied vote shares required by $\mathcal{P}$. Events with lower $G$ values are exponentially more frequent than those with higher $G$ values.
The proof of Lemma~\ref{lem:strategic_unit}, including the explicit form of $Q$, is in Appendix~\ref{apx:stratic_unit_lemma}.

The sign of the utility difference is determined simply by comparing the minimum $G$ functions.
Given an information structure $\sigp$, let $G^+_{\min}(\sigp)$ and $G^-_{\min}(\sigp)$ be the minimum $G(\calP, \Quallist, \sigp)$ among all $(\calP, \Quallist) \in \piv^+$ and $\piv^-$, respectively. If $G^-_{\min}(\sigp) < G^+_{\min}(\sigp)$, the negative term will dominate the utility difference, and agent $i$ is willing to deviate. 

Then, we consider the second scenario, where the signal vector $\Sigv_i$ has a positive signal only for alternative $m$< and the deviating action is not approving any alternative. Then, $\piv^+$ (defined in the first case) includes all the pivotal cases leading to utility {\em loss} in the second scenario, and $\piv^-$ includes all the pivotal cases leading to the utility gain. Therefore, if $G^-_{\min}(\sigp) > G^+_{\min}(\sigp)$, the negative term will dominate the utility difference, and agent $i$ is willing to deviate. 

Combining the restrictions in the two scenarios leads to a necessary condition for informative voting to be a BNE, which finishes our proof. 

\myparagraph{Theorem 5} (Formal Version) 
{\em Given an environment $\inst$ with unit cost, let $G^+(\sigp) = \min_{(\calP, \Quallist) \in \piv+} G(\calP, \Quallist, \sigp)$, and $G^-(\sigp)$ likewise. Then, when $n\to\infty$, informative voting is a BNE only if $G^-_{\min}(\sigp) = G^+_{\min}(\sigp)$.}

% \begin{thm}
% \label{thm:strategic_unit}
% Given an environment $\inst$ with unit cost, let $G^+(\sigp) = \min_{(\calP, \Quallist) \in \piv+} G(\calP, \Quallist, \sigp)$, and $G^-(\sigp)$ likewise. Then, when $n\to\infty$, informative voting is a BNE only if $G^-_{\min}(\sigp) = G^+_{\min}(\sigp)$. 
% \end{thm}

\section{Proof of Lemma~\ref{lem:strategic_unit}.}
\label{apx:stratic_unit_lemma}

Before we start the proof, we reconstruct the formation into a more generic form. For a partition $\calP = (\atn_>, \atn_=, \atn_<)$, we renumber of index so that $\atn_= = \{2, 3, \cdots, k_1\}$, $\atn_> = \{k_1+1,\cdots, k_2\}$, and $\atn_< = \{k_2+1, \cdots, m\}$. Given the quality vector $\Quallist$, $\xrv_j$ follows binomial distribution $B(n-1, \sigp_j^{\Qual_j})$. Let $p_j = \sigp_j^{\Qual_j}$. In this way, we convert the statement of Lemma~\ref{lem:strategic_unit} into the following form.

Let $X_1, \ldots, X_m$ be independent binomial random variables with $X_i \sim \text{Binomial}(n, p_i)$ where $p_i \in (0,1)$ are fixed parameters. Let
\[
P = \Pr[(X_2= \cdots=X_{k_1} = X_1), (X_{k_1+1}, \cdots, X_{k_2} > X_1), (X_{k_2+1}, \cdots, X_{m} < X_1)] = \sum_{x=0}^{n} f(x)
\]
where
\[
f(x) = \prod_{i=1}^{k_1} \Pr[X_i = x] \cdot \prod_{i=k_1+1}^{k_2} \Pr[X_i > x] \cdot \prod_{j=k_2+1}^{m} \Pr[X_j < x].
\]

Here $0 \leq k_1 \leq k_2 \leq m$. We assume $k_1 \geq 1$ (at least one equality constraint), which aligns with the pivotal case where $A\neq \emptyset$ and provides the key structure for our analysis. The following theorem directly implies Lemma~\ref{lem:strategic_unit}.

\begin{thm}\label{thm:main}
As $n \to \infty$:
\[
\ln P = -nG(\tilde{t}) + O(\ln n)
\]
where $G$ and $\tilde{t}$ are defined below. Equivalently, $P = e^{-nG(\tilde{t}) + O(\ln n)}$.
\end{thm}

%==============================================================================
\subsection{Definitions}
%==============================================================================

\begin{dfn}[Active sets]
For $t \in [0,1]$, define:

\myparagraph{For ``$>$'' constraints} ($i \in \{k_1+1, \ldots, k_2\}$):
\begin{align*}
\calA(t) &= \{i : p_i > t\} & &\text{($\Pr[X_i > nt] \to 1$)} \\
\calB(t) &= \{i : p_i < t\} & &\text{($\Pr[X_i > nt] \to 0$ exponentially)} \\
\calAp(t) &= \{i : p_i = t\} & &\text{($\Pr[X_i > nt] \to 1/2$)}
\end{align*}

\myparagraph{For ``$<$'' constraints} ($j \in \{k_2+1, \ldots, m\}$):
\begin{align*}
\calC(t) &= \{j : p_j < t\} & &\text{($\Pr[X_j < nt] \to 1$)} \\
\calD(t) &= \{j : p_j > t\} & &\text{($\Pr[X_j < nt] \to 0$ exponentially)} \\
\calCp(t) &= \{j : p_j = t\} & &\text{($\Pr[X_j < nt] \to 1/2$)}
\end{align*}

The \myparagraph{active set} is:
\[
\calE(t) = \{1, \ldots, k_1\} \cup \calB(t) \cup \calD(t).
\]
These are the indices contributing to the exponential decay rate.
\end{dfn}

\begin{dfn}[Rate function]
The rate function $G: [0,1] \to [0, +\infty)$ is:
\[
G(t) = \sum_{i \in \calE(t)} \KL{t}{p_i}
\]
where $\KL{t}{p} = t\ln\frac{t}{p} + (1-t)\ln\frac{1-t}{1-p}$ is the KL divergence (relative entropy), with the conventions $0 \ln 0 = 0$. 
\end{dfn}

As we will show in the following section, $G(t)$ has a unique minimum in $[0, 1]$. We denote the $t$ at this minimum as $\tilde{t}$. We give an algorithm to explicitly compute $\tilde{t}$ in Appendix~\ref{apx:compute_tilde_t}.

Note that as $t$ increases, different indices $i$ will transfer into/out of the active set $\calE(t)$ as $t$ crosses $p_i$. Consequently, in these points, $G(t)$ may not be differentiable. We define these points as follows. 

\begin{dfn}[Singular points]
Let $Q = \{p_{k_1+1}, \ldots, p_m\} \cap (0,1)$ be the set of parameters for non-equality constraints (duplicated $p_i$ are removed). A point $t \in (0,1)$ is \myparagraph{singular} if $t \in Q$, and \myparagraph{non-singular} otherwise.

At singular points, the active set $\calE(t)$ changes as $t$ crosses from below to above.
\end{dfn}

%==============================================================================
\subsection{Properties of $G(t)$}
%==============================================================================

\begin{lem}\label{lem:G-continuous}
$G(t)$ is continuous on $[0,1]$ and finite everywhere.
\end{lem}

\begin{proof}
On each interval between consecutive singular points, $G(t)$ is a sum of $\KL{t}{p_i}$ terms with a fixed index set, hence smooth.

At a singular point $t = p_j$ (where $j \in \{k_1+1, \ldots, m\}$), index $j$ either enters or leaves $\calE(t)$. The term being added or removed is $\KL{p_j}{p_j} = 0$. Therefore:
\[
\lim_{t \to p_j^-} G(t) = \lim_{t \to p_j^+} G(t) = G(p_j).
\]

At the boundaries: $G(0) = \sum_{i=1}^{k_1} D(0\|p_i) + \sum_{j \in \calD(0)} D(0\|p_j)$ and $G(1) = \sum_{i=1}^{k_1} D(1\|p_i) + \sum_{i \in \calB(1)} D(1\|p_i)$, both finite since $D(0\|p) = -\ln(1-p)$ and $D(1\|p) = -\ln p$ are finite for $p \in (0,1)$.
\end{proof}

\begin{lem}\label{lem:G-boundary}
The minimum of $G$ is attained in the interior $(0,1)$, not at the boundaries.
\end{lem}

\begin{proof}
We show that $G$ is strictly decreasing near $t = 0$ and strictly increasing near $t = 1$.

The derivative of a single KL-Divergence term is:
\[
\frac{d}{dt} \KL{t}{p} = \ln \frac{t(1-p)}{p(1-t)}.
\]

For any fixed $\calE$:
\[
G'(t) = \sum_{i \in \calE(t)} \ln \frac{t(1-p_i)}{p_i(1-t)}.
\]

Near $t = 0$: For small $t > 0$, we have $t < p_i$ for all $i \in \calE(t)$. For each $i$:
\[
\ln \frac{t(1-p_i)}{p_i(1-t)} \to -\infty \quad \text{as } t \to 0^+.
\]
Therefore $G'(t) \to -\infty$ as $t \to 0^+$, so $G$ is strictly decreasing near $t = 0$.

Near $t = 1$: Similarly, for each $i \in \calE(t)$:
\[
\ln \frac{t(1-p_i)}{p_i(1-t)} \to +\infty \quad \text{as } t \to 1^-.
\]
Therefore $G'(t) \to +\infty$ as $t \to 1^-$, so $G$ is strictly increasing near $t = 1$.

Since $G$ is continuous on $[0,1]$, strictly decreasing near 0, and strictly increasing near 1, the minimum must be attained at some $\tilde{t} \in (0,1)$.
\end{proof}

\begin{lem}\label{lem:G-minimum}
\begin{enumerate}
    \item[(a)] $G(t)$ has a unique global minimum at some $\tilde{t} \in (0,1)$. 
    \item[(b)] For any $\delta > 0$, there exists $\varepsilon > 0$ such that $|t - \tilde{t}| \geq \delta$ implies $G(t) \geq G(\tilde{t}) + \varepsilon$ for all $t \in [0,1]$.
\end{enumerate}
\end{lem}

\begin{proof}
\myparagraph{Property (a):}
By Lemmas~\ref{lem:G-continuous} and~\ref{lem:G-boundary}, $G$ is continuous on $(0,1)$, and the minimum cannot not be at $t= 0$ or $t  =1$. Therefore $G$ attains its minimum at some $\tilde{t} \in (0,1)$.

\myparagraph{Uniqueness:} Let $0 = q_0 < q_1 < \cdots < q_r < q_{r+1} = 1$ be the sorted singular points (including 0 and 1). On each interval $I_j = (q_j, q_{j+1})$, the active set $\calE(t)$ is constant, so:
\[
G(t) = \sum_{i \in \calE_j} \KL{t}{p_i} \quad \text{for } t \in I_j.
\]
Each term $\KL{t}{p_i}$ is strictly convex in $t$ (since $\frac{d^2}{dt^2}\KL{t}{p_i} = \frac{1}{t(1-t)} > 0$ and that $\calE$ is non-empty). Therefore, $G$ is strictly convex on each $I_j$.

Suppose $G$ has two distinct global minimizers $\tilde{t}_1 < \tilde{t}_2$. Let $I_j \ni \tilde{t}_1$ and $I_k \ni \tilde{t}_2$.

\textit{Case 1:} $j = k$ (same interval). This contradicts strict convexity of $G$ on $I_j$.

\textit{Case 2:} $j < k$ (different intervals). Consider the function $G$ on $[\tilde{t}_1, \tilde{t}_2]$. Since $G(\tilde{t}_1) = G(\tilde{t}_2) = \min G$, and $G$ is continuous, $G$ must equal this minimum value somewhere in each interval between. But strict convexity on each piece implies $G$ is strictly above the minimum in the interior, except at the minimizer. The only way for $G$ to equal the minimum in multiple intervals is if the boundary points $q_j, q_{j+1}, \ldots$ are also minimizers. But at a boundary $q_\ell$, the left and right pieces are both strictly convex, so $q_\ell$ can only be a minimizer if both pieces have zero derivative there. By strict convexity, this means $q_\ell$ is the unique minimizer for both pieces, contradicting our assumption of two distinct minimizers.

Therefore, the global minimizer $\tilde{t}$ is unique.

\myparagraph{Property (b):} Since $G$ is continuous on the compact set $[0,1]$ and has a unique minimizer $\tilde{t}$, for any $\delta > 0$, the set $K = \{t \in [0,1] : |t - \tilde{t}| \geq \delta\}$ is compact and does not contain $\tilde{t}$. Therefore $G$ attains its minimum on $K$ at some $t^* \in K$ with $G(t^*) > G(\tilde{t})$. Setting $\varepsilon = G(t^*) - G(\tilde{t}) > 0$ gives the result.

\end{proof}

%==============================================================================
\subsection{Upper and Lower Bounds for $f(x)$}
%==============================================================================

\begin{lem}\label{lem:f-upper}
For all $x \in \{0, 1, \ldots, n-1, n\}$ with $t = x/n$:
\[
f(x) \leq e^{-nG(t)}.
\]
\end{lem}

\begin{proof}
\myparagraph{Point probabilities:} For $x \in \{1, \ldots, n-1\}$ with $t = x/n \in (0,1)$, the standard entropy bound gives $\binom{n}{x} \leq e^{nH(t)}$ where $H(t) = -t\ln t - (1-t)\ln(1-t)$ is the binary entropy. Thus:
\[
P(X_i = x) = \binom{n}{x} p_i^x (1-p_i)^{n-x} \leq e^{nH(t)} \cdot p_i^{nt} (1-p_i)^{n(1-t)} = e^{-n\KL{t}{p_i}}.
\]

For $x = 0$: $P(X_i = 0) = (1-p_i)^n = e^{n\ln(1-p_i)} = e^{-nD(0\|p_i)}$ where $D(0\|p_i) = -\ln(1-p_i)$.

For $x = n$: $P(X_i = n) = p_i^n = e^{n\ln p_i} = e^{-nD(1\|p_i)}$ where $D(1\|p_i) = -\ln p_i$.

So the bound $P(X_i = x) \leq e^{-n\KL{t}{p_i}}$ holds for all $x \in \{0, 1, \ldots, n\}$.

\myparagraph{Tail probabilities:} For ``$>$'' constraints with $p_i < t$ (i.e., $i \in \calB(t)$), Chernoff's bound gives:
\[
\Pr[X_i > x] \leq e^{-n\KL{t}{p_i}}.
\]
For ``$<$'' constraints with $p_j > t$ (i.e., $j \in \calD(t)$):
\[
\Pr[X_j < x] \leq e^{-n\KL{t}{p_j}}.
\]
For other tail probabilities ($p_i > t$ or $p_j < t$), we use the trivial bound $\leq 1$.

\myparagraph{Combining:} 
\[
f(x) \leq \prod_{i=1}^{k_1} e^{-n\KL{t}{p_i}} \cdot \prod_{i \in \calB(t)} e^{-n\KL{t}{p_i}} \cdot \prod_{j \in \calD(t)} e^{-n\KL{t}{p_j}} = e^{-nG(t)}. \qedhere
\]
\end{proof}

\begin{lem}\label{lem:f-lower}
There exist constants $\delta_0 > 0$ and $C > 0$ (depending only on $\{p_i\}$) such that for $t = x/n \in [\delta_0, 1-\delta_0]$:
\[
f(x) \geq C \cdot n^{-m/2} \cdot e^{-nG(t)}.
\]
\end{lem}

\begin{proof}
\myparagraph{Point probabilities:} By Stirling's approximation, for $t \in [\delta_0, 1-\delta_0]$:
\[
\Pr[X_i = x] = \frac{e^{-n\KL{t}{p_i}}}{\sqrt{2\pi n t(1-t)}} \cdot \sqrt{\frac{t(1-t)}{p_i(1-p_i)}} \cdot (1 + O(1/n)) \geq \frac{c_1}{\sqrt{n}} e^{-n\KL{t}{p_i}}
\]
for some constant $c_1 > 0$ depending on $\delta_0$ and $\{p_i\}$.

\myparagraph{Tail probabilities:} By the Berry-Esseen theorem, for $X_i \sim \text{Binomial}(n, p_i)$:
\[
\sup_x \left| P\left(\frac{X_i - np_i}{\sqrt{np_i(1-p_i)}} \leq x\right) - \Phi(x) \right| \leq \frac{C}{\sqrt{n}}
\]
where $\Phi$ is the standard normal CDF and $C$ is an absolute constant.

\begin{itemize}
    \item For $i \in \calA(t)$ (i.e., $p_i > t$): The standardized threshold is $z_n = \frac{nt - np_i}{\sqrt{np_i(1-p_i)}} = -\sqrt{n} \cdot \frac{p_i - t}{\sqrt{p_i(1-p_i)}} \to -\infty$. Thus $\Pr[X_i > nt] = 1 - \Phi(z_n) + O(1/\sqrt{n}) \to 1$. For $t \in [\delta_0, 1-\delta_0]$ and $p_i$ bounded away from $t$, we have $\Pr[X_i > nt] \geq 1/2$ for all $n \geq 1$.
    
    \item For $i \in \calAp(t)$ (i.e., $p_i = t$): The standardized threshold is $z_n = 0$, so $\Pr[X_i > nt] = 1 - \Phi(0) + O(1/\sqrt{n}) = 1/2 + O(1/\sqrt{n}) \geq 1/4$ for large $n$.
    
    \item For $i \in \calB(t)$ (i.e., $p_i < t$): We have $\Pr[X_i > x] \geq \Pr[X_i = x+1]$. Let $t' = (x+1)/n = t + 1/n$. By Stirling's approximation:
    \[
    \Pr[X_i = x+1] = \frac{e^{-n\KL{t'}{p_i}}}{\sqrt{2\pi n t'(1-t')}} \cdot \sqrt{\frac{t'(1-t')}{p_i(1-p_i)}} \cdot (1 + O(1/n)).
    \]
    Since $\KL{t'}{p_i} = \KL{t}{p_i} + O(1/n)$ when $t \in [\delta_0, 1 - \delta_0]$, we have:
    \[
    \Pr[X_i > x] \geq \Pr[X_i = x+1] \geq \frac{c_2}{\sqrt{n}} e^{-n\KL{t}{p_i}}
    \]
    for some constant $c_2 > 0$ depending on $\delta_0$ and $\{p_i\}$.
\end{itemize}

Similar bounds hold for the ``$<$'' constraints (using $\Pr[X_j < x] \geq \Pr[X_j = x-1]$ for $j \in \calD(t)$).

\myparagraph{Combining:} Each of the $m$ factors contributes at least $c/\sqrt{n} \cdot e^{-n\KL{t}{p_i}}$ (where the KL term is 0 for non-active indices). Taking the product:
\[
f(x) \geq C \cdot n^{-m/2} \cdot e^{-nG(t)}. \qedhere
\]
\end{proof}

%==============================================================================
\subsection{Localizing the Sum}
%==============================================================================

We now show that the sum $P = \sum_x f(x)$ is dominated by $x$ near $n\tilde{t}$.

\begin{lem}\label{lem:localize}
For any $\eta > 0$:
\[
\sum_{|x/n - \tilde{t}| \geq \eta} f(x) = o\left(n^{-m/2} \cdot e^{-nG(\tilde{t})}\right).
\]
\end{lem}

\begin{proof}
By Lemma~\ref{lem:G-minimum}(b) with $\delta = \eta$, there exists $\varepsilon > 0$ such that $|t - \tilde{t}| \geq \eta$ implies $G(t) \geq G(\tilde{t}) + \varepsilon$ for all $t \in [0,1]$.

For any $x$ with $|x/n - \tilde{t}| \geq \eta$ and $x \in \{0, 1, \ldots, n-1, n\}$:
\[
f(x) \leq e^{-nG(x/n)} \leq e^{-n(G(\tilde{t}) + \varepsilon)}.
\]

There are at most $n+1$ terms in the sum, so:
\[
\sum_{|x/n - \tilde{t}| \geq \eta} f(x) \leq (n+1) e^{-n(G(\tilde{t}) + \varepsilon)} = e^{-nG(\tilde{t})} \cdot (n+1) e^{-n\varepsilon} = o\left(n^{-m/2} e^{-nG(\tilde{t})}\right). \qedhere
\]
\end{proof}

By this lemma, we have:
\[
P = \sum_{|x - n\tilde{t}| < \eta n} f(x) + o\left(n^{-m/2} e^{-nG(\tilde{t})}\right).
\]

%==============================================================================
\subsection{Evaluating the Main Sum}
%==============================================================================

We now evaluate $\sum_{|x - n\tilde{t}| < \eta n} f(x)$. Set $s = |\calE(\tilde{t})|$ (the size of the active set at the minimum).

\subsubsection{Case 1: $\tilde{t}$ is non-singular}

If $\tilde{t} \neq p_j$ for all $j \in \{k_1+1, \ldots, m\}$, then $\tilde{t}$ lies in the interior of some interval $(q_j, q_{j+1})$, and we can choose $\eta$ small enough that $[\tilde{t} - \eta, \tilde{t} + \eta] \subset (q_j, q_{j+1})$.

In this region, the active set $\calE(t)$ is constant (equal to $\calE(\tilde{t})$), and $G(t)$ is smooth. Since $\tilde{t}$ is a local (hence global) minimum in the interior:
\[
G'(\tilde{t}) = 0, \quad G''(\tilde{t}) = \frac{s}{\tilde{t}(1-\tilde{t})} > 0.
\]

For $x = n\tilde{t} + z$ with $|z| < \eta n$, the Taylor series is as follows:
\[
G(x/n) = G(\tilde{t}) + \frac{G''(\tilde{t})}{2} \cdot \frac{z^2}{n^2} + O(z^3/n^3) = G(\tilde{t}) + \frac{s z^2}{2n^2 \tilde{t}(1-\tilde{t})} + O(z^3/n^3).
\]

\begin{lem}\label{lem:f-precise-nonsingular}
In the non-singular case, for $|z| \leq L := \sqrt{n} \ln n$:
\[
f(n\tilde{t} + z) = \frac{C_0}{\sqrt{n}^s} \cdot e^{-nG(\tilde{t})} \cdot e^{-\frac{s z^2}{2n\tilde{t}(1-\tilde{t})}} \cdot (1 + O(\ln n / \sqrt{n}))
\]
where $C_0 > 0$ is a constant depending on $\{p_i\}$ and $\tilde{t}$.
\end{lem}

\begin{proof}
Let $x = n\tilde{t} + z$ with $|z| \leq L = \sqrt{n}\ln n$, and let $t = x/n = \tilde{t} + z/n$.

\myparagraph{Equality constraints ($i \in \{1, \ldots, k_1\}$):}

Since $\tilde{t}$ is not at $0$ or $1$, we could take a small $\eta$ such that for all $t \in [\tilde{t} - \eta, \tilde{t} + \eta]$, $nt = \Theta(n)$, and $1 - nt = \Theta(n)$, then by Stirling's approximation:
\[
\Pr[X_i = x] = \binom{n}{x} p_i^x (1-p_i)^{n-x} = \frac{1}{\sqrt{2\pi n t(1-t)}} \cdot e^{-n\KL{t}{p_i}} \cdot \sqrt{\frac{t(1-t)}{p_i(1-p_i)}} \cdot (1 + O(1/n)).
\]

For $t = \tilde{t} + z/n$ with $|z| \leq L$:
\begin{itemize}
    \item $\sqrt{t(1-t)} = \sqrt{\tilde{t}(1-\tilde{t})}(1 + O(z/n)) = \sqrt{\tilde{t}(1-\tilde{t})}(1 + O(\ln n/\sqrt{n}))$.
    \item $\KL{t}{p_i} = \KL{\tilde{t}}{p_i} + \KL{\tilde{t}}{p_i}' \cdot \frac{z}{n} + \frac{1}{2}\KL{\tilde{t}}{p_i}'' \cdot \frac{z^2}{n^2} + O(z^3/n^3)$.
\end{itemize}

Since $\KL{t}{p}' = \ln\frac{t(1-p)}{p(1-t)}$ and $\KL{t}{p}'' = \frac{1}{t(1-t)}$, we have:
\[
n\KL{t}{p_i} = n\KL{\tilde{t}}{p_i} + \lambda_i z + \frac{z^2}{2n\tilde{t}(1-\tilde{t})} + O(z^3/n^2)
\]
where $\lambda_i = \ln\frac{\tilde{t}(1-p_i)}{p_i(1-\tilde{t})}$.

Thus:
\[
\Pr[X_i = x] = \frac{c_i}{\sqrt{n}} \cdot e^{-n\KL{\tilde{t}}{p_i}} \cdot e^{-\lambda_i z} \cdot e^{-\frac{z^2}{2n\tilde{t}(1-\tilde{t})}} \cdot (1 + O(\ln n/\sqrt{n}))
\]
where $c_i = \frac{1}{\sqrt{2\pi \tilde{t}(1-\tilde{t})}} \cdot \sqrt{\frac{\tilde{t}(1-\tilde{t})}{p_i(1-p_i)}}$.

\myparagraph{``$>$'' constraints with $i \in \calA(\tilde{t})$ (i.e., $p_i > \tilde{t}$):}

Since $\tilde{t}$ is non-singular and $\eta$ is chosen so that the active set is constant on $[\tilde{t}-\eta, \tilde{t}+\eta]$, we have $p_i > t$ for all $t$ in this range (when $|z| \leq \eta n$). Thus $\Pr[X_i > x] \to 1$ exponentially fast:
\[
\Pr[X_i > x] = 1 - O(e^{-cn}) = 1 \cdot (1 + O(e^{-cn})).
\]

\myparagraph{``$>$'' constraints with $i \in \calAp(\tilde{t})$ (i.e., $p_i = \tilde{t}$):}

This case does not occur when $\tilde{t}$ is non-singular.

\myparagraph{``$>$'' constraints with $i \in \calB(\tilde{t})$ (i.e., $p_i < \tilde{t}$):}

For such $i$, the event $\{X_i > x\}$ is rare. Since $p_i < \tilde{t}$ and $\tilde{t}$ is bounded away from $p_i$ (by the non-singular assumption with small enough $\eta$), we have $\lambda_i = \ln\frac{\tilde{t}(1-p_i)}{p_i(1-\tilde{t})} \geq c > 0$ for some constant $c$.

We compute $\Pr[X_i > x] = \sum_{k=1}^{n-x} \Pr[X_i = x+k]$. The ratio of consecutive terms is:
\[
\frac{\Pr[X_i = x+k+1]}{\Pr[X_i = x+k]} = \frac{(n-x-k)p_i}{(x+k+1)(1-p_i)}.
\]
Let $t = x/n$. At $k = 0$, this ratio equals $\frac{(n-x)p_i}{(x+1)(1-p_i)} = \frac{(1-t)p_i}{t(1-p_i)} \cdot \frac{x}{x+1} = e^{-\lambda_i}(1 + O(1/n))$.

More generally, for $k \leq K := \sqrt{\sqrt n \cdot \ln n}$:
\[
\frac{\Pr[X_i = x+k+1]}{\Pr[X_i = x+k]} = e^{-\lambda_i}\left(1 - \frac{k}{n-x}\right)\left(1 + \frac{k}{x+1}\right)^{-1} = e^{-\lambda_i} (1 + O(k/n)).
\]

Thus 
\begin{align*}
    \Pr[X_i = x+k] = &\ \Pr[X_i = x] \cdot\prod_{j=0}^k\frac{\Pr[X_i = x+k+1]}{\Pr[X_i = x+k]}\\
    =&\ \Pr[X_i = x]\cdot  e^{-k\lambda_i} (1 + O(k^2/n)).
\end{align*}
and:
\begin{align*}
    \sum_{k=1}^{K} \Pr[X_i = x+k] =&\ \Pr[X_i = x] \sum_{k=1}^{K} e^{-k\lambda_i} (1 + O(k^2/n))\\
    =&\ \Pr[X_i = x](1 + O(k^2/n)) \cdot ( \frac{e^{-\lambda_i}(1 - e^{-K\lambda_i})}{1-e^{-\lambda_i}})\\
    =&\ \Pr[X_i = x] \cdot \frac{e^{-\lambda_i}}{1-e^{-\lambda_i}}(1 + O(\ln n/\sqrt{n})). 
\end{align*}

At the same time, for $k > K$, note that $\frac{\Pr[X_i = x+k+1]}{\Pr[X_i = x+k]} \le e^{-\lambda_i}$
\begin{align*}
    \sum_{k=K+1}^{n-x} \Pr[X_i = x+k]\le \Pr[X_i = x]\cdot  \sum_{k=K+1}^{+\infty} e^{-\lambda_i} = \Pr[X_i = x]\cdot \frac{e^{-(K+1)\lambda_i}}{1-e^{-\lambda_i}},
\end{align*}
which is negligible compared to $\sum_{k=1}^{K} \Pr[X_i = x+k]$.
Therefore:
\[
\Pr[X_i > x] = \Pr[X_i = x] \cdot \frac{e^{-\lambda_i}}{1 - e^{-\lambda_i}} \cdot (1 + O(\ln n/\sqrt{n})).
\]

Combined with the Stirling series for $\Pr[X_i = x]$:
\[
\Pr[X_i > x] = \frac{c_i'}{\sqrt{n}} \cdot e^{-n\KL{\tilde{t}}{p_i}} \cdot e^{-\lambda_i z} \cdot e^{-\frac{z^2}{2n\tilde{t}(1-\tilde{t})}} \cdot (1 + O(\ln n/\sqrt{n}))
\]
for some constant $c_i' > 0$.

\myparagraph{``$<$'' constraints:} By symmetric arguments.

\myparagraph{Combining all factors:}

Taking the product over all $m$ factors:
\[
f(x) = \prod_{i=1}^{m} (\text{factor}_i) = \frac{C_0}{\sqrt{n}^s} \cdot \exp\left(-\sum_{i \in \calE(\tilde{t})} n\KL{\tilde{t}}{p_i}\right) \cdot \exp\left(-\sum_{i \in \calE(\tilde{t})} \lambda_i z\right) \cdot e^{-\frac{sz^2}{2n\tilde{t}(1-\tilde{t})}} \cdot (1 + O(\ln n/\sqrt{n}))
\]
where $s = |\calE(\tilde{t})|$ and $C_0$ is a positive constant.

Since $\tilde{t}$ is the minimizer of $G$, we have $G'(\tilde{t}) = 0$, which means $\sum_{i \in \calE(\tilde{t})} \lambda_i = 0$ (this is the first-order optimality condition). Therefore $\sum_{i \in \calE(\tilde{t})} \lambda_i z = 0$, and:
\[
f(x) = \frac{C_0}{\sqrt{n}^s} \cdot e^{-nG(\tilde{t})} \cdot e^{-\frac{sz^2}{2n\tilde{t}(1-\tilde{t})}} \cdot (1 + O(\ln n/\sqrt{n})). \qedhere
\]
\end{proof}

\begin{lem}\label{lem:sum-nonsingular}
In the non-singular case:
\[
\sum_{|z| \leq \eta n} f(n\tilde{t} + z) = \Theta\left(n^{-(s-1)/2} \cdot e^{-nG(\tilde{t})}\right) = e^{-nG(\tilde{t}) + O(\ln n)}.
\]
\end{lem}

\begin{proof}
\textbf{Main region $|z| \leq L$:} Using Lemma~\ref{lem:f-precise-nonsingular}:
\[
\sum_{|z| \leq L} f(n\tilde{t} + z) = \frac{C_0}{\sqrt{n}^s} e^{-nG(\tilde{t})} \sum_{|z| \leq L} e^{-\frac{s z^2}{2n\tilde{t}(1-\tilde{t})}} (1 + O(\ln n /\sqrt{n})).
\]

Let $\sigma^2 = n\tilde{t}(1-\tilde{t})/s = \Theta(n)$. We need to evaluate:
\[
S := \sum_{|z| \leq L} e^{-z^2/(2\sigma^2)}.
\]

\textit{Step 1: Full sum over $\mathbb{Z}$.} By the Poisson summation formula:
\[
\sum_{z \in \mathbb{Z}} e^{-z^2/(2\sigma^2)} = \sqrt{2\pi}\sigma \sum_{k \in \mathbb{Z}} e^{-2\pi^2 k^2 \sigma^2}.
\]
The $k = 0$ term gives $\sqrt{2\pi}\sigma$. For $k \neq 0$, each term is $O(e^{-2\pi^2\sigma^2}) = O(e^{-cn})$ since $\sigma^2 = \Theta(n)$. Therefore:
\[
\sum_{z \in \mathbb{Z}} e^{-z^2/(2\sigma^2)} = \sqrt{2\pi}\sigma (1 + O(e^{-cn})).
\]

\textit{Step 2: Truncation error.} For $|z| > L = \sqrt{n}\ln n$:
\[
\sum_{|z| > L} e^{-z^2/(2\sigma^2)} \leq 2\int_L^\infty e^{-x^2/(2\sigma^2)} dx = 2\sigma \int_{L/\sigma}^\infty e^{-u^2/2} du.
\]
Since $L/\sigma = \sqrt{n}\ln n / \sqrt{n\tilde{t}(1-\tilde{t})/s} = \Theta(\ln n)$, the integral is $O(e^{-c\ln^2 n}/\ln n) = o(n^{-k})$ for any $k > 0$.

\textit{Step 3: Combining.} Therefore:
\[
S = \sqrt{2\pi}\sigma (1 + O(e^{-cn})) - o(n^{-k}) = \sqrt{2\pi}\sigma (1 + o(1)) = \sqrt{\frac{2\pi n \tilde{t}(1-\tilde{t})}{s}}(1 + o(1)).
\]

Substituting back:
\[
\sum_{|z| \leq L} f(n\tilde{t} + z) = \frac{C_0}{\sqrt{n}^s} e^{-nG(\tilde{t})} \cdot \sqrt{\frac{2\pi n \tilde{t}(1-\tilde{t})}{s}} \cdot (1 + O(\ln n / \sqrt{n})) = \Theta(n^{-(s-1)/2} e^{-nG(\tilde{t})}).
\]

\myparagraph{Tail region $L < |z| \leq \eta n$:} Since $G$ is $C^2$ on $(q_j, q_{j+1}) \ni \tilde{t}$ with $G'(\tilde{t}) = 0$ and $G''(t) = s/(t(1-t))$, by Taylor's series with Lagrange remainder: for $t = \tilde{t} + z/n$ with $|z| \leq \eta n$,
\[
G(t) = G(\tilde{t}) + \frac{1}{2}G''(\xi) \cdot \frac{z^2}{n^2}
\]
for some $\xi$ between $\tilde{t}$ and $t$. Since $\xi \in [\tilde{t} - \eta, \tilde{t} + \eta] \subset (q_j, q_{j+1})$ and $G''(\xi) = s/(\xi(1-\xi)) \geq s/((\tilde{t}+\eta)(1-\tilde{t}+\eta)) =: 2c > 0$, we have:
\[
G(t) \geq G(\tilde{t}) + c \cdot \frac{z^2}{n^2}.
\]

Thus $f(n\tilde{t} + z) \leq e^{-nG(t)} \leq e^{-nG(\tilde{t}) - cz^2/n}$.

The sum over $L < |z| \leq \eta n$ is:
\[
\sum_{L < |z| \leq \eta n} f(n\tilde{t} + z) \leq \sum_{|z| > L} e^{-nG(\tilde{t}) - cz^2/n} \leq e^{-nG(\tilde{t})} \cdot O(ne^{-cL^2/n}) = e^{-nG(\tilde{t})} \cdot O(n^{1-c\ln^2 n}).
\]
This is $o(n^{-(s-1)/2} e^{-nG(\tilde{t})})$.
\end{proof}

\subsubsection{Case 2: $\tilde{t}$ is singular}

Suppose $\tilde{t} = q_j$ for some singular point $q_j \in (0, 1)$. Then $G$ is continuous but not differentiable at $\tilde{t}$.

Let $\calE^- = \calE(\tilde{t}^-)$ and $\calE^+ = \calE(\tilde{t}^+)$ be the active sets just below and just above $\tilde{t}$. We have $|\calE^-| = s^-$ and $|\calE^+| = s^+$, with $|s^+ - s^-| = 1$.

On $(q_{j-1}, \tilde{t})$, $G$ is strictly convex with $G''(t) = s^-/(t(1-t))$.

On $(\tilde{t}, q_{j+1})$, $G$ is strictly convex with $G''(t) = s^+/(t(1-t))$.

Since $\tilde{t}$ is a global minimum:
\begin{itemize}
    \item $G'(\tilde{t}^-) \leq 0$ (non-increasing from the left)
    \item $G'(\tilde{t}^+) \geq 0$ (non-decreasing to the right)
\end{itemize}

\begin{lem}\label{lem:sum-singular}
In the singular case:
\[
\sum_{|z| \leq \eta n} f(n\tilde{t} + z) = e^{-nG(\tilde{t}) + O(\ln n)}.
\]
\end{lem}

\begin{proof}
\myparagraph{Upper bound:}

By Lemma~\ref{lem:f-upper}, $f(x) \leq e^{-nG(x/n)}$ for all $x$. Since $G(t) \geq G(\tilde{t})$ for all $t$:
\[
\sum_{|z| \leq \eta n} f(n\tilde{t} + z) \leq (2\eta n + 1) \cdot e^{-nG(\tilde{t})} = O(n) \cdot e^{-nG(\tilde{t})}.
\]

\myparagraph{Lower bound:}

By Lemma~\ref{lem:f-lower}, there exists $\delta_0 > 0$ such that for $t \in [\delta_0, 1-\delta_0]$:
\[
f(x) \geq C n^{-m/2} e^{-nG(t)}.
\]

Consider the single term $x = \lfloor n\tilde{t} \rfloor$ (or $\lceil n\tilde{t} \rceil$). Let $t^* = x/n$, so $|t^* - \tilde{t}| \leq 1/n$.

By Taylor series of $G$ around $\tilde{t}$:
\[
G(t^*) = G(\tilde{t}) + O(1/n)
\]
(using the fact that $G$ is continuous and has bounded derivatives on compact subsets).

Therefore:
\[
f(x) \geq C n^{-m/2} e^{-nG(t^*)} = C n^{-m/2} e^{-nG(\tilde{t}) - O(1)} = \Omega(n^{-m/2}) \cdot e^{-nG(\tilde{t})}.
\]

\myparagraph{Conclusion:}

Combining the bounds:
\[
\Omega(n^{-m/2}) \cdot e^{-nG(\tilde{t})} \leq \sum_{|z| \leq \eta n} f(n\tilde{t} + z) \leq O(n) \cdot e^{-nG(\tilde{t})}.
\]

\end{proof}

\begin{remark}\label{rem:polynomial-factors}
In the non-singular case, the more precise result $\sum f = \Theta(n^{-(s-1)/2} e^{-nG(\tilde{t})})$ holds (Lemma~\ref{lem:sum-nonsingular}). 

In the singular case where $\tilde{t} = p_j$, the polynomial factor depends on the left and right derivatives $G'(\tilde{t}^-)$ and $G'(\tilde{t}^+)$. Specifically:
\begin{itemize}
    \item If $G'(\tilde{t}^-) = G'(\tilde{t}^+) = 0$ (i.e., $\tilde{t}$ is also a critical point from both sides), the polynomial factor is $\Theta(n^{-(s-1)/2})$, same as the non-singular case.
    \item If $G'(\tilde{t}^-) < 0 < G'(\tilde{t}^+)$ (the generic singular case), the tilted Gaussian sums on each side yield polynomial factors between $n^{-m/2}$ and $n^{-(s-1)/2}$.
\end{itemize}
In all cases, the weaker bound $e^{-nG(\tilde{t}) + O(\ln n)}$ holds.
\end{remark}

%==============================================================================
\subsection{Completing the Proof of Theorem~\ref{thm:main}}
%==============================================================================

\begin{proof}[Proof of Theorem~\ref{thm:main}]
By Lemma~\ref{lem:localize} with any fixed $\eta > 0$:
\[
P = \sum_{|x - n\tilde{t}| < \eta n} f(x) + o(n^{-m/2} e^{-nG(\tilde{t})}).
\]

By Lemmas~\ref{lem:sum-nonsingular} and~\ref{lem:sum-singular}:
\[
\sum_{|x - n\tilde{t}| < \eta n} f(x) = e^{-nG(\tilde{t}) + O(\ln n)}.
\]

The error term $o(n^{-m/2} e^{-nG(\tilde{t})}) = O(e^{-nG(\tilde{t}) + O(\ln n)})$ as well.

Therefore:
\[
\ln P = -nG(\tilde{t}) + O(\ln n). \qedhere
\]
\end{proof}

%==============================================================================
\subsection{Discussion}
%==============================================================================
\label{apx:algorithm}
\subsubsection{Computing $\tilde{t}$}
\label{apx:compute_tilde_t}
Given parameters $\{p_1, \ldots, p_m\}$, the minimizer $\tilde{t}$ can be found by the following algorithm:

\myparagraph{Step 1:} Sort the singular points. Let $Q = \{p_{k_1+1}, \ldots, p_m\}$ be the parameters of non-equality constraints, and sort them as $0 = q_0 < q_1 < \cdots < q_r < q_{r+1} = 1$.

\myparagraph{Step 2:} For each interval $(q_j, q_{j+1})$, the active set $\calE_j := \calE(t)$ is constant for $t \in (q_j, q_{j+1})$. Compute minimum of G(t) with $\calE_j$ in $(0, 1)$:
\[
t_j^* = \frac{\alpha_j}{\alpha_j + \beta_j}, \quad \text{where } \alpha_j = \left(\prod_{i \in \calE_j} p_i\right)^{1/|\calE_j|}, \quad \beta_j = \left(\prod_{i \in \calE_j} (1-p_i)\right)^{1/|\calE_j|}.
\]
If $t_j^* \in (q_j, q_{j+1})$, compute $G(t_j^*)$.

\myparagraph{Step 3:} For each singular point $q_\ell \in Q$, check if $\tilde{t} = q_\ell$ is a local minimum by verifying:
\[
G'(q_\ell^-) \leq 0 \leq G'(q_\ell^+).
\]
If so, compute $G(q_\ell)$.

\myparagraph{Step 4:} The global minimizer $\tilde{t}$ is the candidate with the smallest value of $G$.

Since there are at most $r+1$ intervals and $r$ singular points, this algorithm runs in $O(m \log m)$ time (dominated by sorting).

\subsubsection{The role of $s$}

The parameter $s = |\calE(\tilde{t})|$ counts:
\begin{itemize}
    \item All equality constraints ($k_1$ indices)
    \item ``$>$'' constraints where $p_i < \tilde{t}$ (rare events)
    \item ``$<$'' constraints where $p_j > \tilde{t}$ (rare events)
\end{itemize}

In the non-singular case, the polynomial correction is precisely $n^{-(s-1)/2}$.

\subsubsection{Explicit formula in the non-singular case}

When $\tilde{t}$ is non-singular, the minimizer $\tilde{t}$ satisfies $G'(\tilde{t}) = 0$, which gives:
\[
\sum_{i \in \calE(\tilde{t})} \ln\frac{\tilde{t}(1-p_i)}{p_i(1-\tilde{t})} = 0.
\]
Define:
\[
\alpha = \left(\prod_{i \in \calE(\tilde{t})} p_i\right)^{1/s}, \quad \beta = \left(\prod_{i \in \calE(\tilde{t})} (1-p_i)\right)^{1/s}.
\]

Then $\tilde{t} = \alpha/(\alpha + \beta)$ and $e^{-nG(\tilde{t})} = (\alpha + \beta)^{sn}$.

The full formula is:
\[
P = \frac{(\alpha+\beta)^{sn}}{(2\pi n)^{(s-1)/2} \cdot (\alpha\beta)^{(s-1)/2} \cdot \sqrt{s}} \cdot \prod_{i \in \calB(\tilde{t})} \frac{e^{-\lambda_i}}{1-e^{-\lambda_i}} \cdot \prod_{j \in \calD(\tilde{t})} \frac{e^{-\mu_j}}{1-e^{-\mu_j}} \cdot \frac{1}{2^{a'+c'}} \cdot (1 + O(1/\sqrt{n}))
\]
where $\lambda_i = \ln\frac{\tilde{t}(1-p_i)}{p_i(1-\tilde{t})}$ for $i \in \calB(\tilde{t})$, $\mu_j = \ln\frac{p_j(1-\tilde{t})}{\tilde{t}(1-p_j)}$ for $j \in \calD(\tilde{t})$, and $a' = |\calAp(\tilde{t})|$, $c' = |\calCp(\tilde{t})|$.

\subsubsection{The singular case}

When $\tilde{t} = p_j$ for some $j \in \{k_1+1, \ldots, m\}$, the minimizer is explicitly given (not determined by a first-order condition). The exponential rate is:
\[
G(\tilde{t}) = G(p_j) = \sum_{i \in \calE(p_j)} \KL{p_j}{p_i}
\]
where $\calE(p_j) = \{1, \ldots, k_1\} \cup \calB(p_j) \cup \calD(p_j)$ and index $j$ itself is \emph{not} in $\calE(p_j)$ (since $p_j \not< p_j$ and $p_j \not> p_j$).

Equivalently:
\[
e^{-nG(p_j)} = \prod_{i \in \calE(p_j)} \left(\frac{p_j}{p_i}\right)^{np_j} \left(\frac{1-p_j}{1-p_i}\right)^{n(1-p_j)}.
\]

The polynomial factor in this case is not simply $n^{-(s-1)/2}$ but depends on $G'(p_j^-)$ and $G'(p_j^+)$, as discussed in Remark~\ref{rem:polynomial-factors}.

\section{Numerical Verification for Theorem~\ref{thm:strategic_unit}.}
\label{apx:rarity}
Here we display our numerical verification that the necessary condition characterized in Theorem~\ref{thm:strategic_unit} is rare in all information structures. 

\myparagraph{Simulation Design} In each setting, we fix the number of alternative $m$, number of winners $B$, range of quality $\qualset$. Then we enumerate all the pairs of partition $\calP$ and quality vector $\Quallist$ in $\piv^+$ and $\piv^-$ (defined in Appendix~\ref{apx:strategic_unit}) and construct the two sets. Then, we randomly sample 1000 points from the information structure space $(0, 1)^{\bar{\ell} m}$. For each sample $\sigp$, we calculate $G(\calP, \Quallist, \sigp)$ for every pair in $\piv^+$ and $\piv^-$, and compare the minimum between both sets. If the difference between $G^+(\sigp)$ and $G^-(\sigp)$ is in $\varepsilon = 10^{-8}$, we view as the necessary condition holds under this $\sigp$. $G$ functions are computed as the algorithm in Appendix~\ref{apx:algorithm}. 

\myparagraph{Execution Environment.} All the experiments are run with Python 3.14 on a Laptop with an Intel Ultra 258V CPU, 32G RAM, and Windows 11 Home System. 

\myparagraph{Result}
Table~\ref{tbl:rarity} shows our results of the simulation. In both settings, the condition holds in a rate less than 1\%. 
\begin{table}[htbp]
\centering
\begin{tabular}{@{}rrrrrrr@{}}
\toprule
$m$ & $B$ & Max Quality $\bar{\ell}$ & Pairs in $\piv^+$ & Pairs in $\piv^-$ & Conditions hold times & Condition hold rate \\ \midrule
5 & 2 & 2 & 3159 & 3159 & 7 & 0.7\% \\
6 & 3 & 2 & 24543 & 24543 & 2 & 0.2\% \\ \bottomrule
\end{tabular}
\caption{Result of Simulation~\label{tbl:rarity}}
\end{table}

\section{More Experiments}
\label{apx:experiment}

Here, we give a more detailed explanation of how we design our experiments as well as a more thorough description of the experimental results. 

\subsection{Experimental Design}

\myparagraph{Instance Generation.} In each instance, we fix the number of agents $n$, number of alternatives $m$, cost ratio $\alpha$ (we fix the min-cost at 1 in the experiments), budget $B$, the utility function $\vt$, and quality range $\qualset$. We call a set of fixed parameters a {\em setting}. In the first step, we generate a quality vector, the cost of agents, and the information structure. The common prior of every alternative is fixed to the uniform distribution on $\qualset$. The cost of each alternative is uniformly generated from $[1, \alpha]$, with a guaranteed 1 and a guaranteed $\alpha$. For information structure, we divide $[0.1, 0.9]$ by the number of possible qualities, and pick each $\sigp_j^\ell$ u.a.r in each region. We call a setting together with the first-step generated parameters a {\em trial}. In the second step, we follow the generated information structure to drawn agent signal i.i.d. A trial together with a set of signals is called a {\em sample}. For each set of fixed parameters, we generate 100 trials, and for each trial, we generate 100 samples. Therefore, each set of fixed parameters (which corresponds to one point in the figures) includes 10000 samples.

\myparagraph{Voting Rules.} We test eight rules in our experiments: $\av$, $\av$/cost (pick the alternative with the highest (approval count)/(cost) ratio in each step, degenerated to $\av$ in unit cost), $\pav$, $\gc$, $\phr$, $\mes$, $\mes$+$\av$, and $\mes$ + $\phr$. 

\myparagraph{Performance Calculation.} For each sample, we run each voting rule and calculate the utility of its outcome. At the same time, we also run a knapsack program to find the optimum to find the feasible set of alternatives with the highest utility. For each voting rule $r$, we compute the ratio between its outcome and the optimal outcome. The empirical performance of $r$ in a setting is the average of all such ratios of all the samples in this setting. 

\myparagraph{Execution Environment.} All the experiments are run with Python 3.14 on a Laptop with an Intel Ultra 258V CPU, 32G RAM, and Windows 11 Home System. 

\subsection{Results.}

In all the experiments, we fix $m = 8$, and $\qualset= \{0, 1, 2\}$. 

\begin{figure}
\centering
\begin{subfigure}{0.49\textwidth}
\centering
\includegraphics[width=\textwidth]{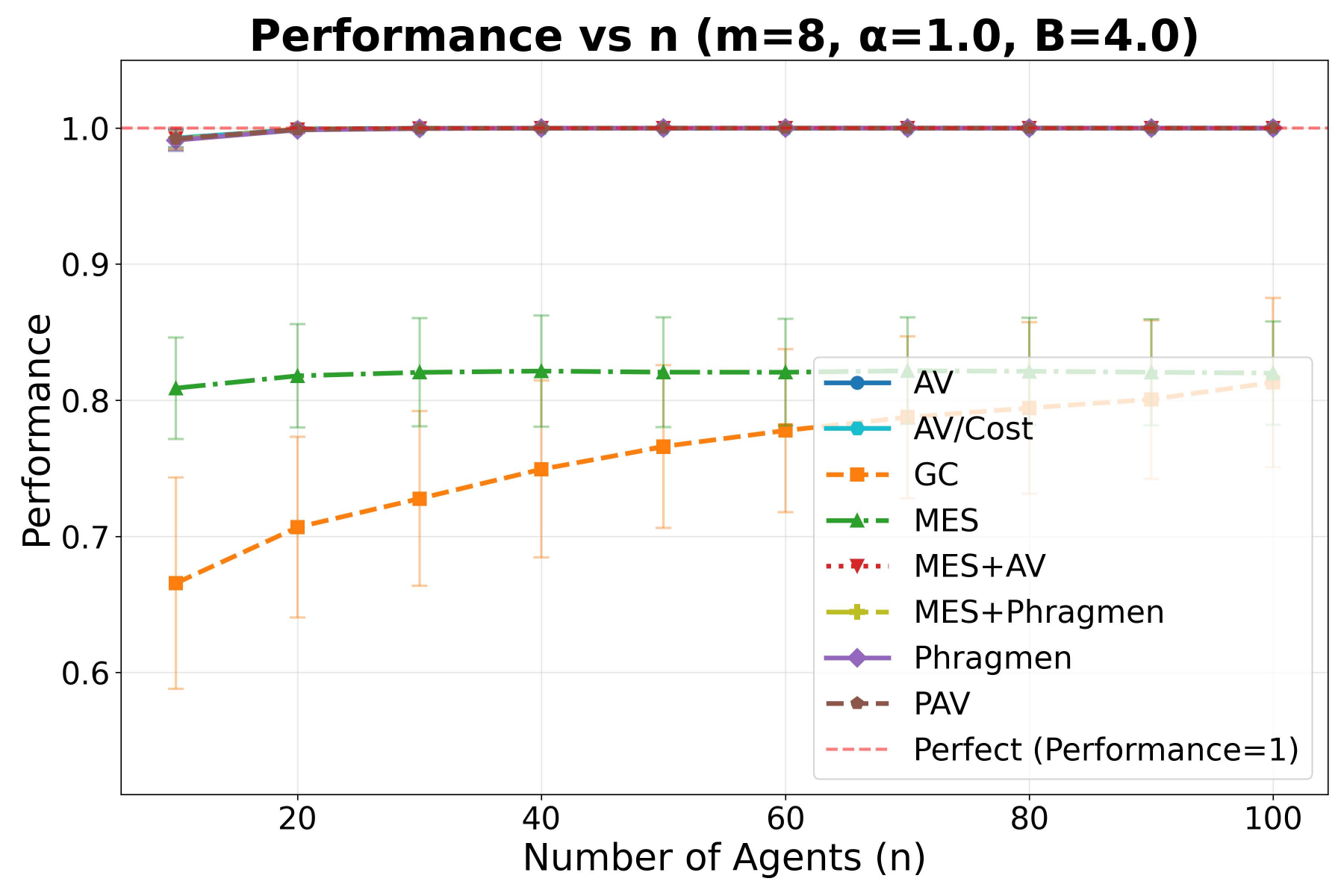}
\caption{$\alpha=1$}
\end{subfigure}
\begin{subfigure}{0.49\textwidth}
\centering
\includegraphics[width=\textwidth]{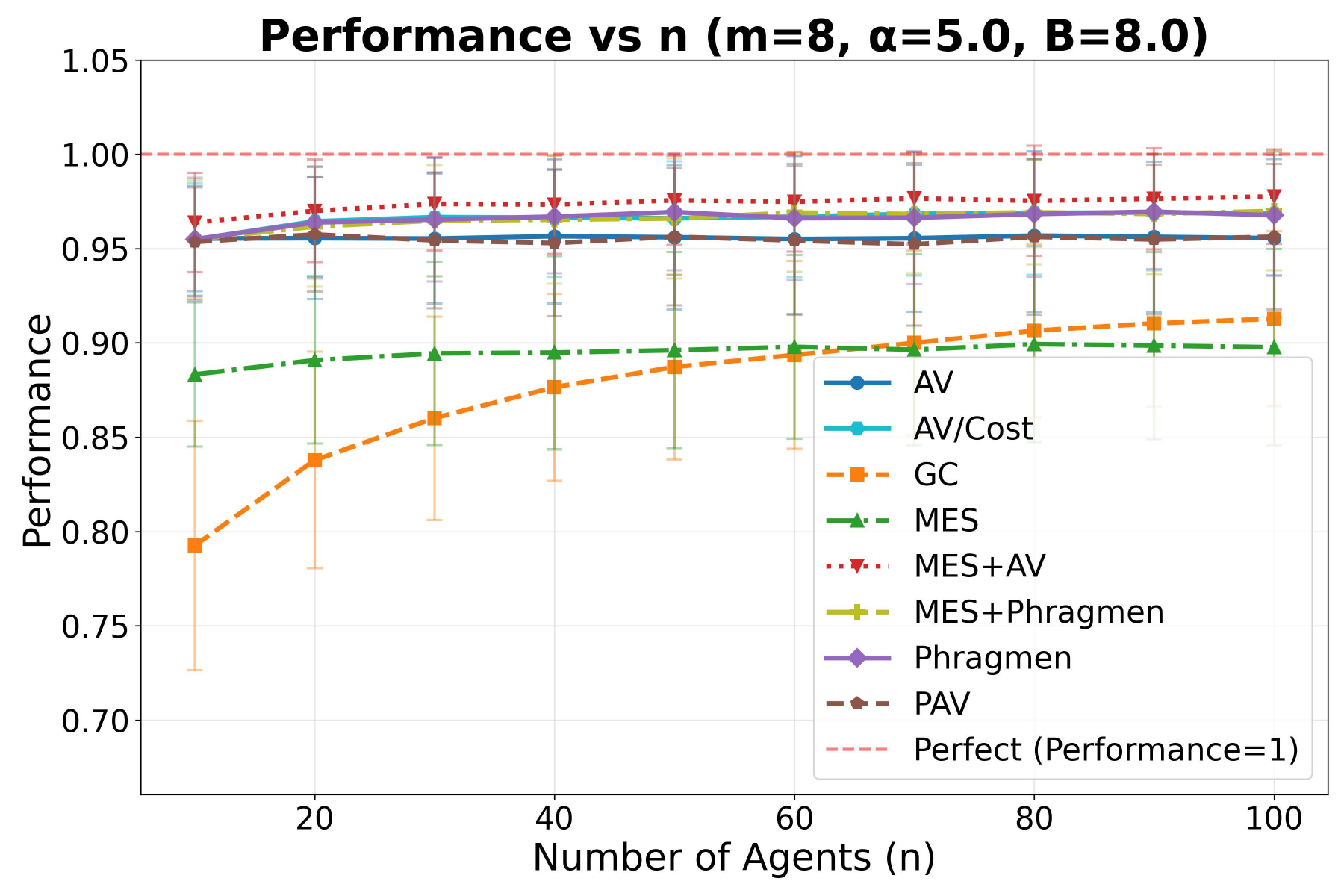}
\caption{$\alpha = 5$}
\end{subfigure}
\caption{Performance on the number of agents, normal utility.\label{fig:apx_performance}}
\end{figure}

\myparagraph{Performance on number of agents.} Figure~\ref{fig:apx_performance} shows the relationship between performance and the number of agents in different scenarios. We consider two settings in this scenario: unit cost ($\alpha=1$ and $B = 4$) and general cost ($\alpha = 5$ and $B = 8$). The number of agents takes $n = 10, 20, \cdots, 100$. In the unit case, most voting rules have performance very close to 1 starting from $n = 20$, supporting our Theorem~\ref{thm:unit_positive}. The only two exceptions are $\mes$ and $\gc$. The cause of $\mes$'s low performance is that it may not exhaust the budget. After adding an exhaustion rule, the performance goes back to 1. For $\gc$, as we will see in Figure~\ref{fig:gc}, $\gc$ has an extremely low convergence rate, and its performance reaches $0.99$ only after $n \ge 3000$. In the general case, on the other hand, the performance of all rules is bounded away from 1, but still much better than the worst case upper bound, which is $4/7\approx 0.57$ in this case. With a paired t-test on each pair of voting rules, we have 95\% confidence that $\mes$ + $\av$ has the highest performance; $\av$/cost, $\mes$+$\phr$, and $\phr$ are in the second place; and $\av$ and $\pav$ are in the third place. Additional experiments on the convergence of each rule can be found in Figure~\ref{fig:convergence}. 

\myparagraph{Effect of budget }
Figure~\ref{fig:budget} shows the effect of the budget $B$ on the performance. We fix $\alpha =5$ and increase the budget $B$ from $5$ to $30$ with a step of 5. When $B$ is close to 30 and can hold most of the alternatives (the expected total cost is 27), all rules have higher performance. However, when $B$ is small, the behavior of different voting rules is quite different.
\begin{figure}
    \centering
    \includegraphics[width=0.6\linewidth]{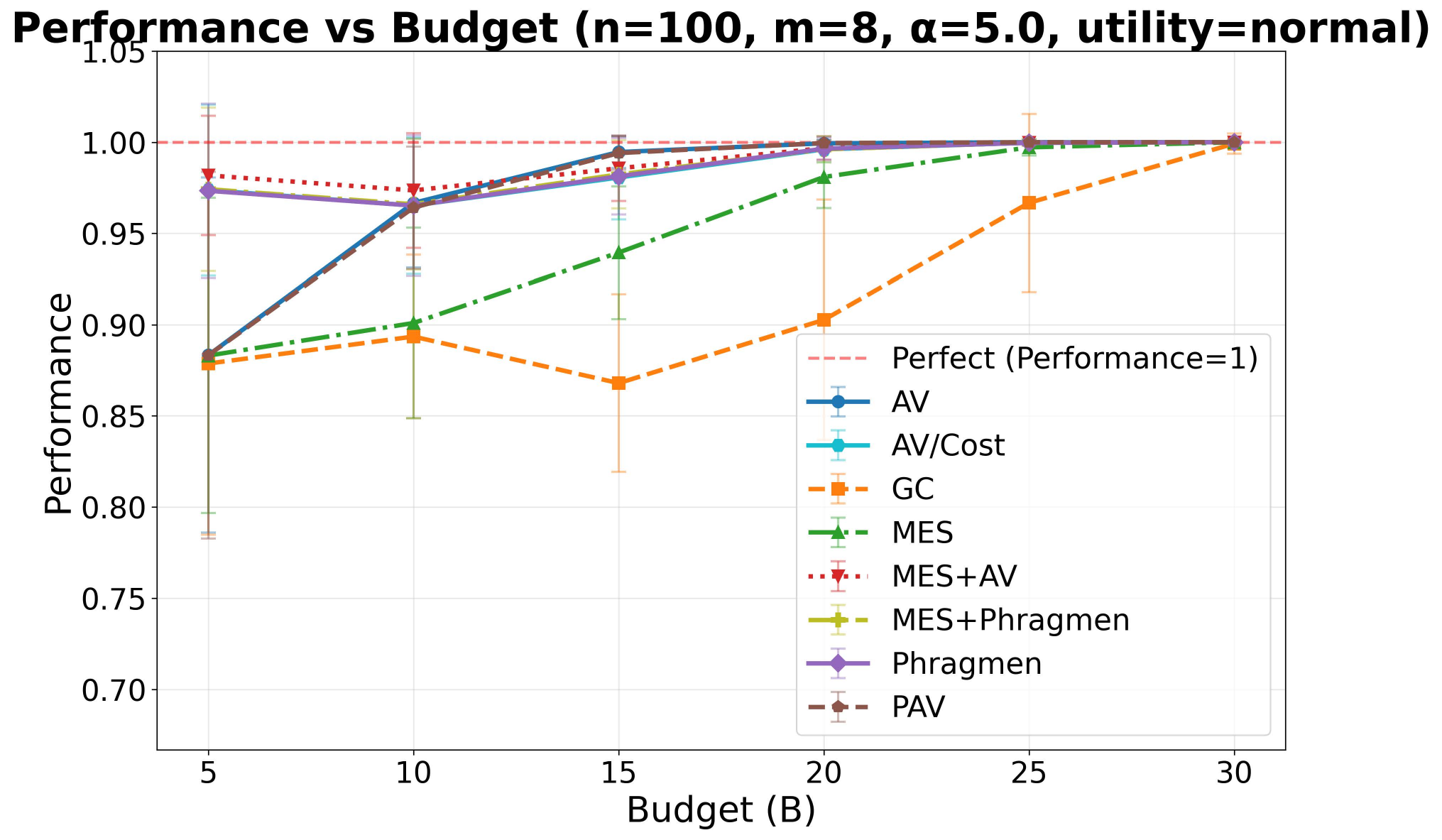}
    \caption{Performance under different budgets, normal utility \label{fig:budget}}. 
   
\end{figure}

\myparagraph{Effect of cost ratio.} Figure~\ref{fig:apx_alpha} tests the effect of diferrent cost ratio $\alpha$ on the performance. We fix the budget $B = 7$ and set $\alpha = 1, 3, 5, 7$. In general, increasing $\alpha$ causes a decrease in the performance. On the other hand, the effects of different voting rules are different. For  $\av$ and $\pav$, there is an obvious difference (95\% confidence by paired t-test) between different levels of $\alpha$. For $\mes$ and its exhausted variants, the performance between $\alpha = 5$ and $\alpha = 7$ is not significant. Most surprisingly, for $\gc$, the performance is higher for $\alpha = 5$ than $\alpha = 3$. 

We also do an additional test by fixing $B$ to be exactly half of the expected total cost (which is $\frac{(\alpha + 1)m}{2}$) and increasing $\alpha$ from 1 to 10, as shown in Figure~\ref{fig:alpha_ratio}. Except for $\mes$ and $\gc$, there is a decrease in most voting rules as $\alpha$ increases, and the performance seems to be converging as $\alpha$ approaches 10. We conjecture that both the increase in the absolute value of $\alpha$ and the decrease in the ratio between $B$ and $\alpha$ cause the performance decrease in Figure~\ref{fig:apx_alpha}. 

\myparagraph{Utility Function.} For all the experiments mentioned above, we run the experiments for both normal utility and cost-proportional utility. No significant difference occurs after changing the utility function. The cost-proportional figures are at the end of this section. 

\begin{figure}
    \centering
    \includegraphics[width=0.5\linewidth]{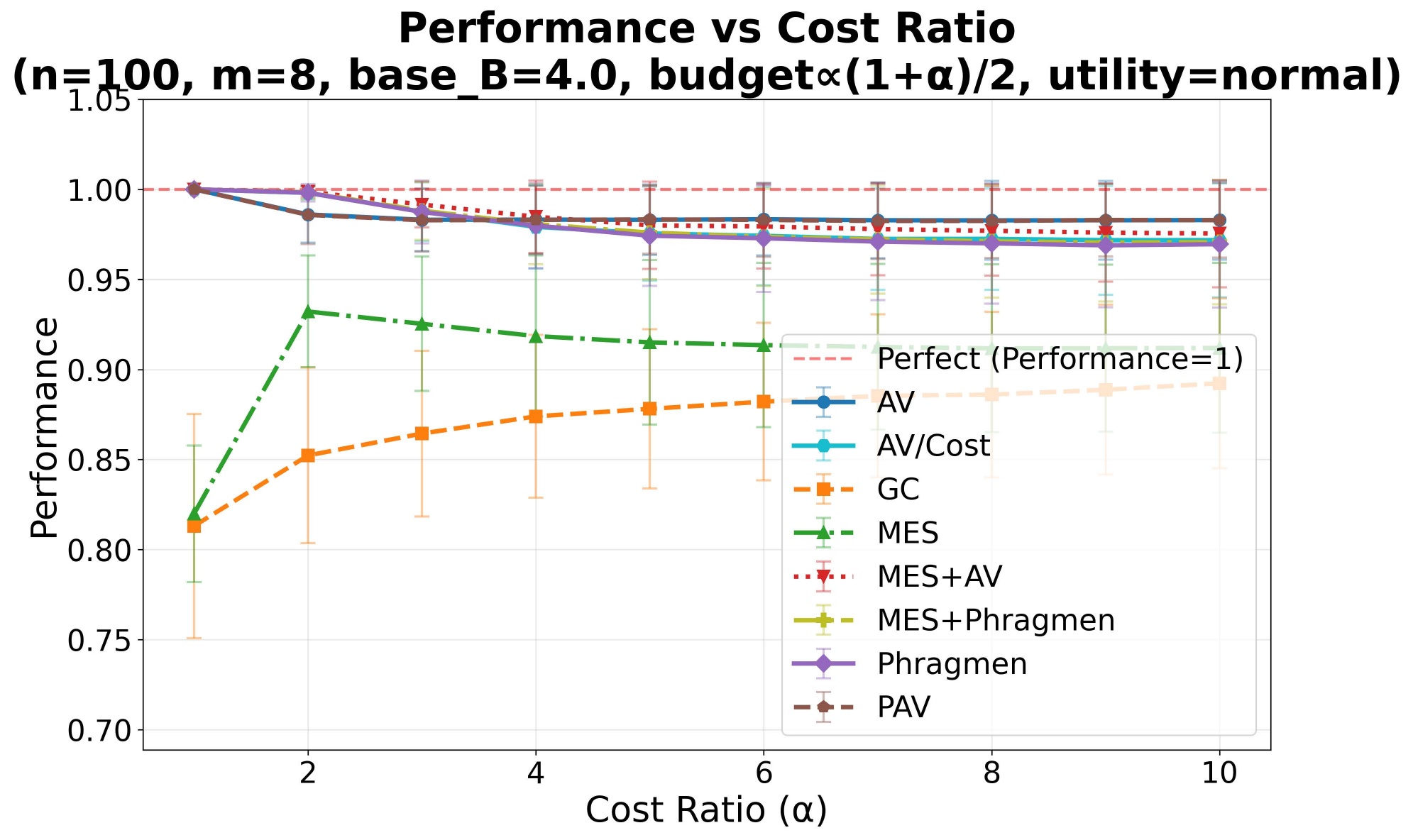}
    \caption{performance when $B$ and total cost are correlated, normal utility \label{fig:alpha_ratio}}
    
\end{figure}

\begin{figure}
    \centering
    \includegraphics[width=0.8\linewidth]{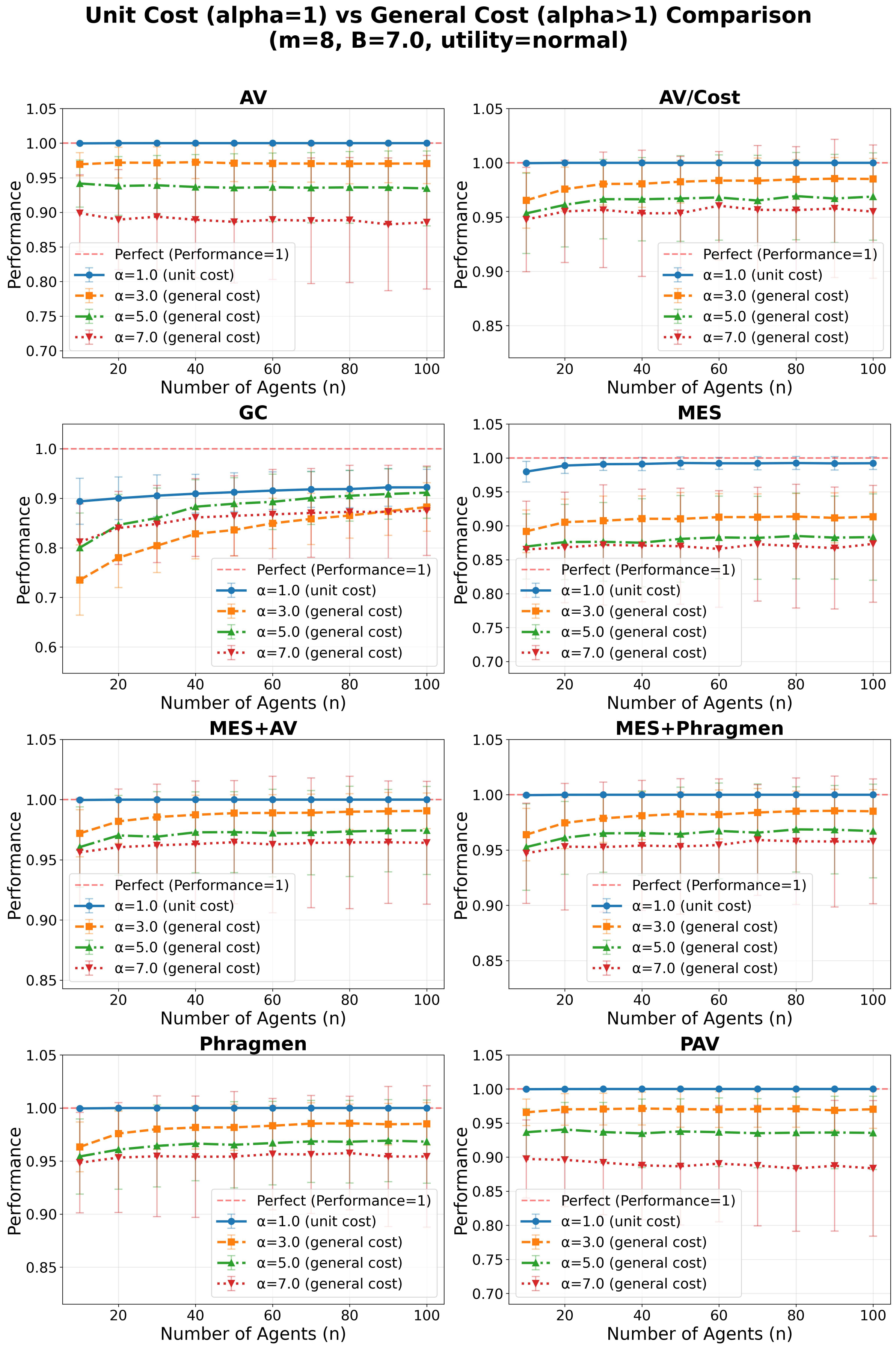}
    \caption{Performance on different cost ratios, normal utility\label{fig:apx_alpha}}
    
\end{figure}

\begin{figure}
    \centering
    \includegraphics[width=0.8\linewidth]{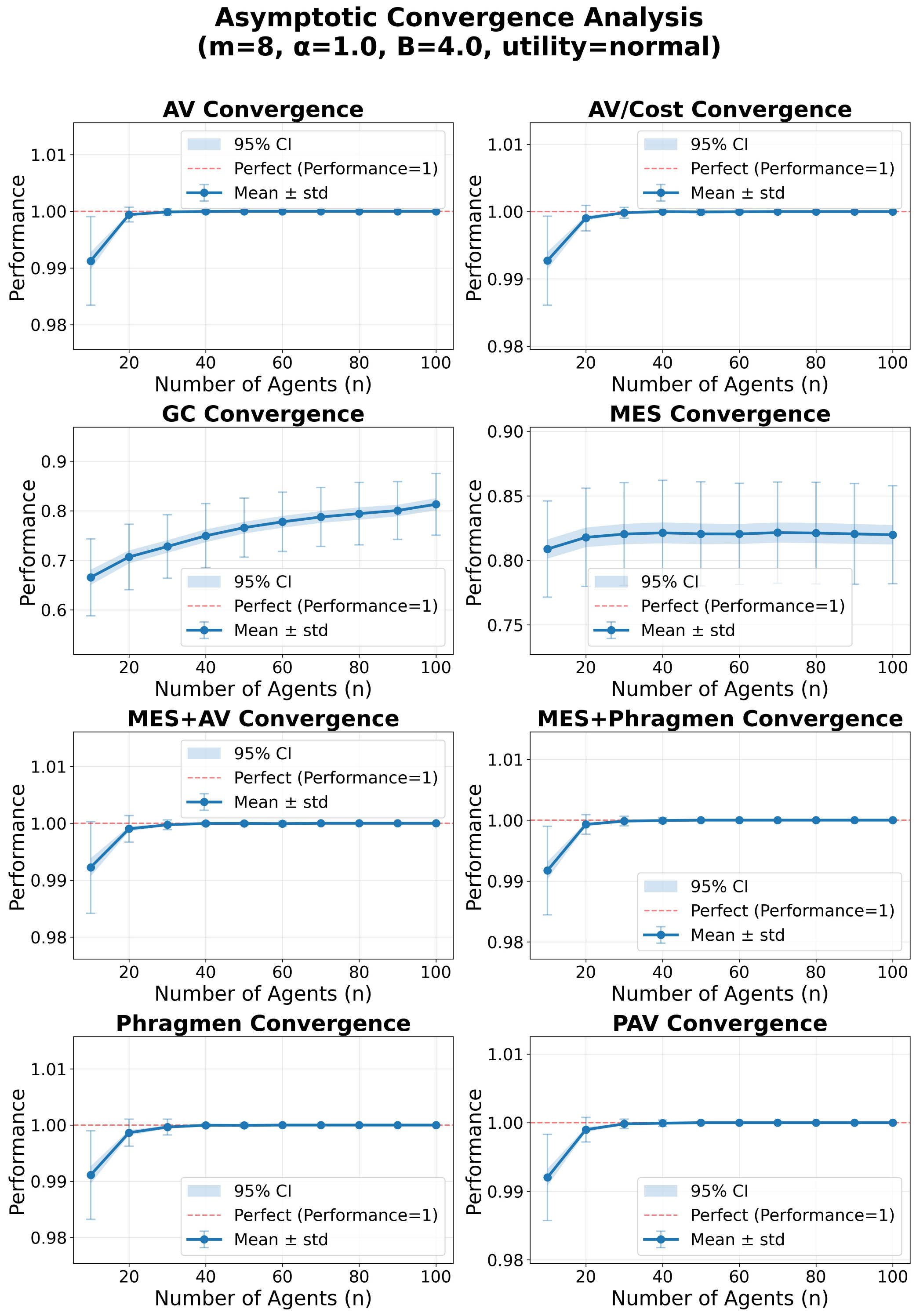}
    \caption{Convergence of rules, normal utility}
    \label{fig:convergence}
\end{figure}

\begin{figure}
    \centering
    \includegraphics[width=0.99\linewidth]{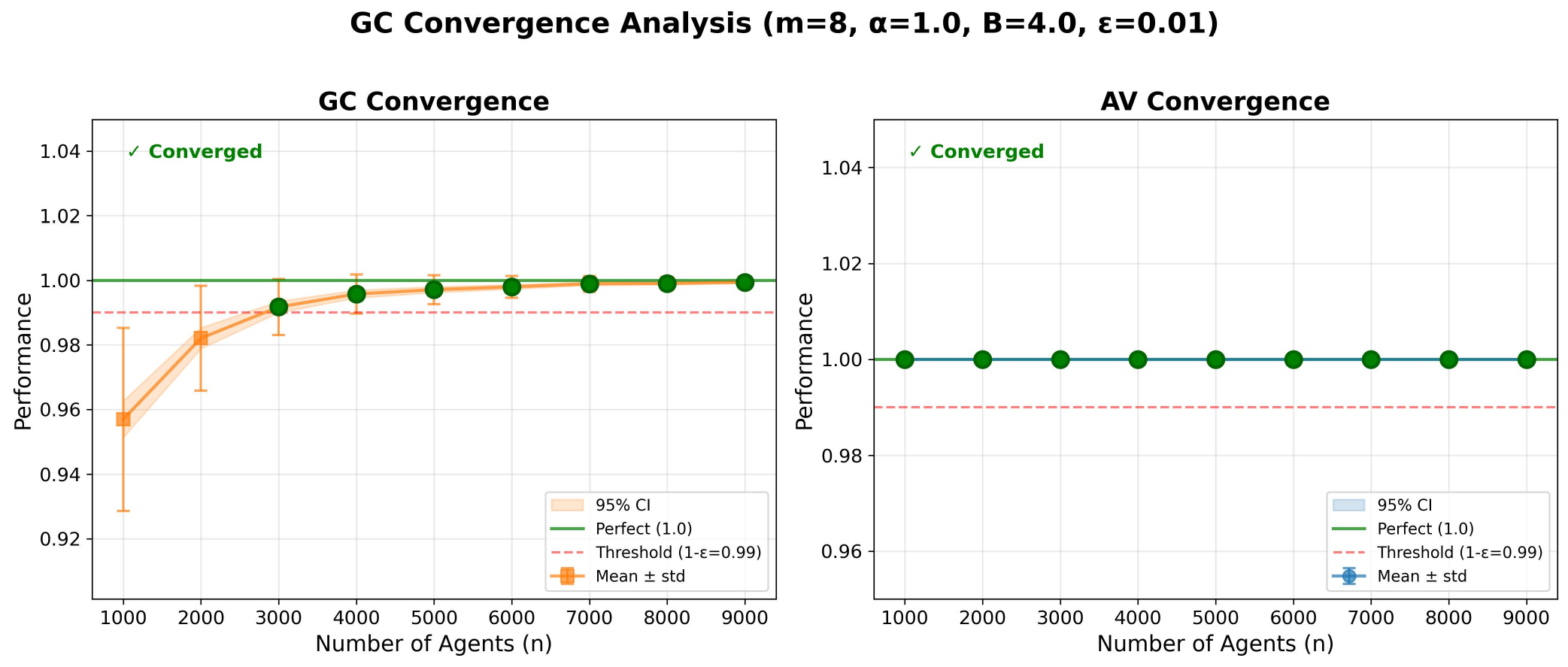}
    \caption{Convergence on $\gc$ in large $n$ (in comparison with $\av$), normal utility\label{fig:gc}}. 
    
\end{figure}

\begin{figure}[b]
\centering
\begin{subfigure}{0.49\textwidth}
\centering
\includegraphics[width=\textwidth]{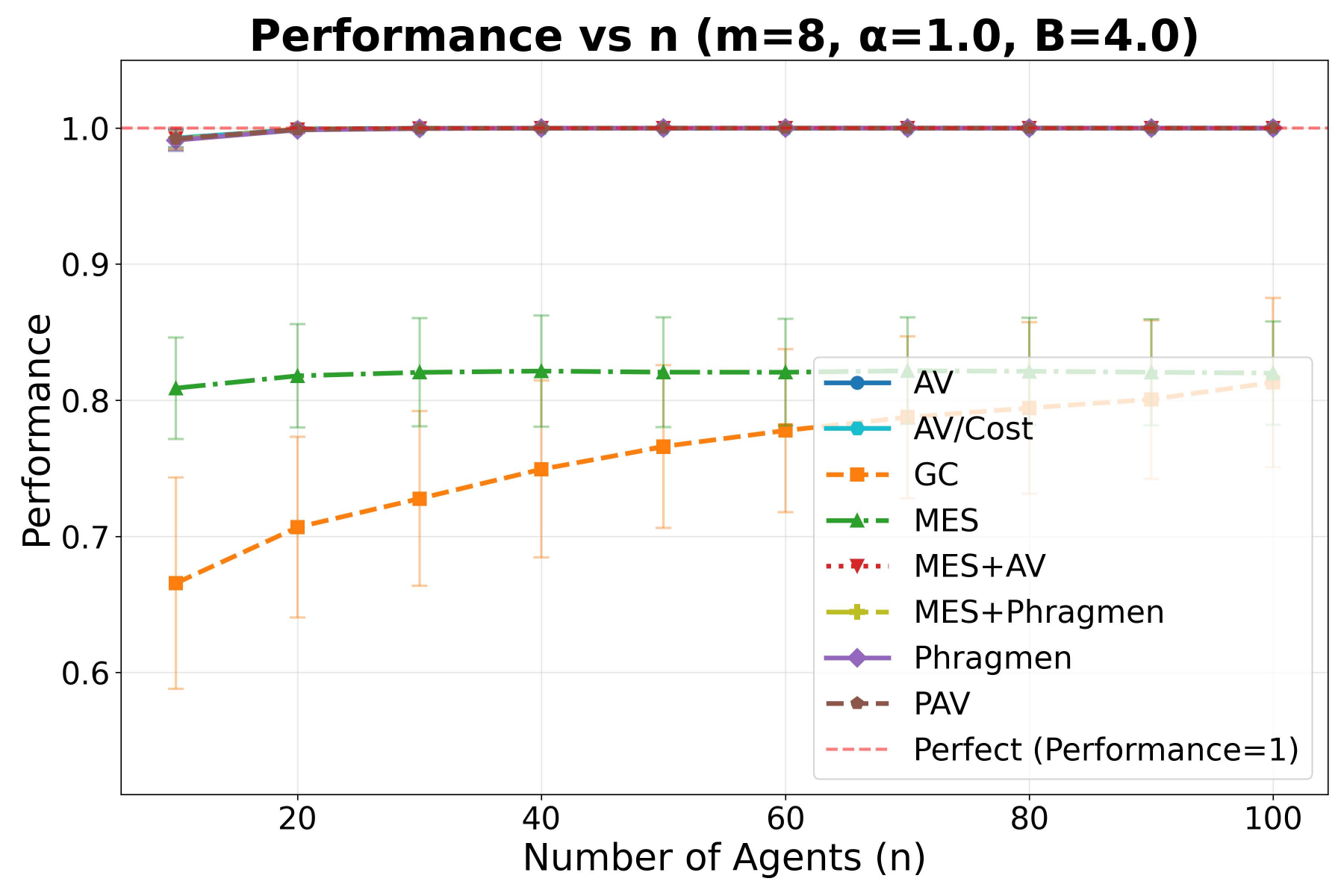}
\caption{$\alpha=1$}
\end{subfigure}
\begin{subfigure}{0.49\textwidth}
\centering
\includegraphics[width=\textwidth]{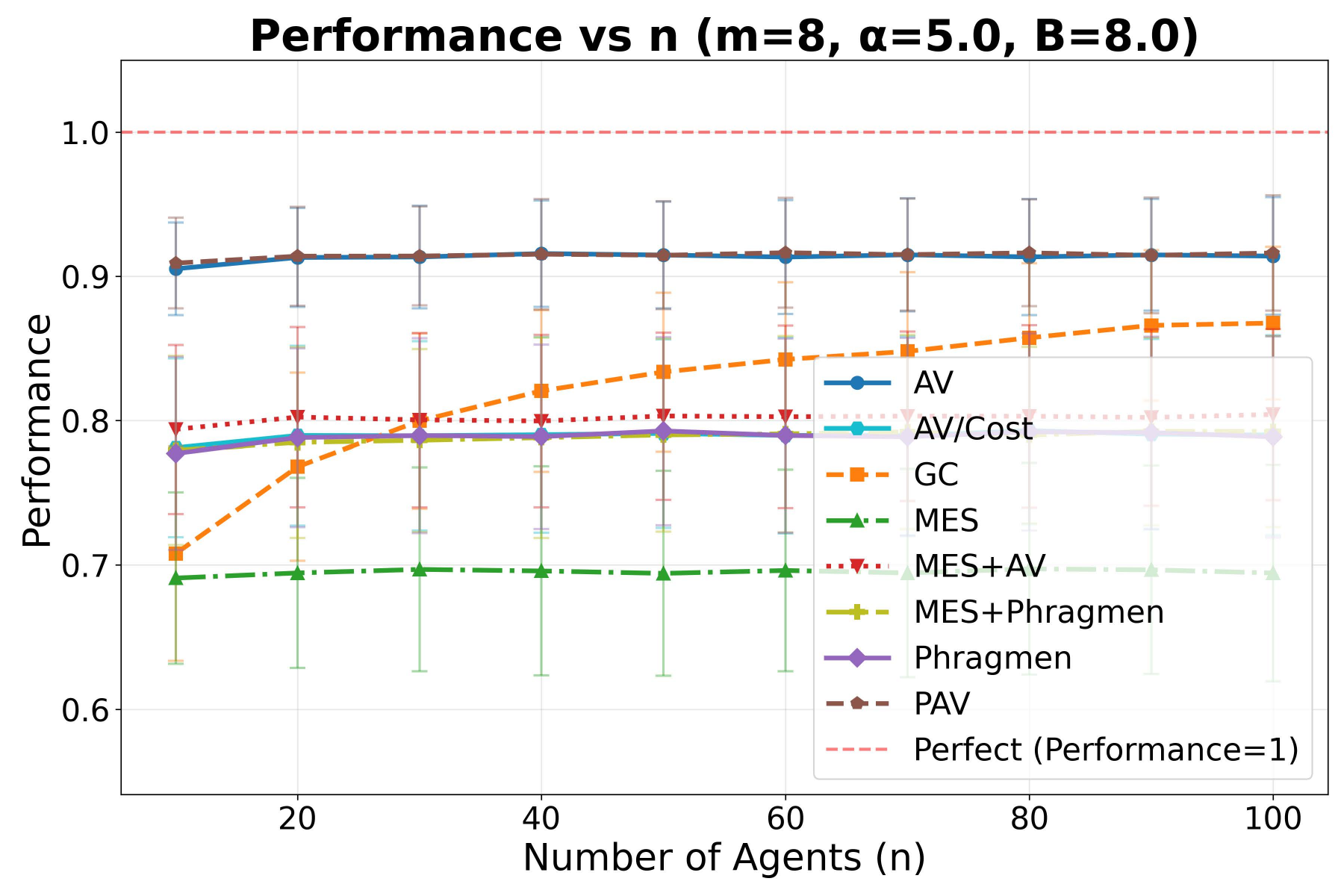}
\caption{$\alpha = 5$}
\end{subfigure}
\caption{Performance on the number of agents, cost-proportional utility.\label{fig:apx_performance_c}}
\end{figure}

\begin{figure}[b]
    \centering
    \includegraphics[width=0.6\linewidth]{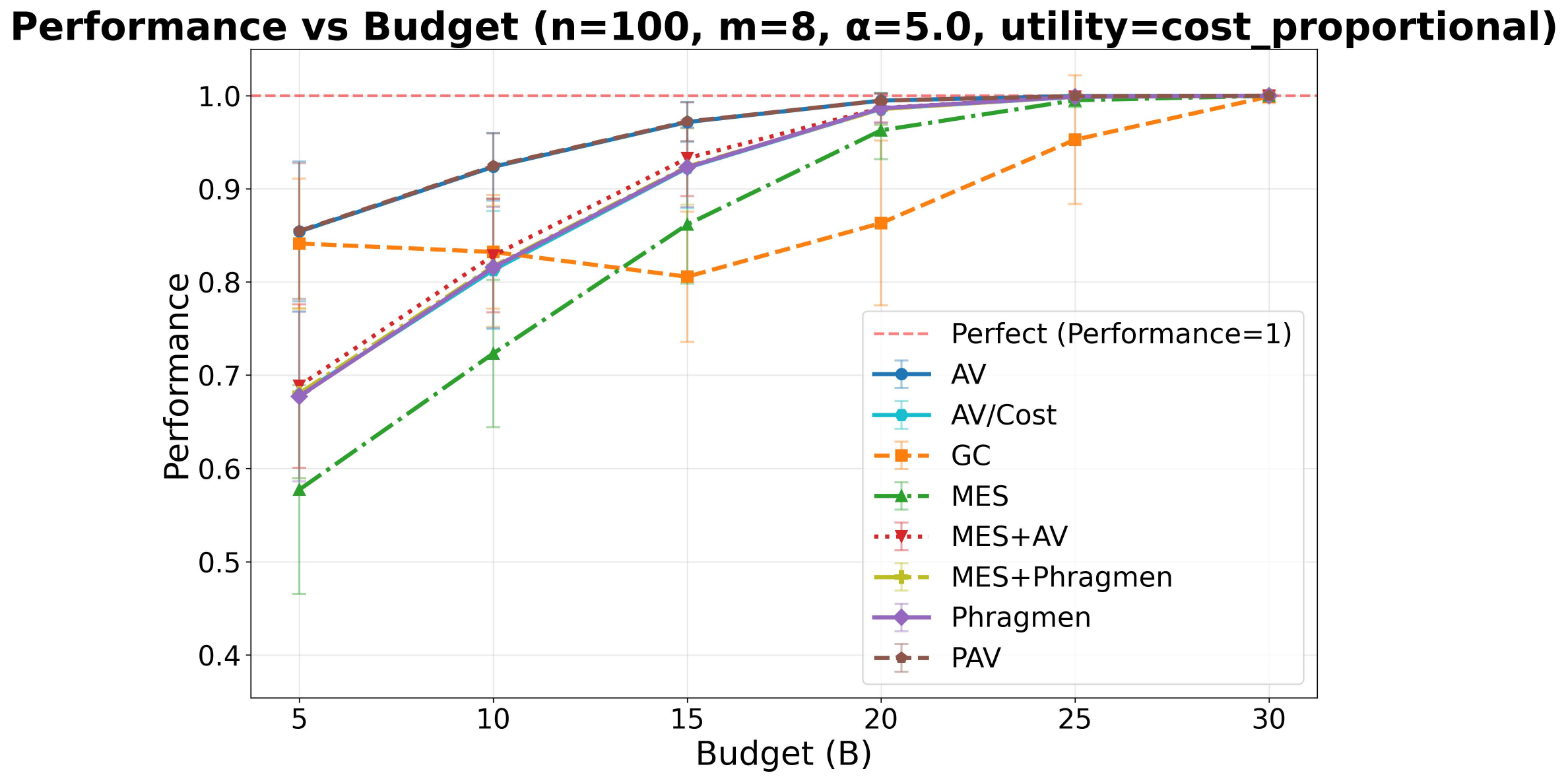}
    \caption{Performance under different budgets, cost-proportional utility \label{fig:budget_c}} 
   
\end{figure}

\begin{figure}[b]
    \centering
    \includegraphics[width=0.6\linewidth]{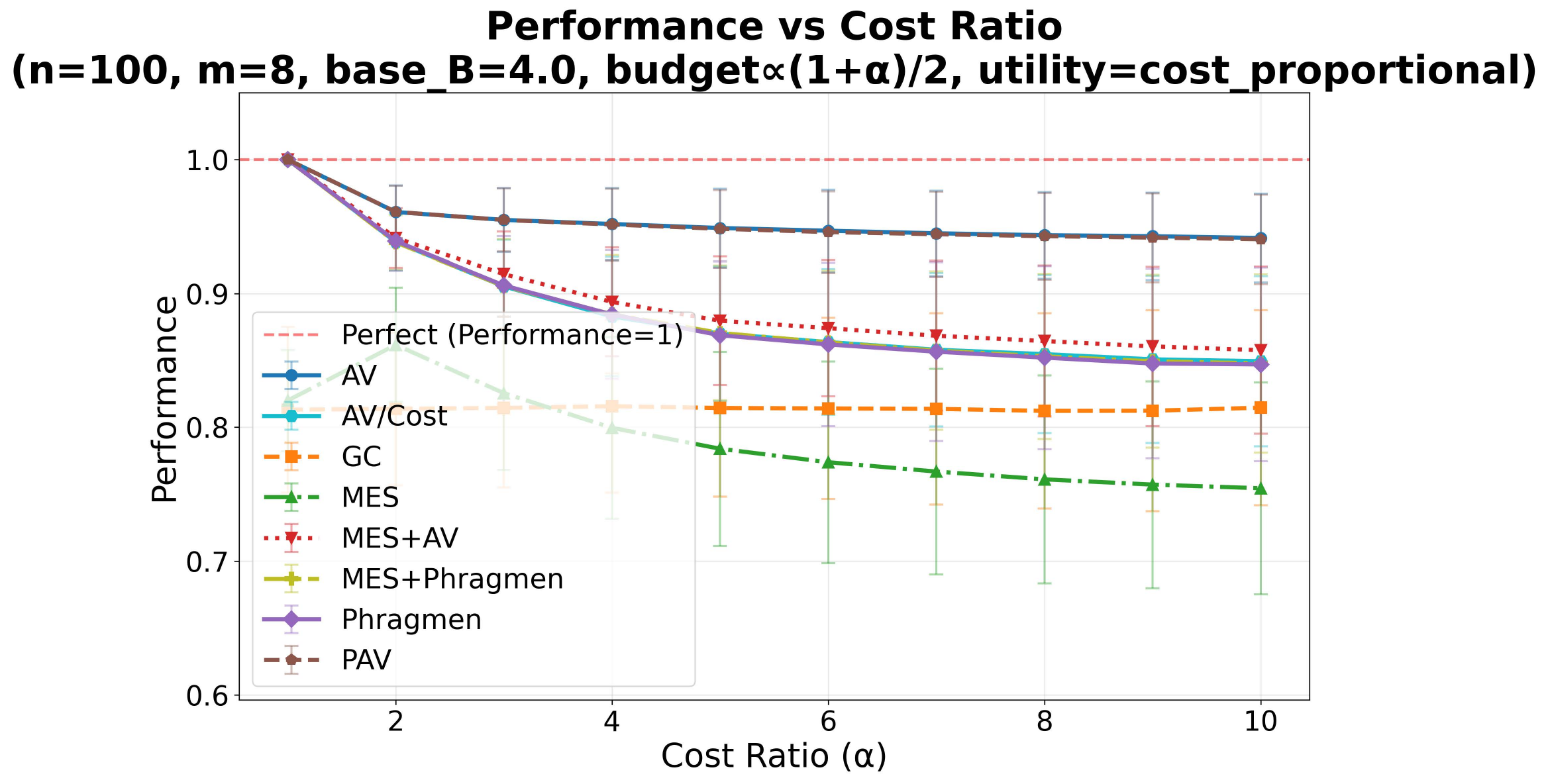}
    \caption{performance when $B$ and total cost are correlated, cost-proportional utility \label{fig:alpha_ratio_c}}
    
\end{figure}

\begin{figure}[b]
    \centering
    \includegraphics[width=0.8\linewidth]{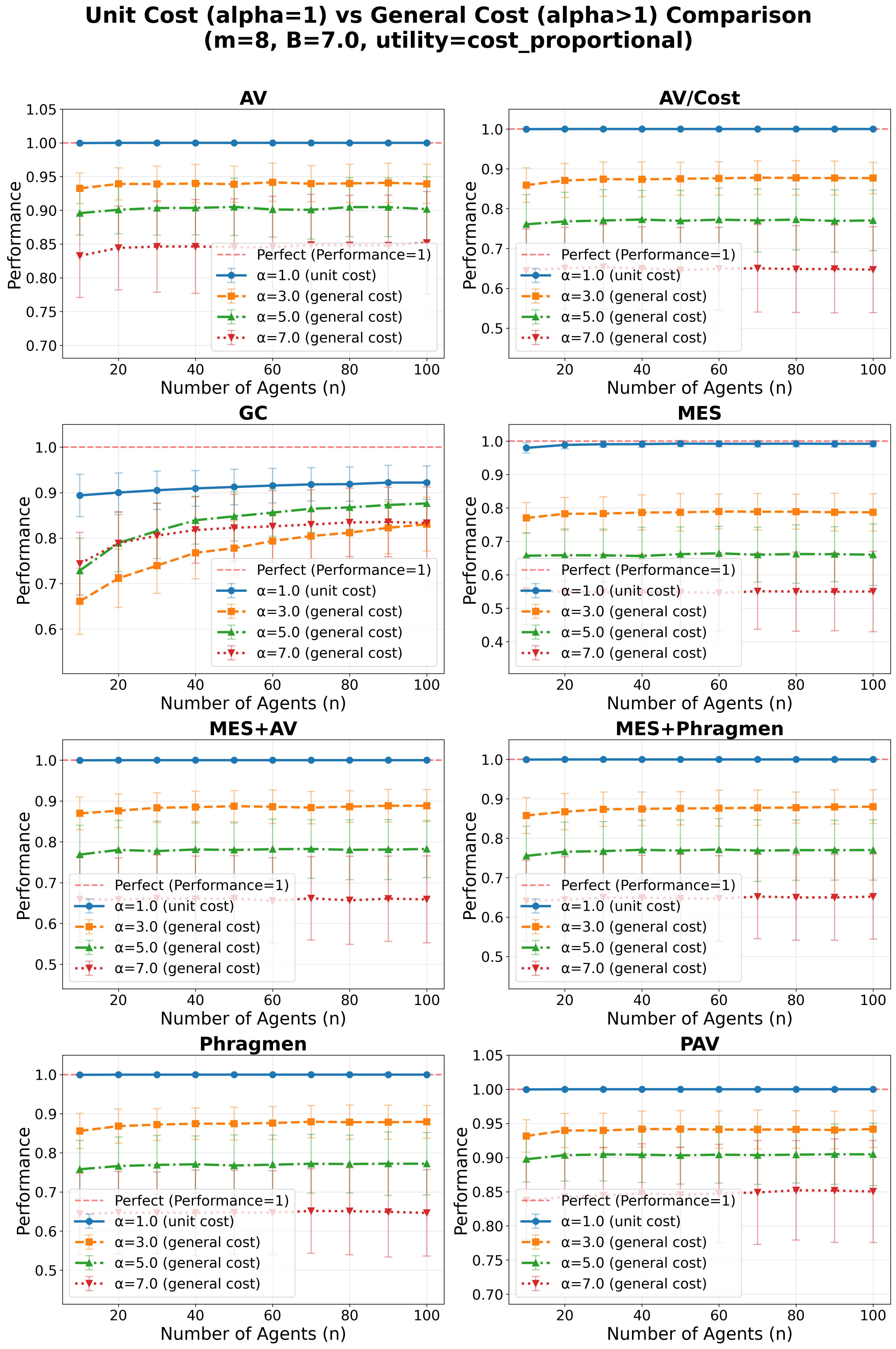}
    \caption{Performance on different cost ratios, cost-proportional utility\label{fig:apx_alpha_c}}
    
\end{figure}

\begin{figure}
    \centering
    \includegraphics[width=0.8\linewidth]{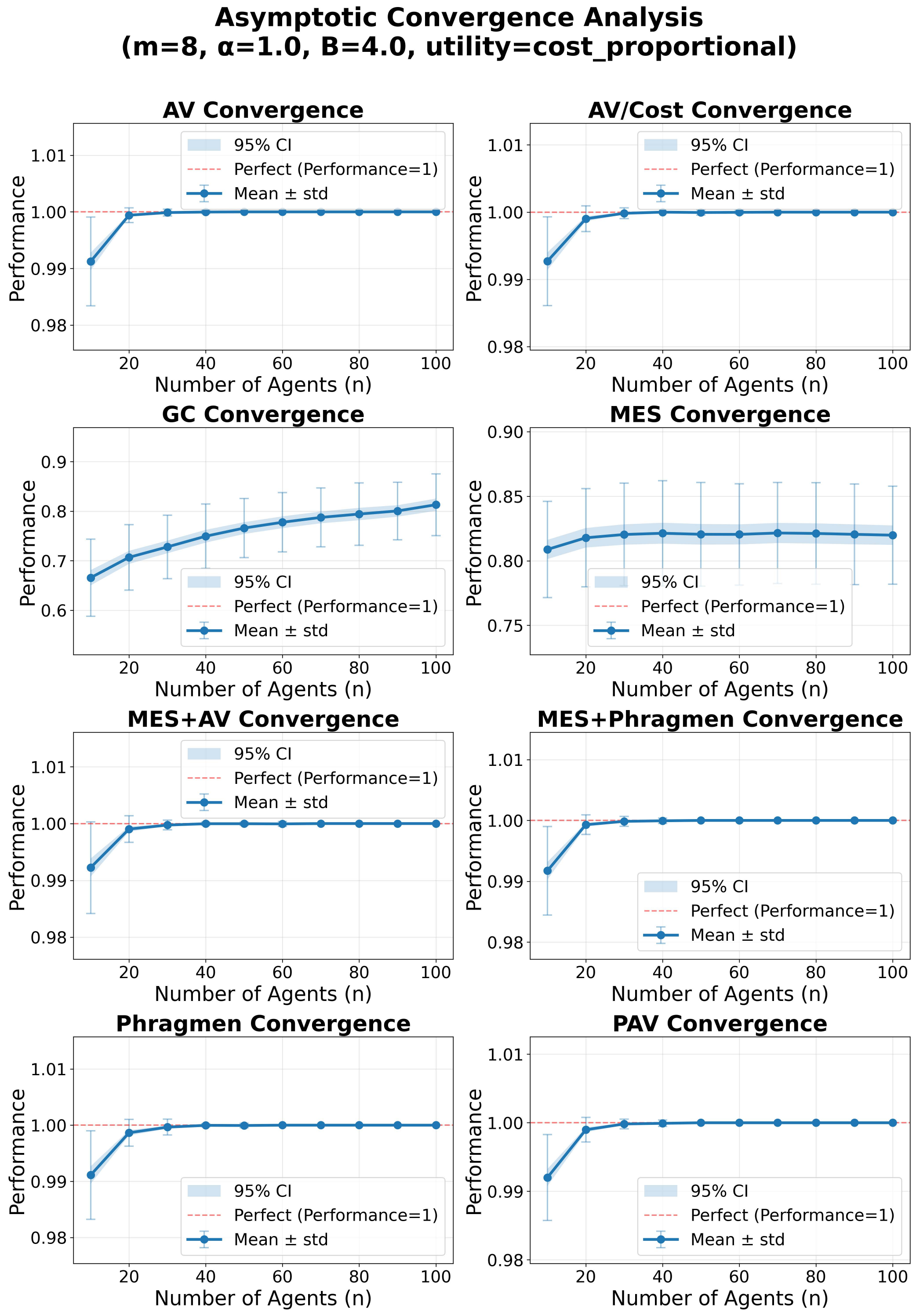}
    \caption{Convergence of rules, cost-proportional utility}
    \label{fig:convergence_c}
\end{figure}

\begin{figure}
    \centering
    \includegraphics[width=0.99\linewidth]{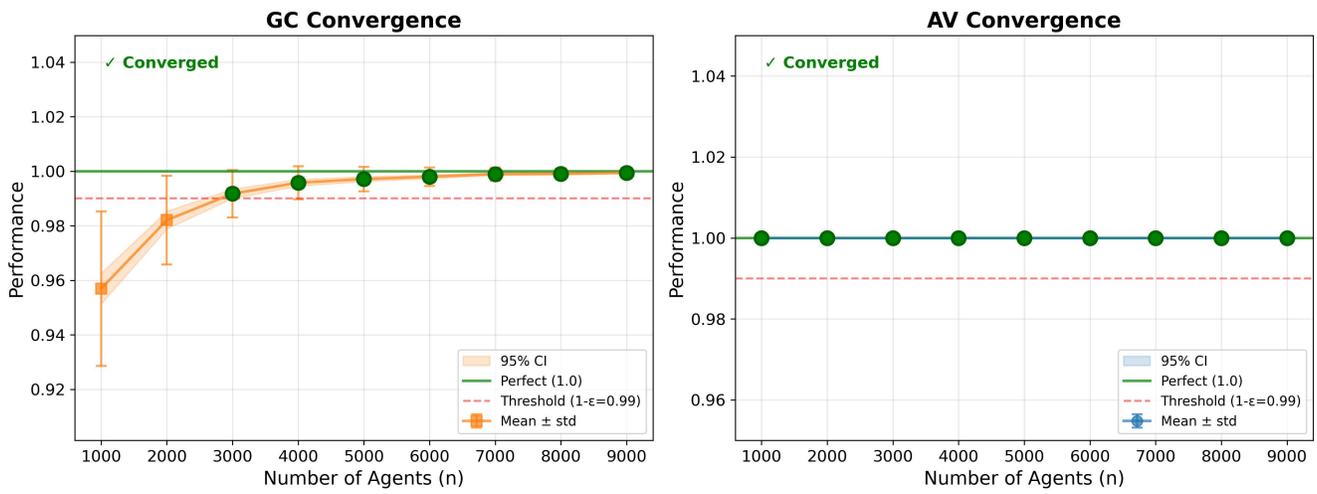}
    \caption{Convergence on $\gc$ in large $n$ (in comparison with $\av$), cost-proportional utility\label{fig:gc_c}}. 
    
\end{figure}

\end{document}